\documentclass[11pt, oneside]{article}
\usepackage{geometry}
\usepackage{graphicx}				% Use pdf, png, jpg, or eps§ with pdflatex; use eps in DVI mode
\usepackage{setspace}							% TeX will automatically convert eps --> pdf in pdflatex		
\usepackage{amssymb}
\usepackage{amsmath}
\usepackage{amsfonts}
\usepackage{float}
\usepackage{theoremref}
\usepackage{hyperref}
\usepackage{appendix}
\usepackage[colorinlistoftodos]{todonotes}
\usepackage{cleveref}
\usepackage[utf8]{inputenc}
\usepackage{mathtools}
\usepackage{abstract}
\usepackage{tikz}
\usepackage[ruled,vlined]{algorithm2e}
\usepackage{enumitem}
\usepackage{comment}
\usepackage{authblk}
\usepackage{amsthm}
\usepackage{etoc}
\usepackage{calc}

\newtheorem{theorem}{Theorem}

\theoremstyle{plain}
\newtheorem{proposition}{Proposition}

\newtheorem{lemma}{Lemma}

\theoremstyle{plain}
\newtheorem{assumption}{Assumption}

\theoremstyle{plain}
\newtheorem{definition}{Definition}

\theoremstyle{remark}

\theoremstyle{definition}

\usetikzlibrary{bayesnet}
\usetikzlibrary{fit}
\DeclareMathOperator*{\argmin}{arg\,min}

\DeclareMathOperator{\Cov}{Cov}
\DeclareMathOperator{\Tr}{Tr}

\DeclarePairedDelimiter\abs{\lvert}{\rvert}%
\DeclarePairedDelimiter\norm{\lVert}{\rVert}%

\renewcommand{\leq}{\leqslant}
\renewcommand{\geq}{\geqslant}

\newcommand{\approxP}{\mathrel{\stackrel{{\rm P}}{\mathrel{\scalebox{1.8}[1]{$\simeq$}}}}}

\newcommand{\R}{\mathbb{R}}
\newcommand{\E}{\mathbb{E}}
\newcommand{\N}{\mathbb{N}}

\newcommand{\bG}{\mathbf{G}}
\newcommand{\bA}{\mathbf{A}}
\newcommand{\bQ}{\mathbf{Q}}
\newcommand{\bX}{\mathbf{X}}
\newcommand{\cS}{\mathcal{S}}

\newcommand{\raphael}[1]{{\color{blue} R: #1}}
\newcommand{\cedric}[1]{{\color{red} CEDRIC: #1}}
\newcommand{\bm}{{\mathbf{m}}}
\newcommand{\bM}{{\mathbf{M}}}
\newcommand{\bx}{{\mathbf{x}}}
\newcommand{\bP}{{\mathbf{P}}}
\newcommand{\bH}{{\mathbf{H}}}
\newcommand{\bh}{{\mathbf{h}}}
\newcommand{\be}{{\mathbf{e}}}
\newcommand{\bZ}{{\mathbf{Z}}}
\newcommand{\bz}{{\mathbf{z}}}
\newcommand{\bb}{{\mathbf{b}}}

\newcommand{\bS}{{\mathbf{S}}}
\newcommand{\bY}{{\mathbf{Y}}}
\newcommand{\bU}{{\mathbf{U}}}
\newcommand{\bV}{{\mathbf{V}}}
\newcommand{\bu}{{\mathbf{u}}}
\newcommand{\bv}{{\mathbf{v}}}
\begin{document}
\etocdepthtag.toc{mtmain}
\etocsettagdepth{mtmain}{subsection}
\etocsettagdepth{mtappendix}{none}
\title{Graph-based Approximate Message Passing Iterations}
\date{}

\author[1]{C\'edric Gerbelot\thanks{cedric.gerbelot@ens.fr}}
\author[2]{Raphaël Berthier\thanks{raphael.berthier@inria.fr}}
\affil[1]{Laboratoire de Physique de l’Ecole Normale Supérieure, Université
  PSL, CNRS, Paris, France}
  \affil[2]{Inria - Département d'informatique de l'ENS,
  Université PSL, Paris, France}
\maketitle 

\begin{abstract}%
Approximate-message passing (AMP) algorithms have become an important element of high-dimensional statistical inference, mostly due to their adaptability and concentration properties, the state evolution (SE) equations. This is demonstrated by the growing number of new iterations proposed for increasingly complex problems, ranging from multi-layer inference to low-rank matrix estimation with elaborate priors. In this paper, we address the following questions: is there a structure underlying all AMP iterations that unifies them in a common framework? Can we use such a structure to give a modular proof of state evolution equations, adaptable to new AMP iterations without reproducing each time the full argument ? We propose an answer to both questions, showing that AMP instances can be generically indexed by an oriented graph. This enables to give a unified interpretation of these iterations, independent from the problem they solve, and a way of composing them arbitrarily. We then show that all AMP iterations indexed by such a graph admit rigorous SE equations, extending the reach of previous proofs, and proving a number of recent heuristic derivations of those equations. Our proof naturally includes non-separable functions and we show how existing refinements, such as spatial coupling or matrix-valued variables, can be combined with our framework.
\end{abstract}

\tableofcontents

\section{Introduction}

Approximate Message Passing (AMP) algorithms are iterative equations solving inference problems involving high-dimensional random variables with random interactions \cite{donoho2009message,zdeborova2016statistical}. For the typical case in which AMP iterations were initially studied, the interactions involve an i.i.d. Gaussian matrix. These algorithms are inspired from Bolthausen's iterative solution of the celebrated Thouless-Anderson-Palmer (TAP) equations of spin glass theory \cite{mezard1987spin,bolthausen2014iterative,bolthausentap}. However, they are usually derived as heuristic relaxations of the belief propagation equations \cite{pearl2014probabilistic} on dense factor graphs in a manner often encountered in the context of statistical physics of disordered systems. A central property of AMP iterations is that the distribution of their outputs can be tracked rigorously in the high-dimensional limit by low-dimensional equations called \emph{state evolution} (SE). This property can be seen as similar to the concept of density evolution from coding theory \cite{richardson2008modern}, but in the case of dense factor graphs.
\paragraph{}
In recent years, the growing interest in high-dimensional inference and learning problems has motivated the introduction of approximate-message passing algorithms as solutions to many inference problems, and as analytical tools---thanks to the SE equations---to study the statistical properties of learned estimators, notably starting with the LASSO \cite{bayati2011dynamics,krzakala2012probabilistic,donoho2009message}. A number of extensions were then proposed for inference problems of growing complexity: generalized linear modelling and robust m-estimators \cite{rangan2011generalized,donoho2016high,zdeborova2016statistical}, low-rank matrix reconstruction \cite{rangan2012iterative,lesieur2017constrained}, principal component analysis (PCA) \cite{deshpande2014information,lesieur2015phase}, inference in deep multilayer networks with random weights \cite{manoel2017multi}, matrix-valued inference problems \cite{aubin2019committee} or matrix recovery under generative priors \cite{aubin2020spiked}, among others. Interestingly, AMP algorithms can be composed with one another to solve inference problems obtained by combining factor graphs, as demonstrated in \cite{aubin2020spiked}, where each part of the factor graph represents an elaborate prior and inference process. This demonstrates the adaptability of such iterations, even more so as the state evolution equations are shown to hold, often heuristically, for these composite structures. 

\paragraph{Contributions.}As the diversity of inference problems and AMP iterations increases, it is important to identify a common structure underlying the known AMP algorithms. Such a partial unification was done in \cite{javanmard2013state,berthier2020state}: symmetric and asymmetric AMP iterations are treated in a common framework. However, these results do not apply to the more recent AMP iterations designed for more complex problems presenting multilayered structures or ones obtained by combining factor graphs. 

In this paper, our first contribution is to show how AMP algorithms are naturally indexed by a graph that determines its form. Seeing AMP algorithms as supported by this graph helps understanding the iterations, especially the multi-layer ones, in a unified way. In this regard, we hope that our framework will be used as a tool to generate new AMP iterations. Roughly speaking, the graph underlying the AMP iteration represents the interaction of the high-dimensional variables of the associated inference problem. However, this graph is not the factor graph representing the inference problem that sometimes appears in the derivation of AMP equations, see \cite{krzakala2012probabilistic} for example.  The factor graph is microscopic, in the sense that it disappears when taking the dense limit leading to the AMP equations. On the contrary, the graph that we consider here is macroscopic: it structures the AMP iteration itself. It is insensitive to the underlying inference problem that has generated the AMP equation; for instance, it can be used in both Bayes optimal or non-Bayes optimal scenarios. 

The second contribution of this paper is to use the graph framework to show that all graph-based AMP iterations admit a rigorous SE description. This generalizes the previous works of \cite{bayati2011dynamics,javanmard2013state,berthier2020state} on SE to more complex iterations. Using our result, writing and proving the state evolution equations is reduced to the identification of a specific structure in the AMP iteration, instead of heuristically deriving or reproducing the rigorous proof entirely for problems of increasing complexity. In particular, it gives a theoretical grounding for the analysis of AMP on recent multi-layer structures \cite{manoel2017multi,aubin2019committee,aubin2020spiked}. Related to \cite{manoel2017multi}, this paper proves that AMP algorithms are a rigorously grounded approach to understanding multi-layer neural networks, albeit only when the weights are random and when we perform inference with an AMP algorithm. Still, in a context where theory struggles to explain the behavior of multi-layered neural networks, it is interesting to see that this particular case can be rigorously studied, even for deep architectures. 

We illustrate the flexibility of our framework by applying it to diverse inference problems mentioned above, notably multilayer generalized linear estimation problems and low-rank matrix recovery with deep generative priors. We also show how our results can be extended to handle matrix-valued variables and combined with the spatial coupling framework introduced in \cite{krzakala2012probabilistic,javanmard2013state}.

\medskip
\paragraph{Related work.} There is a rich literature of proofs of state evolution equations, notably starting with Bolthausen's iterative scheme \cite{bolthausen2014iterative,bolthausentap} based on Gaussian conditioning. The technique was then adapted and extended to the case of a more generic AMP iteration related to the LASSO problem in \cite{bayati2011dynamics}, where it is mentioned that Gaussian conditioning methods also appear in \cite{donoho2006most} to tackle fundamental random convex geometry problems. The analysis was then extended to matrix-valued variables with block-separable non-linearities in \cite{javanmard2013state} and for vector-valued variables with non-separable non-linearities in \cite{berthier2020state}, which also show that symmetric AMP and asymmetric AMP can be treated in the same framework. Our proof is partly based on the same iterative Gaussian conditioning method but is additionally combined with an embedding specific to the graph framework. To the best of our knowledge, the latter part of the proof is novel. 

Another line of work---called VAMP (vector approximate message passing) algorithms---handles rotationnally invariant matrices \cite{rangan2019vector} with generic spectrum. This family of VAMP iterations is obtained using a Gaussian parametrization of \emph{expectation propagation} \cite{minka2013expectation,opper2005expectation}, a variational inference algorithm based on iterative moment-matching between a chosen form of probability distribution (e.g., Gaussian nodes on a factor graph) and a target distribution observed through empirical data.
These iterations also verify SE equations proven with a similar conditioning method \cite{takeuchi2017rigorous,rangan2019vector}, handling a different kind of randomness than i.i.d.~Gaussian matrices. The SE proof for VAMP iterations was then extended to multilayer inference problems and their matrix-valued counterparts in \cite{fletcher2018inference,pandit2020inference}. In these works, the conditioning method is applied in a sequential manner to each layer of the problem, making it specific to multilayer inference problems. On the contrary, our proof method is not restricted to sequential multilayer estimation as mentioned in the contributions, and does not rely on iterating through the graph. However, our proof does not apply to all rotationnally invariant matrices. We handle mostly Gaussian or GOE matrices, with extensions to correlated Gaussian matrices, products of Gaussian matrices and spatially coupled Gaussian matrices. This is discussed in greater detail in Sections \ref{sec:se-graph} and  \ref{sec:applications}.
\medskip
\paragraph{Outline of the paper.} 
The paper is organised as follows: we start by presenting the indexation of AMP iterations by an oriented graph in Section \ref{sec:graph-based-amp}. Several conceptual examples are provided. We present the state evolution equations on any graph-supported AMP iteration in Section \ref{sec:se-graph}, along with its proof, which constitutes the main technical contribution of this paper. We then move to applications to inference problems in Section \ref{sec:applications} and conclude on related open problems in Section~\ref{sec:perspectives}. All proofs of auxiliary results are deferred to the Appendix. 
\paragraph{Notations.}
We adopt similar notations to those of \cite{berthier2020state}. Differences are mainly due to the matrix variables framework.

We denote scalars with lowercase letters, vectors with bold lowercase letters and matrices with bold uppercase ones. Inner products are denoted by brackets $\langle . , . \rangle$, and the canonical inner products are chosen for vectors and matrices, i.e., $\langle \mathbf{x}, \mathbf{y} \rangle = \mathbf{x}^{\top}\mathbf{y}$, $\langle \mathbf{X}, \mathbf{Y} \rangle = \mbox{Tr} \left(\mathbf{X}^{\top}\mathbf{Y}\right)$. The associated norms are respectively denoted $\norm{.}_{2}$ and $\norm{.}_{F}$ for the Frobenius norm. 

For two random variables $X$ and $Y$, and a $\sigma$-algebra $\mathfrak{S}$, we use $X\vert_{\mathfrak{S}} \stackrel{d} = Y$ to mean that for any integrable function $\phi$ and any $\mathfrak{S}$-measurable bounded random variable $Z$, $\mathbb{E}\left[\phi(X)Z\right] = \mathbb{E}\left[\phi(Y)Z\right]$. For two sequences of random variables $X_{n},Y_{n}$, we write $X_{n} \stackrel{P}\simeq Y_{n}$ when their difference converges in probability to $0$, i.e., $X_{n}-Y_{n} \xrightarrow[]{P}0$.

 We use $\mathbf{I}_{N}$ to denote the $N \times N$ identity matrix, and $0_{N \times N}$ the $N \times N$ matrix with zero entries. We use $\sigma_{min}(\mathbf{Q})$ and $\sigma_{max}(\mathbf{Q}) = \norm{\mathbf{Q}}_{op}$ to denote the minimum and maximum singular values of a given matrix $\mathbf{Q}$. For two matrices $\mathbf{Q}$ and $\mathbf{P}$ with the same number of rows, we denote their horizontal concatenation with $\left[\mathbf{P} \vert \mathbf{Q}\right]$. The orthogonal projector onto the range of a given matrix $\mathbf{M}$ is denoted $\mathbf{P}_{\mathbf{M}}$, and let $\mathbf{P}_{\mathbf{M}}^{\perp} = \mathbf{I}-\mathbf{P}_{\mathbf{M}}$. 
 
 Let $\cS_q^+$ denote the space of positive semi-definite matrices of size $q \times q$.
For any matrix $\boldsymbol{\kappa} \in \cS_{q}^{+}$ and a random matrix $\bZ\in \mathbb{R}^{N \times q}$ we write $\bZ \sim \mathbf{N}(0,\boldsymbol{\kappa}\otimes \mathbf{I}_{N})$ if $\bZ$ is a matrix with jointly Gaussian entries such that for any $1\leqslant i,j\leqslant q$, $\mathbb{E}[\bZ^{i}(\bZ^{j})^{\top}] = \boldsymbol{\kappa}_{i,j}\mathbf{I}_{N}$, where $\bZ^{i},\bZ^{j}$ denote the i-th and j-th columns of $\bZ$. The i-th line of the matrix $\bZ$ is denoted $\bZ_{i}$. %In the following statements, we will often describe the joint distributions of given columns, of given matrix variables: if $s, s' \in \{1, \dots, t \}$, $i,i' \in \{1, \dots, N \}$ and $p,p' \in \{1, \dots, q \}$, the covariance between the centered Gaussian random vectors $\bZ^s_{i,p}, \bZ^{s'}_{i',p'}$ is $\Cov(\bZ^s_{i,p}, \bZ^{s'}_{i',p'}) = \delta_{i,i'} \boldsymbol{\kappa}^{s,s'}_{p,p'}$. Note that this is equivalent to saying that the matrix $[\bZ^{1}|...|\bZ^{t}]$ has i.i.d.~rows pulled from the correlated Gaussian distribution with the $tq\times tq$ covariance matrix $\begin{pmatrix} \boldsymbol{\kappa}^{1,1} & \cdots & \boldsymbol{\kappa}^{1,t} \\
%\vdots & \ddots & \vdots \\
%\boldsymbol{\kappa}^{t,1} & \cdots & \boldsymbol{\kappa}^{t,t}
%\end{pmatrix}$.

If $f: \R^{N \times q} \to \R^{N \times q}$ is an function and $i \in \{1, \dots N\}$, we write $f_{i}:\mathbb{R}^{N \times q} \to \mathbb{R}^{q}$  the component of $f$ generating the $i$-th line of its image, i.e., if $\bX \in \mathbb{R}^{N \times q}$, 
\begin{equation*}
    f(\bX) = \begin{bmatrix}
    f_{1}(\bX) \\
    \vdots \\
    f_{N}(\bX)\end{bmatrix} \in \mathbb{R}^{N \times q} \, .
\end{equation*}
We write $\frac{\partial f_{i}}{\partial \bX_{i}}$ the $q\times q$ Jacobian containing the derivatives of $f_{i}$ with respect to (w.r.t.) the $i$-th line $\bX_{i}\in \mathbb{R}^{q}$:
\begin{equation}
\label{eq:Ons_jacob}
    \frac{\partial f_{i}}{\partial \bX_{i}} = \begin{bmatrix}\frac{\partial (f_{i}(\bX))_{1}}{\partial \bX_{i1}} & \dots & \frac{\partial (f_{i}(\bX))_{1}}{\partial \bX_{iq}} \\
    \vdots& &\vdots \\
    \frac{\partial (f_{i}(\bX))_{q}}{\partial \bX_{i1}} & \dots & \frac{\partial (f_{i}(\bX))_{q}}{\partial \bX_{iq}}
    \end{bmatrix} \in \mathbb{R}^{q \times q} \, .
\end{equation}

\section{Graph-based AMP iterations}
\label{sec:graph-based-amp}

We start by defining the class of graphs indexing AMP iterations. 

\begin{definition}[graph notions]
A \emph{finite directed graph}---also simply called \emph{graph} in the following---is a pair $G = (V,\overrightarrow{E})$ where $V$ is a finite set, called the \emph{vertex set}, and $\overrightarrow{E}$ is a subset of $V \times V$, called the \emph{edge set}. This definition of graphs uses directed edges and allows loops. 

A graph $G = (V,\overrightarrow{E})$ is said to be \emph{symmetric} if for all $v, w \in \overrightarrow{E}$, $(v,w) \in \overrightarrow{E}$ if and only if $(w,v) \in \overrightarrow{E}$. 

The \emph{degree} $\deg v$ of a node $v \in V$ is the number of edges of which it is the end-node. In symmetric graphs, it is also the number of edges of which $v$ is the starting-node.  
\end{definition}

\paragraph{Graph notations.} Given a symmetric graph $G = (V,\overrightarrow{E})$, the following notations are useful. We sometimes write $v \rightarrow w$ to mean that $\overrightarrow{e} = (v,w)$ is an edge of the graph. We say that $v$ is the starting-node of $\overrightarrow{e}$ and $w$ the end-node of $\overrightarrow{e}$. We denote $\overleftarrow{e} = (w,v) \in \overrightarrow{E}$ the symmetric edge of~$\overrightarrow{e}$. If $\overrightarrow{e}$ is a loop, then $\overleftarrow{e} = \overrightarrow{e}$. We write $\overrightarrow{e} \to \overrightarrow{e}'$ as a shorthand to say that the end-node of $\overrightarrow{e} \in \overrightarrow{E}$ is the starting-node of $\overrightarrow{e}' \in \overrightarrow{E}$. Note that for any $\overrightarrow{e} \in \overrightarrow{E}$, $\overleftarrow{e} \rightarrow \overrightarrow{e}$.

\paragraph{Iteration.} We now fix a symmetric finite directed graph $G = (V,\overrightarrow{E})$. We associate an AMP iteration supported by the graph $G$ as follows. 
\begin{itemize}
    \item The variables $\bx^t_{\overrightarrow{e}}$ of the AMP iteration are indexed by the iteration number $t \in \N$ and the oriented edges of the graph $\overrightarrow{e} \in \overrightarrow{E}$. 
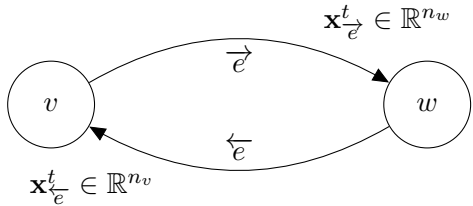
    \begin{center}
	\begin{tikzpicture}[scale = 1]
    \tikzstyle{point}=[draw,circle,minimum width=3em];
    \tikzstyle{fleche}=[->];
    \node[point] (v) at (0,0) {$v$};
    \node[point] (w) at (5,0) {$w$};
    \draw[fleche] (v) to[bend left]
    node[below,midway]{$\overrightarrow{e}$}
    node[above right, near end]{$\mathbf{x}^t_{\overrightarrow{e}} \in \R^{n_w}$}
    (w);
    \draw[fleche] (w) to[bend left]
    node[above,midway]{$\overleftarrow{e}$}
    node[below left, near end]{$\mathbf{x}^t_{\overleftarrow{e}} \in \R^{n_v}$}
    (v);
	\end{tikzpicture}
\end{center}
    \item All variables associated to edges $\overrightarrow{e} = (v,w)$ with end-node $w \in V$ have a same dimension $n_w \in \N_{>0}$, i.e., $\mathbf{x}^t_{\overrightarrow{e}} \in \R^{n_w}$. We define $N = \sum_{(v,w) \in \overrightarrow{E}} n_w$ the sum of the dimensions of all variables.
    \item Matrices of the AMP iteration are also indexed by the edges of the graph. If $\overrightarrow{e} = (v,w) \in \overrightarrow{E}$, $\bA_{\overrightarrow{e}} \in \R^{n_w \times n_v}$. These matrices must satisfy the symmetry condition $\bA_{(v,w)} = \bA_{(w,v)}^\top$. In particular, this implies that matrices $\bA_{(v,v)} \in \R^{n_v \times n_v}$ associated to loops $(v,v)\in \overrightarrow{E}$ must be symmetric.
    
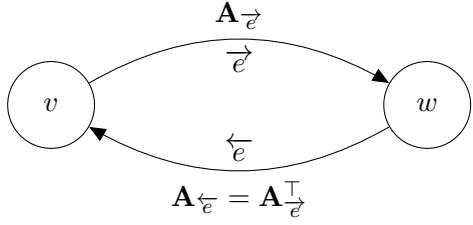
\begin{center}
	\begin{tikzpicture}[scale = 1]
    \tikzstyle{point}=[draw,circle,minimum width=3em];
    \tikzstyle{fleche}=[->];
    \node[point] (v) at (0,0) {$v$};
    \node[point] (w) at (5,0) {$w$};
    \draw[fleche] (v) to[bend left]
    node[above,midway]{$\mathbf{A}_{\overrightarrow{e}}$} node[below,midway]{$\overrightarrow{e}$}
    (w);
    \draw[fleche] (w) to[bend left]
    node[below,midway]{$\mathbf{A}_{\overleftarrow{e}} = \mathbf{A}_{\overrightarrow{e}}^\top$} node[above,midway]{$\overleftarrow{e}$} 
    (v);
	\end{tikzpicture}
\end{center}
    \item Non-linearities of the AMP iteration are also indexed by the edges of the graph (and possibly by the iteration number $t$). If $t \geq 0$ and $\overrightarrow{e} = (v,w) \in \overrightarrow{E}$, $f^t_{(v,w)}\left(\left(\mathbf{x}^t_{\overrightarrow{e}'}\right)_{\overrightarrow{e}':\overrightarrow{e}' \to \overrightarrow{e}}\right)$ is a function of all the variables of the edges whose end-node is the starting-node $v$ of $\overrightarrow{e}$, as denoted by the condition $\overrightarrow{e}' \to \overrightarrow{e}$. It is a function from $(\R^{n_v})^{\deg v}$ to $\R^{n_v}$. 
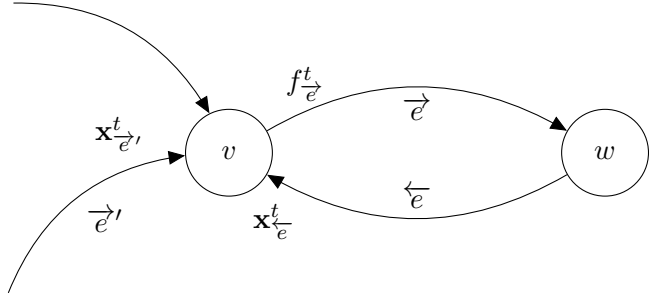
    \begin{center}
	\begin{tikzpicture}[scale = 1]
    \tikzstyle{point}=[draw,circle,minimum width=3em];
    \tikzstyle{fleche}=[->];
    \node[point] (v) at (0,0) {$v$};
    \node[point] (w) at (5,0) {$w$};
    \node[] (v1) at (-3,2) {};
    \node[] (v2) at (-3,-2) {};
    \draw[fleche] (v) to[bend left]
    node[below,midway]{$\overrightarrow{e}$}
    node[above,very near start]{$f^t_{\overrightarrow{e}}$}
    (w);
    \draw[fleche] (w) to[bend left]
    node[above,midway]{$\overleftarrow{e}$} 
    node[below left, very near end]{$\mathbf{x}^t_{\overleftarrow{e}}$}
    (v);
    \draw[fleche] (v1) to[bend left]
    (v);
    \draw[fleche] (v2) to[bend left]
    node[below right,midway]{$\overrightarrow{e}'$}
    node[above left, very near end]{$\mathbf{x}^t_{\overrightarrow{e}'}$}
    (v);
	\end{tikzpicture}
\end{center}
\end{itemize}
Once these parameters $\left(\mathbf{A}_{\overrightarrow{e}}\right)_{\overrightarrow{e} \in \overrightarrow{E}}$ and $\left(f^t_{\overrightarrow{e}}\right)_{t \geq 0, \overrightarrow{e} \in \overrightarrow{E}}$ are given, we can choose an arbitrary initial condition $\mathbf{x}^0_{\overrightarrow{e}} \in \R^{n_w}$ for all oriented edges ${\overrightarrow{e}} \in \overrightarrow{E}$ of the graph. We define recursively the AMP iterates $\left(\mathbf{x}^t_{\overrightarrow{e}}\right)_{t \geq 0, \overrightarrow{e} \in \overrightarrow{E}}$, by the iteration: for all $t \geq 0, \overrightarrow{e} \in \overrightarrow{E}$, 
\begin{align}
    \mathbf{x}^{t+1}_{\overrightarrow{e}} &= \bA_{\overrightarrow{e}} \mathbf{m}^t_{\overrightarrow{e}} - b^t_{\overrightarrow{e}} \mathbf{m}^{t-1}_{\overleftarrow{e}} \, , \label{eq:graph-amp-1}\\
    \mathbf{m}^t_{\overrightarrow{e}} &= f^t_{\overrightarrow{e}}\left(\left(\mathbf{x}^t_{\overrightarrow{e}'}\right)_{\overrightarrow{e}':\overrightarrow{e}' \to \overrightarrow{e}}\right) \, , \label{eq:graph-amp-2}
\end{align}
where $b^t_{\overrightarrow{e}}$ is the so-called \emph{Onsager term} 
\begin{equation}
    b^t_{\overrightarrow{e}} = \frac{1}{N} \Tr \frac{\partial f^t_{\overrightarrow{e}}}{\partial \mathbf{x}_{\overleftarrow{e}}} \left(\left(\mathbf{x}^t_{\overrightarrow{e}'}\right)_{\overrightarrow{e}':\overrightarrow{e}' \to \overrightarrow{e}}\right) \qquad \in \R \, . \label{eq:graph-amp-3}
\end{equation}
The above partial derivative makes sense as $\overleftarrow{e} \to \overrightarrow{e}$, thus $\bx_{\overleftarrow{e}}$ is a variable of $f^t_{\overrightarrow{e}}$. Note that in \eqref{eq:graph-amp-1}, the Onsager term multiplies the vector $\bm^{t-1}_{\overleftarrow{e}}$ indexed by the symmetric edge ${\overleftarrow{e}}$ of~${\overrightarrow{e}}$.
\paragraph{}
Let us derive some simple particular cases of this framework, first to recover the classical asymmetric and symmetric AMP iterations, and second to cover multi-layer AMP iterations. 

\paragraph{Asymmetric AMP.} The asymmetric AMP iteration appeared first in the literature to solve the compressed sensing problem \cite{donoho2009message} and then more generally to tackle generalized linear estimation, see, e.g., \cite{rangan2011generalized,schniter2014compressive,donoho2016high}. It corresponds to a simple underlying graph composed of two nodes and two symmetric directed edges between them. 
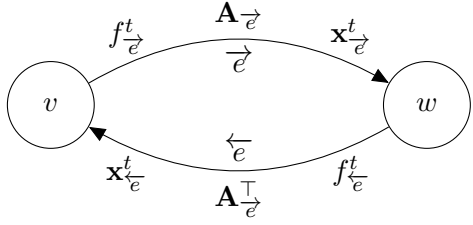
\begin{center}
	\begin{tikzpicture}[scale = 1]
	\label{graph:asy_AMP}
    \tikzstyle{point}=[draw,circle,minimum width=3em];
    \tikzstyle{fleche}=[->];
    \node[point] (v) at (0,0) {$v$};
    \node[point] (w) at (5,0) {$w$};
    \draw[fleche] (v) to[bend left]
    node[above,very near start]{$f^t_{\overrightarrow{e}}$}
    node[above,midway]{$\bA_{\overrightarrow{e}}$} node[below,midway]{$\overrightarrow{e}$} 
    node[above,very near end]{$\bx^t_{\overrightarrow{e}}$}
    (w);
    \draw[fleche] (w) to[bend left]
    node[below,very near start]{$f^t_{\overleftarrow{e}}$}
    node[below,midway]{$\bA_{\overrightarrow{e}}^\top$} node[above,midway]{$\overleftarrow{e}$} 
    node[below,very near end]{$\bx^t_{\overleftarrow{e}}$}
    (v);
	\end{tikzpicture}
\end{center}
In this case, the graph AMP equations \eqref{eq:graph-amp-1}-\eqref{eq:graph-amp-2} give
\begin{align}
\begin{split}
\label{eq:asy_amp}
    \mathbf{x}^{t+1}_{\overrightarrow{e}} &= \mathbf{A}_{\overrightarrow{e}} \mathbf{m}^t_{\overrightarrow{e}} - b^t_{\overrightarrow{e}} \mathbf{m}^{t-1}_{\overleftarrow{e}} \, ,   \\
    &\bm^t_{\overrightarrow{e}} = f^t_{\overrightarrow{e}}\left(\bx^t_{\overleftarrow{e}}\right) \, ,  \\
        \bx^{t+1}_{\overleftarrow{e}} &= \bA_{\overrightarrow{e}}^\top\bm^t_{\overleftarrow{e}} - b^t_{\overleftarrow{e}} \bm^{t-1}_{\overrightarrow{e}} \, ,  \\
    &\bm^t_{\overleftarrow{e}} = f^t_{\overleftarrow{e}}\left(\bx^t_{\overrightarrow{e}}\right) \, . 
    \end{split}
\end{align}
The corresponding state evolution (SE) property was proved in \cite{bayati2011dynamics} for the separable case and in \cite{berthier2020state} in the non-separable case. Note that the time indices proposed here are different from the ones appearing in these works. The time index convention adopted here generalizes better to more elaborate graphs. We show how to recover the usual time indices in Appendix \ref{app:time}.

\paragraph{Symmetric AMP.} The symmetric AMP iteration is central to our discussion as we show that all graph AMP iterations can be reduced to this case (with matrix-valued iterates, as detailed below). It is already known that the asymmetric case can be reduced to this case \cite{javanmard2013state}. The symmetric AMP iteration appears, e.g.,~when solving the low-rank matrix recovery problem \cite{rangan2012iterative,deshpande2014information}, or community detection in graphs \cite{deshpande2017asymptotic}. It corresponds to the degenerate graph with only one node and one loop. 
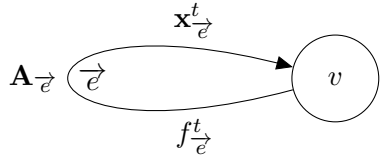
\begin{center}
\label{graph:sym_AMP}
	\begin{tikzpicture}[scale = 1]
    \tikzstyle{point}=[draw,circle,minimum width=3em];
    \tikzstyle{fleche}=[->];
    \node[point] (v) at (0,0) {$v$};
    \draw[fleche] (v) to[loop left, looseness = 35]
    node[below,very near start]{$f^t_{\overrightarrow{e}}$}
    node[left,midway]{$\bA_{\overrightarrow{e}}$} node[right,midway]{$\overrightarrow{e}$} 
    node[above,very near end]{$\bx^t_{\overrightarrow{e}}$}
    (v);
	\end{tikzpicture}
\end{center}
Recall that $\overleftarrow{e} = \overrightarrow{e}$ as $\overrightarrow{e}$ is a loop. In this case, the graph AMP equations \eqref{eq:graph-amp-1}-\eqref{eq:graph-amp-2} give
\begin{align}
\begin{split}
\label{eq:sym_amp}
    \bx^{t+1}_{\overrightarrow{e}} &= \bA_{\overrightarrow{e}} \bm^t_{\overrightarrow{e}} - b^t_{\overrightarrow{e}} \bm^{t-1}_{\overrightarrow{e}} \, ,  \\
    &\bm^t_{\overrightarrow{e}} = f^t_{\overrightarrow{e}}\left(\bx^t_{\overrightarrow{e}}\right) \, ,
    \end{split}
\end{align}
Here, as there is a single edge $\overrightarrow{e}$, the indexes are superfluous and could be dropped. For these equations, the SE property was proved in \cite{javanmard2013state} for the separable case and in \cite{berthier2020state} in the non-separable case. Note that the results of \cite{javanmard2013state} allow matrix-valued variables.

\paragraph{Multi-layer AMP.} The multi-layer AMP iteration appears when considering inference problems through a multi-layer random neural network, see \cite{manoel2017multi}. They correspond to a line graph whose length $l$ is the number of layers. 
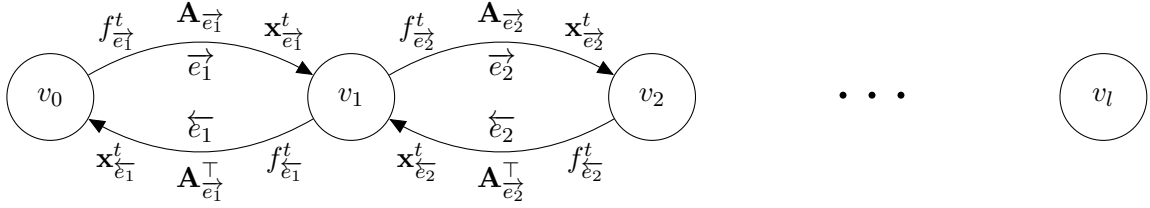
\begin{center}
	\begin{tikzpicture}[scale = 1]
    \tikzstyle{point}=[draw,circle,minimum width=3em];
    \tikzstyle{fleche}=[->];
    \node[point] (v0) at (0,0) {$v_0$};
    \node[point] (v1) at (4,0) {$v_1$};
    \node[point] (v2) at (8,0) {$v_2$};
    \node[] (v3) at (11,0) {\Huge$\cdots$};
    \node[point] (vl) at (14,0) {$v_l$};
    \draw[fleche] (v0) to[bend left]
    node[above,very near start]{$f^t_{\overrightarrow{e_1}}$}
    node[above,midway]{$\bA_{\overrightarrow{e_1}}$} node[below,midway]{$\overrightarrow{e_1}$} 
    node[above,very near end]{$\bx^t_{\overrightarrow{e_1}}$}
    (v1);
    \draw[fleche] (v1) to[bend left]
    node[below,very near start]{$f^t_{\overleftarrow{e_1}}$}
    node[below,midway]{$\bA_{\overrightarrow{e_1}}^\top$} node[above,midway]{$\overleftarrow{e_1}$} 
    node[below,very near end]{$\bx^t_{\overleftarrow{e_1}}$}
    (v0);
    \draw[fleche] (v1) to[bend left]
    node[above,very near start]{$f^t_{\overrightarrow{e_2}}$}
    node[above,midway]{$\bA_{\overrightarrow{e_2}}$} node[below,midway]{$\overrightarrow{e_2}$} 
    node[above,very near end]{$\bx^t_{\overrightarrow{e_2}}$}
    (v2);
    \draw[fleche] (v2) to[bend left]
    node[below,very near start]{$f^t_{\overleftarrow{e_2}}$}
    node[below,midway]{$\bA_{\overrightarrow{e_2}}^\top$} node[above,midway]{$\overleftarrow{e_2}$} 
    node[below,very near end]{$\bx^t_{\overleftarrow{e_2}}$}
    (v1);
	\end{tikzpicture}
\end{center}
In this case, the graph AMP equations \eqref{eq:graph-amp-1}-\eqref{eq:graph-amp-2} give
	\begin{align}
	\begin{split}
	\label{eq:mlamp}
    \bx^{t+1}_{\overrightarrow{e_1}} &= \bA_{\overrightarrow{e_1}} \bm^t_{\overrightarrow{e_1}} - b^t_{\overrightarrow{e_1}} \bm^{t-1}_{\overleftarrow{e_1}} \, ,  \\
    &\bm^t_{\overrightarrow{e_1}} = f^t_{\overrightarrow{e}_1}\left(\bx^t_{\overleftarrow{e_1}}\right) \, , \\
        \bx^{t+1}_{\overleftarrow{e_1}} &= \bA_{\overrightarrow{e_1}}^\top\bm^t_{\overleftarrow{e_1}} - b^t_{\overleftarrow{e_1}} \bm^{t-1}_{\overrightarrow{e_1}} \, ,  \\
    &\bm^t_{\overleftarrow{e_1}} = f^t_{\overleftarrow{e_1}}\left(\bx^t_{\overrightarrow{e_1}},\bx^t_{\overleftarrow{e_2}}\right) \, , \\
    &\qquad \\
        \bx^{t+1}_{\overrightarrow{e_2}} &= \bA_{\overrightarrow{e_2}} \bm^t_{\overrightarrow{e_2}} - b^t_{\overrightarrow{e_2}} \bm^{t-1}_{\overleftarrow{e_2}} \, ,  \\
    &\bm^t_{\overrightarrow{e_2}} = f^t_{\overrightarrow{e}_2}\left(\bx^t_{\overrightarrow{e_1}},\bx^t_{\overleftarrow{e_2}}\right) \, , \\
        \bx^{t+1}_{\overleftarrow{e_2}} &= \bA_{\overrightarrow{e_2}}^\top\bm^t_{\overleftarrow{e_2}} - b^t_{\overleftarrow{e_2}} \bm^{t-1}_{\overrightarrow{e_2}} \, ,  \\
    &\bm^t_{\overleftarrow{e_2}} = f^t_{\overleftarrow{e_2}}\left(\bx^t_{\overrightarrow{e_2}},\bx^t_{\overleftarrow{e_3}}\right) \, , \\
    &\qquad\\
    &\qquad\quad\vdots 
    \end{split}
\end{align}
Note that the non-linearities now take several variables as inputs when there are several incoming edges at a node. 

\paragraph{Spiked matrix model under generative multi-layer priors.} Of course, the structures described above can be combined to tackle new AMP iterations. For instance, the paper \cite{aubin2019committee} studies the recovery of noisy symmetric rank-$1$ matrix when the spike comes from a known multi-layer generative prior. The associated AMP iteration corresponds to the following graph, where the loop corresponds to the spike recovery and the other edges correspond to multi-layer prior on the spike. 
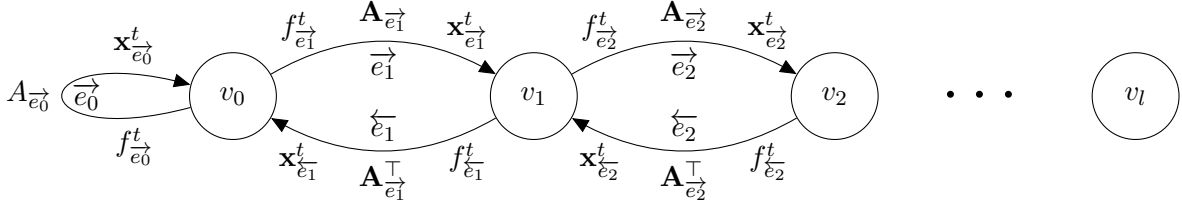
\begin{center}
\label{graph:spike_MLAMP}
	\begin{tikzpicture}[scale = 1]
    \tikzstyle{point}=[draw,circle,minimum width=3em];
    \tikzstyle{fleche}=[->];
    \node[point] (v0) at (0,0) {$v_0$};
    \node[point] (v1) at (4,0) {$v_1$};
    \node[point] (v2) at (8,0) {$v_2$};
    \node[] (v3) at (10,0) {\Huge$\cdots$};
    \node[point] (vl) at (12,0) {$v_l$};
    \draw[fleche] (v0) to[loop left, looseness = 20]
    node[below,very near start]{$f^t_{\overrightarrow{e_0}}$}
    node[left,midway]{$A_{\overrightarrow{e_0}}$} node[right,midway]{$\overrightarrow{e_0}$} 
    node[above,very near end]{$\bx^t_{\overrightarrow{e_0}}$}
    (v0);
    \draw[fleche] (v0) to[bend left]
    node[above,very near start]{$f^t_{\overrightarrow{e_1}}$}
    node[above,midway]{$\bA_{\overrightarrow{e_1}}$} node[below,midway]{$\overrightarrow{e_1}$} 
    node[above,very near end]{$\bx^t_{\overrightarrow{e_1}}$}
    (v1);
    \draw[fleche] (v1) to[bend left]
    node[below,very near start]{$f^t_{\overleftarrow{e_1}}$}
    node[below,midway]{$\bA_{\overrightarrow{e_1}}^\top$} node[above,midway]{$\overleftarrow{e_1}$} 
    node[below,very near end]{$\bx^t_{\overleftarrow{e_1}}$}
    (v0);
    \draw[fleche] (v1) to[bend left]
    node[above,very near start]{$f^t_{\overrightarrow{e_2}}$}
    node[above,midway]{$\bA_{\overrightarrow{e_2}}$} node[below,midway]{$\overrightarrow{e_2}$} 
    node[above,very near end]{$\bx^t_{\overrightarrow{e_2}}$}
    (v2);
    \draw[fleche] (v2) to[bend left]
    node[below,very near start]{$f^t_{\overleftarrow{e_2}}$}
    node[below,midway]{$\bA_{\overrightarrow{e_2}}^\top$} node[above,midway]{$\overleftarrow{e_2}$} 
    node[below,very near end]{$\bx^t_{\overleftarrow{e_2}}$}
    (v1);
	\end{tikzpicture}
\end{center}
In this case, the graph AMP equations \eqref{eq:graph-amp-1}-\eqref{eq:graph-amp-2} give
\begin{align}
	\begin{split}
	\label{eq:spike_mlamp}
	    \bx^{t+1}_{\overrightarrow{e_0}} &= \bA_{\overrightarrow{e_0}} \bm^t_{\overrightarrow{e_0}} - b^t_{\overrightarrow{e_0}} \bm^{t-1}_{\overrightarrow{e_0}} \, ,  \\
    &\bm^t_{\overrightarrow{e_0}} = f^t_{\overrightarrow{e_0}}\left(\bx^t_{\overrightarrow{e_0}},\bx^t_{\overleftarrow{e_1}}\right) \, , \\
    &\qquad\\
    \bx^{t+1}_{\overrightarrow{e_1}} &= \bA_{\overrightarrow{e_1}} \bm^t_{\overrightarrow{e_1}} - b^t_{\overrightarrow{e_1}} \bm^{t-1}_{\overleftarrow{e_1}} \, ,  \\
    &\bm^t_{\overrightarrow{e_1}} = f^t_{\overrightarrow{e}_1}\left(\bx^t_{\overrightarrow{e_0}},\bx^t_{\overleftarrow{e_1}}\right) \, , \\
        \bx^{t+1}_{\overleftarrow{e_1}} &= \bA_{\overrightarrow{e_1}}^\top\bm^t_{\overleftarrow{e_1}} - b^t_{\overleftarrow{e_1}} \bm^{t-1}_{\overrightarrow{e_1}} \, ,  \\
    &\bm^t_{\overleftarrow{e_1}} = f^t_{\overleftarrow{e_1}}\left(\bx^t_{\overrightarrow{e_1}},\bx^t_{\overleftarrow{e_2}}\right) \, , \\
    &\qquad \\
    &\qquad\quad\vdots
    \end{split}
\end{align}
\section{State evolution for graph-based AMP iterations}
\label{sec:se-graph}
In this section, we start by presenting the most straightforward form of our result, and show afterwards how several refinements can be added.
\subsection{Main theorem}
AMP algorithms admit a state evolution description under two major assumptions: that the interactions matrices $\bA_{\overrightarrow{e}}$ are sufficiently random---in our case Gaussian or GOE---and that the dimensions $n = (n_v)_{v \in V}$ of all the variables converge to infinity with fixed ratios. 

\paragraph{Assumptions.} We make the following assumptions:
\begin{enumerate}[font={\bfseries},label={(A\arabic*)}]
\item \label{ass:main1} The matrices $(\bA_{\overrightarrow{e}})_{\overrightarrow{e} \in \overrightarrow{E}}$ are random and independent, up to the symmetry condition $\bA_{\overleftarrow{e}} = \bA_{\overrightarrow{e}}^\top$. Moreover, if $(v,w) \in \overrightarrow{E}$ is not a loop in $G$, i.e., $v \neq w$, then $\bA_{(v,w)}$ has independent centered Gaussian entries with variance $1/N$. If $(v,v) \in \overrightarrow{E}$ is a loop in $G$, then $\bA_{(v,v)}$ has independent entries (up to the symmetry $\bA_{(v,v)} = \bA_{(v,v)}^\top$), centered Gaussian with variance $2/N$ on the diagonal and variance $1/N$ off the diagonal. 
\item For all $v \in V$, $n_v \to \infty$ and $n_v/N$ converges to a well-defined limit $\delta_v \in [0,1]$. We denote by $n \to \infty$ the limit under this scaling.
\item For all $t \in \mathbb{N}$ and $\overrightarrow{e} \in \overrightarrow{E}$, the non-linearity $f^{t}_{\overrightarrow{e}}$ is pseudo-Lipschitz of finite order, uniformly with respect to the problem dimensions $n = (n_v)_{v \in V}$ (see Definition \ref{def:pseudo-lip} in Appendix \ref{app:sec_conc}).
\item For all $\overrightarrow{e}\in E$, $\norm{\bx^{0}_{\overrightarrow{e}}}_{2}/\sqrt{N}$
converges to a finite constant as $n \to \infty$.
\item For all $\overrightarrow{e} \in E$, the following limit exists and is finite:
\begin{equation*}
    \lim_{n\to\infty} \frac{1}{N}  \left\langle f^0_{\overrightarrow{e}} \left(\left(\bx^0_{\overrightarrow{e}'}\right)_{ {\overrightarrow{e}'}:{\overrightarrow{e}'} \to {\overrightarrow{e}}}\right), f^0_{\overrightarrow{e}} \left(\left(\bx^0_{\overrightarrow{e}'}\right)_{ {\overrightarrow{e}'}:{\overrightarrow{e}'} \to {\overrightarrow{e}}}\right) \right\rangle
\end{equation*}
\item Let $(\kappa_{\overrightarrow{e}})_{\overrightarrow{e}\in E}$ be an array of bounded non-negative reals and $\bZ_{\overrightarrow{e}} \sim \mathbf{N}(0,\kappa_{\overrightarrow{e}}\mathbf{I}_{n_{w}})$ independent random variables for all $\overrightarrow{e}$. For all $\overrightarrow{e}\in E$, for any $t \in \mathbb{N}_{>0}$, the following limit exists and is finite:
\begin{equation*}
    \lim_{n \to \infty} \frac{1}{N}  \mathbb{E}\left[\left \langle f^0_{\overrightarrow{e}}  \left(\left(\bx^0_{\overrightarrow{e}'}\right)_{ {\overrightarrow{e}'}:{\overrightarrow{e}'} \to {\overrightarrow{e}}}\right), f^{t}_{\overrightarrow{e}}\left(\left(\bZ^{t}_{\overrightarrow{e}'}\right)_{ {\overrightarrow{e}'}:{\overrightarrow{e}'} \to {\overrightarrow{e}}}\right)\right \rangle\right].
\end{equation*}

\item \label{ass:main7} Consider any array of $2\times 2$ positive definite matrices $(\boldsymbol{S}_{\overrightarrow{e}})_{\overrightarrow{e}\in E}$ and the collection of random variables $(\bZ_{\overrightarrow{e}},\bZ^{'}_{\overrightarrow{e}}) \sim \mathbf{N}(0,\boldsymbol{S}_{\overrightarrow{e}}\otimes\mathbf{I}_{n_{w}}))$ defined independently for each edge $\overrightarrow{e}$. Then for any $\overrightarrow{e}\in E$ and $s,t >0$, the following limit exists and is finite:
\begin{equation*}
    \lim_{n \to \infty} \frac{1}{N}  \mathbb{E}\left[\left \langle f^{s}_{\overrightarrow{e}}\left(\left(\bZ^{s}_{\overrightarrow{e}'}\right)_{ {\overrightarrow{e}'}:{\overrightarrow{e}'} \to {\overrightarrow{e}}}\right), f^{t}_{\overrightarrow{e}}\left(\left(\tilde{\bZ}^{t}_{\overrightarrow{e}'}\right)_{ {\overrightarrow{e}'}:{\overrightarrow{e}'} \to {\overrightarrow{e}}}\right)\right \rangle \right].
\end{equation*}
\end{enumerate}
\paragraph{Remark on the assumptions.}
In the literature, the random matrices $\bA_{(v,w)}$ of AMP iterations are often scaled with variances $1/n_{w}$. To recover the desired scaling, it is sufficient to rescale the non-linearity on which a given matrix acts with the corresponding aspect ratio $\delta_{w}$.
\begin{definition}[State evolution iterates] The state evolution iterates are composed of one infinite-dimensional array $(\boldsymbol{\kappa}_{\overrightarrow{e}}^{s,r})_{r,s > 0}$ of real values for each edge $\overrightarrow{e} \in \overrightarrow{E}$. These arrays are generated as follows. Define the first state evolution iterates 
\begin{equation*}
    \boldsymbol{\kappa}^{1,1}_{\overrightarrow{e}} = \lim_{n\to\infty} \frac{1}{N} \left\Vert f^0_{\overrightarrow{e}} \left(\left(\bx^0_{\overrightarrow{e}'}\right)_{ {\overrightarrow{e}'}:{\overrightarrow{e}'} \to {\overrightarrow{e}}}\right)\right\Vert_2^2  \, , \qquad \overrightarrow{e} \in \overrightarrow{E} \, .
\end{equation*}

Recursively, once $(\boldsymbol{\kappa}^{s,r}_{\overrightarrow{e}})_{s,r \leq t,\overrightarrow{e} \in \overrightarrow{E}}$ are defined for some $t \geq 1$, define independently for each $\overrightarrow{e} \in \overrightarrow{E}$, $\bZ^0_{\overrightarrow{e}} = \bx^0_{\overrightarrow{e}}$ and $(\bZ^1_{\overrightarrow{e}}, \dots, \bZ^t_{\overrightarrow{e}})$ a centered Gaussian random vector of covariance $(\boldsymbol{\kappa}^{r,s}_{\overrightarrow{e}})_{r,s \leq t} \otimes I_{n_w}$. We then define new state evolution iterates 
\begin{align*}
    &\boldsymbol{\kappa}^{t+1, s+1}_{\overrightarrow{e}} = \boldsymbol{\kappa}^{s+1, t+1}_{\overrightarrow{e}} = \lim_{n \to \infty} \frac{1}{N} \E\left[ \left\langle f^s_{\overrightarrow{e}} \left(\left(\bZ^s_{\overrightarrow{e}'}\right)_{ {\overrightarrow{e}'}:{\overrightarrow{e}'} \to {\overrightarrow{e}}}\right), f^t_{\overrightarrow{e}} \left(\left(\bZ^t_{\overrightarrow{e}'}\right)_{ {\overrightarrow{e}'}:{\overrightarrow{e}'} \to {\overrightarrow{e}}}\right) \right\rangle \right] \, \\
    &\mbox{for all} \quad s \in \{1, \dots, t \} \, ,  \overrightarrow{e} \in \overrightarrow{E} \, . 
\end{align*}
\end{definition}
\begin{theorem}
\label{thm:graph-AMP}
Assume \ref{ass:main1}-\ref{ass:main7}. Define, as above, independently for each $\overrightarrow{e} = (v,w) \in \overrightarrow{E}$, $\bZ^0_{\overrightarrow{e}}= \bx^0_{\overrightarrow{e}}$ and $(\bZ^1_{\overrightarrow{e}}, \dots, \bZ^t_{\overrightarrow{e}})$ a centered Gaussian random vector of covariance $(\boldsymbol{\kappa}^{r,s}_{\overrightarrow{e}})_{r,s \leq t} \otimes \mathbf{I}_{n_{w}}$. Then for any sequence of uniformly (in $n$) pseudo-Lipschitz function $\Phi:\R^{(t+1)N} \to \R $, 
\begin{equation*}
    \Phi\left(\left(\bx^s_{\overrightarrow{e}}\right)_{0 \leq s \leq t, \overrightarrow{e} \in \overrightarrow{E}}\right) \approxP \E \left[ \Phi\left(\left(\bZ^s_{\overrightarrow{e}}\right)_{0 \leq s \leq t, \overrightarrow{e} \in \overrightarrow{E}}\right) \right]
\end{equation*}
\end{theorem}

\subsection{Reduction of graph-based AMP iterations to the matrix-valued, non-separable symmetric case}
\label{ap:reduction-graph}

The core strategy in the proof of Theorem \ref{thm:graph-AMP} is to reduce the graph AMP iteration \eqref{eq:graph-amp-1}-\eqref{eq:graph-amp-3} into a symmetric AMP iteration with matrix-valued iteration, i.e., an iteration of the form 
\begin{align}
\bX^{t+1} &= \mathbf{A}\bM^{t}-\bM^{t-1}(\bb^{t})^{\top} && \in \R^{N\times q} \, , \label{eq:sym-amp-iteration-1} \\
\bM^{t} &=f^{t}(\bX^{t}) && \in \R^{N\times q} \, , \\ 
\bb_t &= \frac{1}{N} \sum_{i=1}^N \frac{\partial f^t_i}{\partial \bX_i}(\bX^t) && \in 
\R^{q\times q}\, . \label{eq:sym-amp-iteration-2}
\end{align}
Here, $\bA$ is a $N \times N$ GOE matrix, the iterates $\bX^t, \bM^t$ are $N \times q$ matrices, and $f^t:\R^{N \times q} \to \R^{N \times q}$ are non-separable non-linearities. A rigorous SE description for this iteration is established in Appendix \ref{SE_sym_AMP}; it is an extension of the results of \cite{javanmard2013state,berthier2020state}. 

In this section, we show that the graph AMP iteration \eqref{eq:graph-amp-1}-\eqref{eq:graph-amp-3} can be formulated as a symmetric AMP iteration \eqref{eq:sym-amp-iteration-1}-\eqref{eq:sym-amp-iteration-2} with matrix iterates. In Appendix \ref{ap:proof-graph-amp}, this reduction is used to show that Theorem \ref{thm:graph-AMP} follows from its equivalent on symmetric iterations.

\medskip
Let $q = | \overrightarrow{E} |$,
%{\color{green}(resp.~$q = \sum_{\overrightarrow{e} \in \overrightarrow{E}} q_{\overrightarrow{e}}$)}
$\overrightarrow{e}_1, \dots, \overrightarrow{e}_l$ be the loops of $G$ and $\overrightarrow{e}_{l+1}, \overleftarrow{e}_{l+1}, \dots, \overrightarrow{e}_{m}, \overleftarrow{e}_{m}$ be the other edges of the graph. Define 
\begin{equation*}
    \bX^0 = \begin{pmatrix}
    \bx^0_{\overrightarrow{e}_1} & & & & & & &\\
    & \ddots & & & & & \ast &\\
    & & \bx^0_{\overrightarrow{e}_l} & & & & &\\
    & & & \bx^0_{\overrightarrow{e}_{l+1}} & & & &\\
    & & & &\bx^0_{\overleftarrow{e}_{l+1}} & & &\\
    & & & & & \ddots & & \\
     & \ast & & & & & \bx^0_{\overrightarrow{e}_{m}} & \\
     & & & & & & & \bx^0_{\overleftarrow{e}_{m}}
    \end{pmatrix} \in \R^{N \times q} \, .
\end{equation*}
where $\ast$ denotes entries whose values do not matter for what follows. 
Let $\bA$ be a $N \times N$ GOE matrix such that 
\begin{equation*}
    \bA = 
    \begin{pmatrix}
    \bA_{\overrightarrow{e}_1} & & & & & & & \\
    & \ddots& & & & & \ast & \\
    & & \bA_{\overrightarrow{e}_l} & & & & & \\
    & & & \ast & \bA_{\overrightarrow{e}_{l+1}} & & & \\
    & & & \bA_{\overleftarrow{e}_{l+1}} & \ast & & & \\
    & & & & & \ddots & & \\
    & \ast& & & &  & \ast & \bA_{\overrightarrow{e}_{m}} \\
    & & & & & & \bA_{\overleftarrow{e}_{m}} & \ast 
    \end{pmatrix} \, .
\end{equation*}
Finally, define the non-linearities $f_t: \R^{N \times q} \to  \R^{N \times q}$ as 

\begin{align}
    &f^t
    \begin{pmatrix}
    \bx_{\overrightarrow{e}_1} & & & & & & & \\
    & \ddots & & & & &\ast  & \\
    & & \bx_{\overrightarrow{e}_l} & & & & & \\
    & & & \bx_{\overrightarrow{e}_{l+1}} & & & & \\
    & & & &\bx_{\overleftarrow{e}_{l+1}} & & & \\
    & & & & & \ddots & & \\
     & \ast & & & & & \bx_{\overrightarrow{e}_{m}}  & \\
     & & & & & & &\bx_{\overleftarrow{e}_{m}}
    \end{pmatrix} \label{eq:def-non-linearity}\\
    &\qquad= 
    \begin{pmatrix}
    f^t_{\overrightarrow{e}_1}\left(\left(\bx_{\overrightarrow{e}}\right)_{\overrightarrow{e}:\overrightarrow{e} \rightarrow \overrightarrow{e}_1}\right) & & & & & & & \\
    & \ddots & & & & & 0 & \\
    & & f^t_{\overrightarrow{e}_l}\left(\dots\right) & & & & &\\
    & & & 0 & f^t_{\overleftarrow{e}_{l+1}}(\dots) & & & \\
    & & & f^t_{\overrightarrow{e}_{l+1}}(\dots) & 0 & & & \\
    & & & & & \ddots & & \\
     & 0 & & & & &  0 & f^t_{\overleftarrow{e}_{m}}(\dots)\\
     & & & & & & f^t_{\overrightarrow{e}_{m}}(\dots) & 0
    \end{pmatrix} \nonumber
\end{align}

\begin{lemma}
    \label{lemma:reduction}
Define $\bX^0$, $\bA$ and $f^t$ as above. Then the iterates $\bX^t$ of the symmetric AMP iteration \eqref{eq:sym-amp-iteration-1}-\eqref{eq:sym-amp-iteration-2} are of the form 
\begin{equation*}
    \bX=  \begin{pmatrix}
    \bx^t_{\overrightarrow{e}_1} & & & & & & &\\
    & \ddots & & & & & \ast &\\
    & & \bx^t_{\overrightarrow{e}_l} & & & & &\\
    & & & \bx^t_{\overrightarrow{e}_{l+1}} & & & &\\
    & & & &\bx^t_{\overleftarrow{e}_{l+1}} & & &\\
    & & & & & \ddots & & \\
     & \ast & & & & & \bx^t_{\overrightarrow{e}_{m}} & \\
     & & & & & & & \bx^t_{\overleftarrow{e}_{m}}
    \end{pmatrix} \in \R^{N \times q} \, ,
\end{equation*}
where $\bx^t_{\overrightarrow{e}}$ denote the iterates of the graph-AMP iteration \eqref{eq:graph-amp-1}-\eqref{eq:graph-amp-3}. 
\end{lemma}

\begin{proof}
We proceed by induction. Assume that $\bX^t$ and $\bX^{t-1}$ are indeed of this form and we show the claim for $\bX^{t+1}$. We use equations \eqref{eq:sym-amp-iteration-1}-\eqref{eq:sym-amp-iteration-2} to compute $\bX^{t+1}$; we start by computing the Onsager term $\bb_t = \frac{1}{N} \sum_{i=1}^N \frac{\partial f^t_i}{\partial \bX_i}(\bX^t) \, \in \R^{q\times q}$. From the formula for $f^t$, we compute
\begin{align*}
    \bb_t &= \frac{1}{N} \begin{pmatrix}
    \Tr \frac{\partial f^t_{\overrightarrow{e}_1}}{\partial \bx_{\overrightarrow{e}_1}}(\dots)  & & & & &  \\
    & \ddots & & & 0 & \\
    & & \Tr \frac{\partial f^t_{\overrightarrow{e}_l}}{\partial \bx_{\overrightarrow{e}_l}}(\dots) & & &  \\
    & & & 0 & \Tr \frac{\partial f^t_{\overrightarrow{e}_{l+1}}}{\partial \bx_{\overleftarrow{e}_{l+1}}}(\dots) &  \\
    & 0 & & \Tr \frac{\partial f^t_{\overleftarrow{e}_{l+1}}}{\partial \bx_{\overrightarrow{e}_{l+1}}}(\dots) & 0 & \\
    & & & & & \ddots & & 
    \end{pmatrix} \\
    &= \begin{pmatrix}
    \bb^t_{\overrightarrow{e}_1} & & & & &  \\
    & \ddots & & & 0 &  \\
    & & \bb^t_{\overrightarrow{e}_l} & & &  \\
    & & & 0 & \bb^t_{\overrightarrow{e}_{l+1}} &  \\
    & 0 & & \bb^t_{\overleftarrow{e}_{l+1}} & 0 &  \\
    & & & & & \ddots  \\
    \end{pmatrix} \, .
\end{align*}
%{\color{green}resp.
%\begin{align*}
%    b_t = \frac{1}{N} \begin{pmatrix}
%    \sum_{i=1}^{q_{\overrightarrow{e}_1}} \frac{\partial f_{\overrightarrow{e}_1,i}}{\partial x_{\overrightarrow{e}_1,i}}(\dots)  & & & & & & & \\
%    & \ddots & & & & & 0 & \\
%    & & \sum_{i=1}^{q_{\overrightarrow{e}_l}} \frac{\partial f_{\overrightarrow{e}_l,i}}{\partial x_{\overrightarrow{e}_l,i}}(\dots) & & & & & \\
%    & & & 0 & \sum_{i=1}^{q_{\overrightarrow{e}_{l+1}}} \frac{\partial f_{\overrightarrow{e}_{l+1},i}}{\partial x_{\overleftarrow{e}_{l+1},i}}(\dots) & & & \\
%    & & & \sum_{i=1}^{q_{\overrightarrow{e}_{l+1}}} \frac{\partial f_{\overleftarrow{e}_{l+1},i}}{\partial x_{\overrightarrow{e}_{l+1},i}}(\dots) & 0 & & & \\
%    & & & & & \ddots & & \\
%    & 0 & & & & & 0 & \sum_{i=1}^{q_{\overrightarrow{e}_{m}}} \frac{\partial f_{\overrightarrow{e}_{m},i}}{\partial x_{\overleftarrow{e}_{m},i}}(\dots)\\
%    & & & & & & \sum_{i=1}^{q_{\overrightarrow{e}_{m}}} \frac{\partial f_{\overleftarrow{e}_{m},i}}{\partial x_{\overrightarrow{e}_{m},i}}(\dots) & 0
%    \end{pmatrix}
%\end{align*}}
Then we can now compute
\begin{align*}
    \bX^{t+1} &= \bA \bM^t-  \bM^{t-1} \bb_t^\top\, .
\end{align*}
First, 
\begin{align*}
    \bA \bM&=         \begin{pmatrix}
    \bA_{\overrightarrow{e}_1} & & & & &  \\
    & \ddots& & & & \\
    & & \bA_{\overrightarrow{e}_l} & & & \\
    & & & \ast & \bA_{\overrightarrow{e}_{l+1}} & \\
    & & & \bA_{\overleftarrow{e}_{l+1}} & \ast &\\
    & & & & & \ddots 
    \end{pmatrix}
        \begin{pmatrix}
    f^t_{\overrightarrow{e}_1}\left(
    %\left(\bx^t_{\overrightarrow{e}}\right)_{\overrightarrow{e}:\overrightarrow{e} \rightarrow \overrightarrow{e}_1}
    .\right) & & & & & \\
    & \ddots & & & &  \\
    & & f^t_{\overrightarrow{e}_l}\left(.\right) & & & \\
    & & & 0 & f^t_{\overleftarrow{e}_{l+1}}(.) &  \\
    & & & f^t_{\overrightarrow{e}_{l+1}}(.) & 0 & \\
    & & & & & \ddots
    \end{pmatrix} \\
    &= \begin{pmatrix}
    \bA_{\overrightarrow{e}_1} f^t_{\overrightarrow{e}_1}\left(\left(\bx^t_{\overrightarrow{e}}\right)_{\overrightarrow{e}:\overrightarrow{e}\to\overrightarrow{e_1}}\right) & & & & &  \\
    & \ddots& & & \ast & \\
    & & \bA_{\overrightarrow{e}_l}f^t_{\overrightarrow{e}_l}\left(.\right) & & & \\
    & & & \bA_{\overrightarrow{e}_{l+1}}f^t_{\overrightarrow{e}_{l+1}}(.) & & \\
    & \ast & & & \bA_{\overleftarrow{e}_{l+1}}f^t_{\overleftarrow{e}_{l+1}}(.)  &\\
    & & & & & \ddots 
    \end{pmatrix}\, .
\end{align*} 
Second, 
\begin{align*}
    \bM^{t-1}\bb_t&= 
            \begin{pmatrix}
    f^{t-1}_{\overrightarrow{e}_1}\left(
    %\left(\bx^t_{\overrightarrow{e}}\right)_{\overrightarrow{e}:\overrightarrow{e} \rightarrow \overrightarrow{e}_1}
    .\right) & & & & & \\
    & \ddots & & & &  \\
    & & f^{t-1}_{\overrightarrow{e}_l}\left(.\right) & & & \\
    & & & 0 & f^{t-1}_{\overleftarrow{e}_{l+1}}(.) &  \\
    & & & f^{t-1}_{\overrightarrow{e}_{l+1}}(.) & 0 & \\
    & & & & & \ddots
    \end{pmatrix}
    \begin{pmatrix}
    \bb^t_{\overrightarrow{e}_1} & & & & &  \\
    & \ddots & & & & \\
    & & \bb^t_{\overrightarrow{e}_l} & & & \\
    & & & 0 & \bb^t_{\overleftarrow{e}_{l+1}} &  \\
    & & & \bb^t_{\overrightarrow{e}_{l+1}} & 0 & \\
    & & & & & \ddots 
    \end{pmatrix} \\
    &= 
    \begin{pmatrix}
    \bb^t_{\overrightarrow{e}_1} f^{t-1}_{\overrightarrow{e}_1}\left(\left(\bx^t_{\overrightarrow{e}}\right)_{\overrightarrow{e}:\overrightarrow{e}\to\overrightarrow{e_1}}\right) & & & & &  \\
    & \ddots& & & 0 & \\
    & & \bb^t_{\overrightarrow{e}_l}f^{t-1}_{\overrightarrow{e}_l}\left(.\right) & & & \\
    & & & \bb^t_{\overrightarrow{e}_{l+1}}f^{t-1}_{\overleftarrow{e}_{l+1}}(.) & & \\
    & 0 & & & \bb^t_{\overleftarrow{e}_{l+1}}f^{t-1}_{\overrightarrow{e}_{l+1}}(.)  &\\
    & & & & & \ddots 
    \end{pmatrix}
    \, .
\end{align*}
Thus, combining the above equations, we obtain 
\begin{align*}
    \bX^{t+1} &= \bA \bM-  \bM^{t-1} \bb_t^\top\\
    &= \begin{pmatrix}
    \bA_{\overrightarrow{e}_1} f^t_{\overrightarrow{e}_1}\left(.\right) -b^t_{\overrightarrow{e}_1} f^{t-1}_{\overrightarrow{e}_1}\left(.\right)  & & & &   \\
    & \ddots& & & \ast \\
    & & \bA_{\overrightarrow{e}_l}f^t_{\overrightarrow{e}_l}\left(.\right) -b^t_{\overrightarrow{e}_l} f^{t-1}_{\overrightarrow{e}_l}\left(.\right) & &  \\
    & & & \bA_{\overrightarrow{e}_{l+1}}f^t_{\overrightarrow{e}_{l+1}}(.) - b^t_{\overrightarrow{e}_{l+1}}f^{t-1}_{\overleftarrow{e}_{l+1}}(.) & \\
    & & & & \ddots 
    \end{pmatrix} \\
    &= \begin{pmatrix}
    \bx^{t+1}_{\overrightarrow{e}_1}  & & & &   \\
    & \ddots& & \ast & \\
    & & \bx^{t+1}_{\overrightarrow{e}_l}  & &  \\
    & \ast& & \bx^{t+1}_{\overrightarrow{e}_{l+1}}  & \\
    & & & & \ddots 
    \end{pmatrix} \, .
\end{align*}
This proves the induction. 
\end{proof}

\subsection{Useful extensions} 
\label{sec:extensions}
Here we present several refinements of Theorem \ref{thm:graph-AMP} that can be obtained in a straightforward fashion and appear often in statistical inference problems.
\paragraph{Matrix-valued variables.}
The variables $\bx_{\overrightarrow{e}},\bm_{\overrightarrow{e}}$ initially defined as vectors can be extended to matrices with a finite number of columns, and the non-linearities $f^{t}_{\overrightarrow{e}}$ are then matrix-valued functions of matrix-valued variables. 
\begin{itemize}
    \item $n_v \in \N_{>0}$ is now the number of lines  of the variables coming in node $v \in V$. The definition $N = \sum_{(v,w) \in \overrightarrow{E}} n_w$ remains the same.
    \item Let $q_{\overrightarrow{e}} \in \N_{>0}$ be the number of columns of $\bx^{t}_{\overrightarrow{e}}$. We assume that, for all $\overrightarrow{e} \in E$, $q_{\overrightarrow{e}} = q_{\overleftarrow{e}}$, and the $q_{\overrightarrow{e}}$ remain constant, independently of $n \to \infty$.
    \item The initial condition becomes $\bx^0_{(v,w)} \in \R^{n_w \times q_{(v,w)}}$, for all edges $\overrightarrow{e} = (v,w)$.
    \item Non-linearities $f_{t}$ indexed by the edge $\overrightarrow{e} = (v,w) \in \overrightarrow{E}$, $f^t_{(v,w)}(x^t_{\overrightarrow{e}'}, \overrightarrow{e}' \to \overrightarrow{e})$ are now functions from  $\times_{\overrightarrow{e}' \to \overrightarrow{e}}\R^{n_v \times q_{\overrightarrow{e}'}}$ to $\R^{n_v \times q_{(v,w)}}$.
\end{itemize}
The AMP iterates are then recursively defined with:
\begin{align}
    \mathbf{x}^{t+1}_{\overrightarrow{e}} &= \bA_{\overrightarrow{e}} \mathbf{m}^t_{\overrightarrow{e}} -  \mathbf{m}^{t-1}_{\overleftarrow{e}}(\mathbf{b}^t_{\overrightarrow{e}})^{\top} \quad \in \mathbb{R}^{n_{w}\times q_{\overrightarrow{e}}} \, ,  \\
    \mathbf{m}^t_{\overrightarrow{e}} &= f^t_{\overrightarrow{e}}\left(\left(\mathbf{x}^t_{\overrightarrow{e}'}\right)_{\overrightarrow{e}':\overrightarrow{e}' \to \overrightarrow{e}}\right) \, ,
\end{align}
where each Onsager term is now a matrix given by:
\begin{equation*} 
    \mathbf{b}^t_{\overrightarrow{e}} = \frac{1}{N} \sum_{i=1}^{n_v} \frac{\partial f^t_{\overrightarrow{e},i}}{\partial \mathbf{x}_{\overleftarrow{e},i}} \left(\left(\mathbf{x}^t_{\overrightarrow{e}'}\right)_{\overrightarrow{e}':\overrightarrow{e}' \to \overrightarrow{e}}\right) \qquad \in \mathbb{R}^{q_{\overrightarrow{e}}\times q_{\overrightarrow{e}}} \, .
\end{equation*}
where we used the notation from Eq.\eqref{eq:Ons_jacob}.
The state evolution equations then read
\begin{equation*}
    \boldsymbol{\kappa}^{1,1}_{\overrightarrow{e}} = \lim_{n\to\infty} \frac{1}{N}  f^0_{\overrightarrow{e}} (\bx^0_{\overrightarrow{e}'}, {\overrightarrow{e}'} \to {\overrightarrow{e}})^{\top} f^0_{\overrightarrow{e}} (\bx^0_{\overrightarrow{e}'}, {\overrightarrow{e}'} \to {\overrightarrow{e}}) \quad \in \R^{q_{\overrightarrow{e}}\times q_{\overrightarrow{e}}} \, , \quad \overrightarrow{e} \in \overrightarrow{E} \, .
\end{equation*}
\begin{align*}
   &\boldsymbol{\kappa}^{t+1, s+1}_{\overrightarrow{e}} = \boldsymbol{\kappa}^{s+1, t+1}_{\overrightarrow{e}} = \lim_{n \to \infty} \frac{1}{N} \E\left[  f^s_{\overrightarrow{e}} (\bZ^s_{\overrightarrow{e}'}, {\overrightarrow{e}'} \to {\overrightarrow{e}})^{\top} f^t_{\overrightarrow{e}} (\bZ^t_{\overrightarrow{e}'}, {\overrightarrow{e}'} \to {\overrightarrow{e}})  \right] \quad \in \R^{q_{\overrightarrow{e}}\times q_{\overrightarrow{e}}} \, \\
   &\mbox{for all}\quad 1\leqslant s \leqslant t \, ,  \overrightarrow{e} \in \overrightarrow{E} \, . 
\end{align*}
where the Gaussian fields generalize straightforwardly to  $\bZ^{t}_{\overrightarrow{e}} \sim \mathbf{N}(0,\boldsymbol{\kappa}^{t,t}_{\overrightarrow{e}}\otimes\mathbf{I}_{n_w})  \in \mathbb{R}^{n_{w}\times q_{\overrightarrow{e}}}$ for each edge. Using these generalized definitions, the above statement of Theorem \ref{thm:graph-AMP} and its proof can be adapted easily. We give examples throughout Section \ref{sec:applications}.

\paragraph{Additional random variables in the non-linearities.}
Many inference problems are formulated with a ``planted'' signal, i.e., a ground truth signal parametrizing the function the statistician tries to reconstruct, sometimes called \emph{teacher} in statistical physics. This often leads to the dependence of certain non-linearities on additional random variables. As long as they appropriately concentrate and are independent on the rest of the problem, they can be treated in straightforward fashion with an additional average in the SE equations as done in \cite{javanmard2013state}, where the summability is reduced to second-order moments conditions due to the separability of the update functions.
However it is not always straightforward to isolate the independent contribution in the teacher which is often generated using the matrices found in the AMP algorithm, effectively introducing a correlation between the matrices and non-linearities. In appendix \ref{app:low_rank_pert}, we propose a generic way to deal 
with such dependencies with two additional results in the form of Lemmas \ref{lemma:spike_SE} and Lemma \ref{lemma:proj_SE}. These two lemmas may be combined at will to deal with a wide range of perturbations relevant to inference problems. We now give an example of graph to which we apply those results, recovering the full SE equations of 
\cite{manoel2017multi,aubin2020spiked}: consider any instance of the family of AMP iterations presented in Section \ref{sec:graph-based-amp}, indexed on a given oriented graph $G = (V,E)$, i.e.
\begin{align}
    \label{eq:AMP_teach1}
    \mathbf{x}^{t+1}_{\overrightarrow{e}} &= \hat{\bA}_{\overrightarrow{e}} \mathbf{m}^t_{\overrightarrow{e}} - b^t_{\overrightarrow{e}} \mathbf{m}^{t-1}_{\overleftarrow{e}} \, , \\
    \mathbf{m}^t_{\overrightarrow{e}} &= \tilde{f}^t_{\overrightarrow{e}}\left(\left(\mathbf{x}^t_{\overrightarrow{e}'}\right)_{\overrightarrow{e}':\overrightarrow{e}' \to \overrightarrow{e}}\right)\label{eq:AMP_teach2} \, , 
\end{align}
where, in the notation of Lemma \ref{lemma:reduction}, for any symmetric edge $\overrightarrow{e}$ from the set $\left\{\overrightarrow{e}_{1},...,\overrightarrow{e}_{l}\right\}$, $\hat{\mathbf{A}}_{\overrightarrow{e}} = \mathbf{A}_{\overrightarrow{e}}+\frac{1}{N}\mathbf{v}_{\overrightarrow{e}}\mathbf{v}_{\overrightarrow{e}}^{\top}$, and $\tilde{f}^{t}_{\overrightarrow{e}}(.) = f^{t}_{\overrightarrow{e}}(.)$. Furthermore, for any asymmetric edge 
$\overrightarrow{e}$ from the set $\left\{\overrightarrow{e}_{l+1},...,\overrightarrow{e}_{m}\right\}$, $\hat{\mathbf{A}}_{\overrightarrow{e}} = \mathbf{A}_{\overrightarrow{e}}$ and $\tilde{f}^{t}_{\overrightarrow{e}}(.) = f^{t}(\varphi_{\overrightarrow{e}}(\mathbf{A}_{\overrightarrow{e}}\mathbf{w}_{\overrightarrow{e}}),.)$. The following lemma then gives the SE equations for this iteration:
\begin{lemma}
    \label{lemma:teacher_SE}
    Assume that \ref{ass:main1}-\ref{ass:main7} are verified. Further assume that, for any $\overrightarrow{e} \in \overrightarrow{E}$, $\frac{1}{\sqrt{N}}\norm{\mathbf{v}_{\overrightarrow{e}}}_{F}$ and $\frac{1}{\sqrt{N}}\norm{\mathbf{w}_{\overrightarrow{e}}}_{F}$ converge to finite constants as $N \to \infty$.
    For any symmetric edge $\overrightarrow{e}$ from the set $\left\{\overrightarrow{e}_{1},...,\overrightarrow{e}_{l}\right\}$, define the following SE recursion:
\begin{align}
    \boldsymbol{\mu}_{\overrightarrow{e}}^{0}, \thickspace \boldsymbol{\kappa}_{\overrightarrow{e}}^{1,1} &= \lim_{N \to \infty} \frac{1}{N}  f^{0}_{\overrightarrow{e}}(\left(\mu^{0}_{\overrightarrow{e}'}\mathbf{v}_{\overrightarrow{e}'}+\bx^{0}_{\overrightarrow{e}'}\right)_{\overrightarrow{e}':\overrightarrow{e}'\to \overrightarrow{e}})^{\top}f^{0}_{\overrightarrow{e}}(\left(\mu^{0}_{\overrightarrow{e}'}\mathbf{v}_{\overrightarrow{e}'}+\bx^{0}_{\overrightarrow{e}'}\right)_{\overrightarrow{e}':\overrightarrow{e}'\to \overrightarrow{e}})\\
    \boldsymbol{\mu}^{s+1}_{\overrightarrow{e}} &= \lim_{N \to +\infty} \frac{1}{N}\mathbb{E}\left[(\mathbf{v}_{\overrightarrow{e}})^{\top}f^{s}_{\overrightarrow{e}}(\left(\mu^{s}_{\overrightarrow{e}'}\mathbf{v}_{\overrightarrow{e}'}+\bx^{s}_{\overrightarrow{e}'}\right)_{\overrightarrow{e}':\overrightarrow{e}'\to \overrightarrow{e}})\right]\\
    \boldsymbol{\kappa}^{t+1, s+1}_{\overrightarrow{e}} &= \boldsymbol{\kappa}^{s+1, t+1}_{\overrightarrow{e}} = \lim_{N \to \infty} \frac{1}{N} \E\left[ f^{s}_{\overrightarrow{e}}(\left(\mu^{s}_{\overrightarrow{e}'}\mathbf{v}_{\overrightarrow{e}'}+\bx^{s}_{\overrightarrow{e}'}\right)_{\overrightarrow{e}':\overrightarrow{e}'\to \overrightarrow{e}})^\top f^{t}_{\overrightarrow{e}}(\left(\mu^{t}_{\overrightarrow{e}'}\mathbf{v}_{\overrightarrow{e}'}+\bx^{t}_{\overrightarrow{e}'}\right)_{\overrightarrow{e}':\overrightarrow{e}'\to \overrightarrow{e}}) \right], \notag \\ 
    &\qquad s \in \{ 0, \dots, t \} \, .
\end{align}
where $(\bZ^1_{\overrightarrow{e}}, \dots, \bZ^t_{\overrightarrow{e}})$ is a centered Gaussian random vector of covariance $(\boldsymbol{\kappa}^{r,s}_{\overrightarrow{e}})_{r,s \leq t} \otimes \mathbf{I}_{n_{w}}$. Then, for any sequence of uniformly (in n) pseudo-Lipschitz function $\Phi : \mathbb{R}^{(t+1)n_{w}} \to \mathbb{R}$ :
\begin{equation}
    \Phi\left(\left(\mathbf{x}_{\overrightarrow{e}}^{s}\right)_{0\leqslant s \leqslant t, \overrightarrow{e}\in \overrightarrow{E}_{sym}}\right) \approxP \mathbb{E}\left[\Phi\left((\mu^{s}_{\overrightarrow{e}}\mathbf{v}_{\overrightarrow{e}}+\bZ^{s}_{\overrightarrow{e}})_{0\leqslant s \leqslant t, \overrightarrow{e}\in \overrightarrow{E}_{sym}}\right)\right]
\end{equation}
For any asymmetric edge $\overrightarrow{e}$ from the set $\left\{\overrightarrow{e}_{l+1},...,\overrightarrow{e}_{m}\right\}$, define the following SE recursion :
\begin{align}
    &\boldsymbol{\nu}_{\overrightarrow{e}}^{0}, \hat{\boldsymbol{\nu}}_{\overrightarrow{e}}^{0}, \boldsymbol{\kappa}_{\overrightarrow{e}}^{1,1} = \frac{1}{N}f^{0}_{\overrightarrow{e}}((\bx^{0}_{\overrightarrow{e}'})_{\overrightarrow{e}':\overrightarrow{e}'\to \overrightarrow{e}})^{\top}f^{0}_{\overrightarrow{e}}((\bx^{0}_{\overrightarrow{e}'})_{\overrightarrow{e}':\overrightarrow{e}'\to \overrightarrow{e}}) \\
    &\boldsymbol{\nu}^{t+1}_{\overrightarrow{e}} = \lim_{N \to \infty} \frac{1}{N}\mathbb{E}\left[\mathbf{w}_{\overrightarrow{e}}^{\top}f_{\overrightarrow{e}}^{t}\left(\varphi_{\overrightarrow{e}}(\mathbf{z}_{\mathbf{w}_{\overrightarrow{e}}}), \left(\mathbf{z}_{\mathbf{w}_{\overrightarrow{e}'}}\rho_{\mathbf{w}_{\overrightarrow{e}'}}^{-1}\boldsymbol{\nu}_{\overrightarrow{e}'}^{t}+\mathbf{w}_{\overleftarrow{e}'}\hat{\boldsymbol{\nu}}_{\overrightarrow{e}'}^{t}+\mathbf{Z}_{\overrightarrow{e}'}^{t}\right)_{\overrightarrow{e}':\overrightarrow{e}'\to \overrightarrow{e}}\right)\right] \\
    &\hat{\boldsymbol{\nu}}^{t+1}_{\overrightarrow{e}} = \lim_{N \to \infty} \frac{1}{N}\mathbb{E}\left[\sum_{i=1}^{N}\frac{\partial f_{\overrightarrow{e},i}^{t}}{\partial \mathbf{z}_{\mathbf{w}_{\overrightarrow{e}},i},\varphi_{\overrightarrow{e}}}\left(\varphi_{\overrightarrow{e}}(\mathbf{z}_{\mathbf{w}_{\overrightarrow{e}}}), \left(\mathbf{z}_{\mathbf{w}_{\overrightarrow{e}'}}\rho_{\mathbf{w}_{\overrightarrow{e}'}}^{-1}\boldsymbol{\nu}_{\overrightarrow{e}'}^{t}+\mathbf{w}_{\overleftarrow{e}'}\hat{\boldsymbol{\nu}}_{\overrightarrow{e}'}^{t}+\mathbf{Z}_{\overrightarrow{e}'}^{t}\right)_{\overrightarrow{e}':\overrightarrow{e}'\to \overrightarrow{e}}\right)\right] \\
    &\boldsymbol{\kappa}_{\overrightarrow{e}}^{t+1, s+1} = \boldsymbol{\kappa}_{\overrightarrow{e}}^{s+1, t+1} = \notag \\
    &\lim_{N \to \infty} \frac{1}{N}\mathbb{E}\bigg[\left(f_{\overrightarrow{e}}^{s}\left(\varphi_{\overrightarrow{e}}(\mathbf{z}_{\mathbf{w}_{\overrightarrow{e}}}), \left(\mathbf{z}_{\mathbf{w}_{\overrightarrow{e}'}}\rho_{\mathbf{w}_{\overrightarrow{e}'}}^{-1}\mathbf{m}_{\overrightarrow{e}'}^{s}+\mathbf{w}_{\overleftarrow{e}'}\hat{\boldsymbol{\nu}}_{\overrightarrow{e}'}^{s}+\mathbf{Z}_{\overrightarrow{e}'}^{s}\right)_{\overrightarrow{e}':\overrightarrow{e}'\to \overrightarrow{e}}\right)-\mathbf{w}_{\overrightarrow{e}}\rho_{\mathbf{w}_{\overrightarrow{e}}}^{-1}\boldsymbol{\nu}_{\overrightarrow{e}}^{s+1}\right)^{\top} \notag \\
    &\left(f_{\overrightarrow{e}}^{t}\left(\varphi_{\overrightarrow{e}}(\mathbf{z}_{\mathbf{w}_{\overrightarrow{e}}}), \left(\mathbf{z}_{\mathbf{w}_{\overrightarrow{e}'}}\rho_{\mathbf{w}_{\overrightarrow{e}'}}^{-1}\boldsymbol{\nu}_{\overrightarrow{e}'}^{t}+\mathbf{w}_{\overleftarrow{e}'}\hat{\boldsymbol{\nu}}_{\overrightarrow{e}'}^{t}+\mathbf{Z}_{\overrightarrow{e}'}^{t}\right)_{\overrightarrow{e}':\overrightarrow{e}'\to \overrightarrow{e}}\right)-\mathbf{w}_{\overrightarrow{e}}\rho_{\mathbf{w}_{\overrightarrow{e}}}^{-1}\boldsymbol{\nu}_{\overrightarrow{e}}^{t+1}\right)\bigg]
\end{align}
where $(\bZ^1_{\overrightarrow{e}}, \dots, \bZ^t_{\overrightarrow{e}})$ is a centered Gaussian random vector of covariance $(\boldsymbol{\kappa}^{r,s}_{\overrightarrow{e}})_{r,s \leq t} \otimes \mathbf{I}_{n_{w}}$. Then, for any sequence of uniformly (in n) pseudo-Lipschitz function $\Phi : \mathbb{R}^{(t+1)n_{w}} \to \mathbb{R}$ :
\begin{equation}
    \Phi\left(\left(\mathbf{x}_{\overrightarrow{e}}^{s}\right)_{0\leqslant s \leqslant t, \overrightarrow{e}\in \overrightarrow{E}_{sym}}\right) \approxP \mathbb{E}\left[\Phi\left((\mathbf{z}_{\mathbf{w}_{\overrightarrow{e}}}\rho_{\mathbf{w}_{\overrightarrow{e}}}^{-1}\boldsymbol{\nu}_{\overrightarrow{e}}^{t}+\mathbf{w}_{\overleftarrow{e}}\hat{\boldsymbol{\nu}}_{\overrightarrow{e}}^{t}+\mathbf{Z}_{\overrightarrow{e}}^{t})_{0\leqslant s \leqslant t, \overrightarrow{e}\in \overrightarrow{E}_{asym}}\right)\right]
\end{equation}
\end{lemma}
Note the dependence on $\mathbf{w}_{\overleftarrow{e}}$ of the SE quantities indexed by $\overrightarrow{e}$, which comes from evaluating the matrix products defining the terms in $\mathbf{m}^{t}, \hat{\mathbf{m}}^{t}$. In the AMP litterature, non-linearities often take the form $\tilde{f}^{t}_{\overrightarrow{e}}(.) = f^{t}(\varphi_{\overrightarrow{e}}(\mathbf{A}_{\overleftarrow{e}}\mathbf{w}_{\overrightarrow{e}}),.)$, i.e. with a dependence on the random matrix of the opposite edge. This only changes 
the arrows in $\mathbf{W}_{0}$ i.e.
\begin{equation}
    \mathbf{W}_{0} = 
        \begin{pmatrix}
        0 & & & & & & & \\
        & \ddots& & & & & 0 & \\
        & & 0 & & & & & \\
        & & & 0 & \mathbf{w}_{\overrightarrow{e}_{l+1}} & & & \\
        & & & \mathbf{w}_{\overleftarrow{e}_{l+1}} & 0 & & & \\
        & & & & & \ddots & & \\
        & 0 & & & &  & 0 & \mathbf{w}_{\overrightarrow{e}_{m}} \\
        & & & & & & \mathbf{w}_{\overleftarrow{e}_{m}} & 0 \\
        \end{pmatrix},
\end{equation}
and the corresponding arrows in the SE equations above. It is indeed what is observed in, e.g. \cite{rangan2011generalized,manoel2017multi,aubin2020spiked}.
Examples are given throughout Section \ref{sec:applications}.

\paragraph{Structured and correlated matrices.}
Products of Gaussian matrices can be considered by choosing identities as non-linearities on given edges of the graph. This was done heuristically in \cite{manoel2017multi} to study structured inference problems. Gaussian matrices with generic covariances can also be considered, i.e., $\bA = \bZ\boldsymbol{\Sigma}^{1/2}$ where $\bZ$ is an i.i.d. $\mathbf{N}(0,\frac{1}{d})$ matrix and $\boldsymbol{\Sigma} \in \mathbb{R}^{d\times d}$ is a positive definite matrix. Indeed, the covariance matrix can be absorbed in the non-linearity as a non-separable component. Depending on the non-linearity, expressions may simplify as functions of the spectral distribution of $\boldsymbol{\Sigma}$. Examples are given in Section \ref{subsec:corr_mat}.

\paragraph{Spatial coupling.}
\emph{Spatial coupling} was introduced and studied in \cite{krzakala2012probabilistic,krzakala2012probabilistic,javanmard2013state,donoho2013information} as a mean to reach information theoretic limits in compressed sensing. The idea is to write the state evolution equations when the random matrices have a block structure of the form 
\begin{equation*}
    \bA = \begin{bmatrix}\bA_{11}& \bA_{12} & \dots &\bA_{1l}\\ \bA_{21}& \bA_{22}& \dots & \bA_{2l} \\ \vdots & \vdots & \ddots & \vdots \\ \bA_{k1}& \bA_{k2} & \dots &\bA_{kl}\end{bmatrix} \in \mathbb{R}^{N \times d} \, , 
\end{equation*}
each $\bA_{ij} \in \mathbb{R}^{N_{i} \times d_{j}}$ has i.i.d. $\mathbf{N}(0,\frac{\sigma_{ij}}{d})$ entries and $N_{i}/N, d_{j}/d$ are constant aspect ratios, where $\sum_{i} N_{i} = N$ and $\sum_{j} d_{j} = d$. The proof of SE equations with this kind of matrices was proposed in \cite{javanmard2013state} and relies on a matrix-valued symmetric AMP iteration similar to the one used in our proof, with a family of non-linearities acting on blocks of variables, with a separable effect on each block. Since our proof extends the matrix-valued, symmetric AMP iteration to the fully non-separable case, the same ideas can be applied to our framework to include spatially coupled matrices on each edge of the oriented graph presented in the previous section (with the added possibility of non-separable effects on each block). We now give an example in Section \ref{subsec:spat_coup}.
\section{Applications to inference problems}
\label{sec:applications}
In this section we illustrate our main theorem by showing how several AMP iterations established heuristically in the literature are included in our framework, in particular \cite{manoel2017multi,aubin2019committee,aubin2020spiked,loureiro2021learning}, and how straightforward generalizations can be considered.
We adopt an optimization viewpoint for each problem, omitting the probabilistic inference formulation at the origin of these iterations for simplicity.
\subsection{A building block: AMP for generalized linear models}
\label{subsec:GAMP}
We start with a known AMP iteration for which the state evolution equations were already proven, and build upon the intuition it gives to present more elaborate iterations. 
Consider the task of optimizing a penalized cost functions of the form
\begin{equation}
\label{opti_GLM}
 \hat{\mathbf{x}}\in   \min_{\mathbf{x} \in \mathbb{R}^{d}} g(\bA\mathbf{x},\mathbf{y})+f(\mathbf{x})
\end{equation}
where the vector of labels $\mathbf{y}$ is typically assumed to be generated from another process as
\begin{equation*}
    \mathbf{y} = \phi(\bA\mathbf{x}_{0}),
\end{equation*}
with $\mathbf{x}_{0} \in \mathbb{R}^{d}$ generated from a given distribution $p_{\mathbf{x}_{0}}$ independent from the matrix $\bA$, $\bA \in \mathbb{R}^{N \times d}$ is a matrix with i.i.d. $\mathbf{N}(0,\frac{1}{d})$ elements, and $\phi$ a given function.
The goal is then to reconstruct the vector $\mathbf{x}_{0}$. This formulation is at the basis of many of the fundamental estimation methods in machine learning: least-squares, LASSO, logistic regression, etc. Approximate-message passing algorithms were proposed for this task, notably in $\cite{donoho2009message,bayati2011dynamics,rangan2011generalized,krzakala2012probabilistic,javanmard2013state}$, and take the generic form of the asymmetric AMP iteration \eqref{eq:asy_amp}
where $\bA_{\overrightarrow{e}} = \bA$.
Intuitively, the functions $f^{t}_{\overrightarrow{e}}, f^{t}_{\overleftarrow{e}}$ each correspond to one of the functions $g,f$ from \eqref{opti_GLM} and respectively output an estimate of the quantities $\bA\hat{\mathbf{x}}, \hat{\mathbf{x}}$. As prescribed by the form of the generative model, we expect the update function associated to the loss $g(.,\mathbf{y})$ to be correlated with the matrix $\mathbf{A}$, thus preventing a direct application of the SE equations of Theorem \ref{thm:graph-AMP}, and 
requiring the results of Lemma\ref{lemma:teacher_SE}.
\subsection{Multilayer generalized linear estimation}
\label{subsec:MLAMP}
Consider now the problem of recovering a vector $\mathbf{x}_{0}$ from a more complex generative model involving a multilayer neural network with random weights:
\begin{equation*}
    \mathbf{y} = \phi_{L}(\bA_{L}\phi_{L-1}\left(\bA_{L-1}(...\phi_{1}(\bA_{1}\mathbf{x}_{0}))\right)
\end{equation*}
where one has access to the final output $\mathbf{y}$ and would like to reconstruct the intermediate ones and input $\mathbf{x}_{0}$. For each layer $1\leqslant l \leqslant L$ the matrix $\bA_{l} \in \mathbb{R}^{N_{l+1} \times N_{l}}$ has i.i.d. $\mathbf{N}(0,\frac{1}{N_{l}})$ with $N_{l+1}/N_{l} = \delta_{l}$.
The idea is to solve this sequentially using asymmetric AMP iterations similar to the one presented in the previous section. This approach was originally proposed in \cite{manoel2017multi} under the name multilayer AMP (MLAMP). 
For any $1 \leqslant l \leqslant L+1$, define 
\begin{align*}
    \mathbf{x}_{l} &= \phi_{l-1}(\bA_{l-1}\phi_{l-2}(..\phi_{1}(\bA_{1}\mathbf{x}_{0}))), \\
    \mbox{such that} \quad\mathbf{x}_{l+1} &= \phi_{l}(\bA_{l}\mathbf{x}_{l}) \quad \mbox{and} \quad \mathbf{x}_{L+1} = \mathbf{y}
\end{align*}
The intuition is the following : each $\mathbf{x}_{l}$ is then estimated using the asymmetric AMP corresponding to the problem
\begin{equation*}
    \hat{\mathbf{x}}_{l} = \argmin_{\mathbf{x} \in \mathbb{R}^{N_{l}}} \thickspace g_{l}(\bA_{l}\mathbf{x},\mathbf{y}_{l})+f_{l}(\mathbf{x})
\end{equation*}
the output of which is used to estimate the next, i.e., $\mathbf{y}_{l} = \hat{\mathbf{x}}_{l+1}$, whose statistical properties are given by the SE equations. The complete derivation of the iteration involves writing the belief-propagation equations on the factor graph corresponding to the multilayer inference problem, capturing all the interactions between the different iterates. These SE equations were derived heuristically in \cite{manoel2017multi} for Bayes-optimal inference, and this paper proves them in the generic case.
\subsection{Spiked matrix with generative prior}
\label{subsec:spikeML}
In the same spirit as the composition of generalized linear models defining MLAMP, different tasks can be composed to obtain richer instances of inference problems. For instance in \cite{aubin2019committee}, the reconstruction of a low-rank matrix under a generative prior is considered using an AMP iteration. A rank-one matrix is observed, blurred by Gaussian noise:
\begin{align*}
    \mathbf{Y} = \sqrt{\frac{\lambda}{d}}\mathbf{v}_{0}\mathbf{v}_{0}^{\top}+\mathbf{W}
\end{align*}
where $\mathbf{W} \in GOE(N)$, and the vector $\mathbf{v}_{0}\in \mathbb{R}^{N}$ is assumed to be generated from a multilayer neural network with random weights
\begin{equation*}
    \mathbf{v}_{0} = \phi_{L}(\bA_{L}\phi_{L-1}\left(\bA_{L-1}(...\phi_{1}(\bA_{1}\mathbf{x}_{0}))\right)
\end{equation*}
for a given ground truth vector $\mathbf{x}_{0} \in \mathbb{R}^{N_{1}}$, matrices $\{\bA_{l} \in \mathbb{R}^{N_{l+1} \times N_{l}}\}_{1 \leqslant l \leqslant L}$ and non-linearities $\{\phi_{l}\}_{1 \leqslant l \leqslant L}$.
The AMP iteration to estimate $\mathbf{v}_{0}$ from $\mathbf{Y}$ was first proposed in \cite{rangan2012iterative,deshpande2014information}, and takes the form of a symmetric AMP \eqref{eq:sym_amp}. Similarly to MLAMP, the output of this iteration can then be used as input, leading to the AMP iteration proposed in \cite{aubin2020spiked}, which corresponds to the AMP iteration \eqref{eq:spike_mlamp}. This paper proves the state evolution equations for this iteration.
\subsection{An example with matrix-valued variables}
\label{subsec:mat_val}
Matrix valued variables are encountered in scenarios such as committee machines \cite{aubin2019committee} or multiclass learning problems \cite{loureiro2021learning}, or more generically when a finite ensemble of predictors is learned. Consider the matrix-valued extension of the generalized linear estimation problem Eq.\eqref{opti_GLM}.
\begin{align*}
&\hat{\mathbf{X}} \in \argmin_{\mathbf{X} \in \mathbb{R}^{N \times q}} g(\bA\mathbf{X},\mathbf{Y})+f(\mathbf{X}) \\
    &\mbox{where} \quad \mathbf{Y} = \phi(\bA\mathbf{X}_{0}))
\end{align*}
where $\mathbf{X}_{0} \in \mathbb{R}^{N\times q}$ and $q \in \mathbb{N}$ is kept finite. The SE equations for the asymmetric AMP with matrix valued-variables are included in the result of  \cite{javanmard2013state}.
This can be directly generalized to a multilayer matrix inference problem by considering a generative model of the form
\begin{equation*}
    \mathbf{Y} = \phi_{L}(\bA_{L}\phi_{L-1}\left(\bA_{L-1}(...\phi_{1}(\bA_{1}\mathbf{X}_{0}))\right)
\end{equation*}
and successive application of the matrix-valued asymmetric AMP as proposed for MLAMP in Section \ref{subsec:MLAMP}. The state evolution equations for this problem is included in our framework using the results from Section \ref{sec:extensions}.
\subsection{An example with structured random matrices}
\label{subsec:corr_mat}
Consider a generalized linear inference task where the data is now represented by a Gaussian matrix with a covariance $\boldsymbol{\Sigma} \neq \mathbf{I}_{d}$. This can be dealt with using the non-separable framework.
Assuming the covariance matrix is full-rank, we can equivalently work with the variable $\tilde{\mathbf{x}} = \boldsymbol{\Sigma}^{1/2}\mathbf{x}$, and solve
\begin{equation*}
    \argmin_{\tilde{\mathbf{x}}} g(\tilde{\mathbf{A}}\tilde{\mathbf{x}},\mathbf{y})+f(\boldsymbol{\Sigma}^{-1/2}\tilde{\mathbf{x}}).
\end{equation*}
where $\tilde{A}$ is now ana i.i.d. Gaussian matrix.
This will modify the update function associated to $f$, becoming $f(\boldsymbol{\Sigma}^{-1/2}.)$,
which is non-separable, even if the function $f$ is initially assumed to be separable. The validity of the SE equations for this case follows from the results of \cite{berthier2020state}. This manipulation can also be done on any layer of MLAMP, for a given set of covariance matrices $\boldsymbol{\Sigma}_{1},...,\boldsymbol{\Sigma_{L}}$ associated to each random matrix $\bA_{1},...,\bA_{L}$, with vector or matrix-valued variables. The validity of the SE equations in this case follows from the results of this paper. In the convex GLM case (2-layer), the fixed point of the state evolution equations with a generic covariance gives the same result as (a particular case of) the exact asymptotics recently proposed in \cite{loureiro2021capturing} to study different feature maps in generalized linear models.
\subsection{An example of spatial coupling with non-separable non-linearities}
\label{subsec:spat_coup}
Here we briefly describe an inference problem recently studied in \cite{loureiro2021learning} that can be solved using spatial coupling on a non-separable AMP iteration. Consider the problem of classifying a high-dimensional Gaussian mixture with a finite number $K$ of clusters, described by the joint density
\begin{equation*}
    P(\mathbf{x}\vert \mathbf{y}) = \sum_{k=1}^{K}y_{k} \pi_{k}\mathbf{N}(\boldsymbol{\mu}_{k},\boldsymbol{\Sigma}_{k})
\end{equation*}
where $\mathbf{x} \in \mathbb{R}^{d}$ is a sample, $\mathbf{y} \in \mathbb{R}^{K}$ is a binary label vector, $\{\pi_{k}\}_{k}$ are the cluster probabilities such that $\sum_{k=1}^{K} \pi_{k} = 1$, $\{\boldsymbol{\mu}_{k}\}_{1 \leqslant k \leqslant K}$ are the means and $\{\boldsymbol{\Sigma}_{k}\}_{1 \leqslant k \leqslant K}$ are positive definite covariances,
using a convex generalized linear model, i.e.,
\begin{equation*}
    \mathbf{X} \in \argmin_{\mathbf{X} \in \mathbb{R}^{d \times K}} g(\mathbf{A}\mathbf{X},\mathbf{Y})+f(\mathbf{X})
\end{equation*}
where $\mathbf{Y} \in \mathbb{R}^{N \times K}$ is the concatenated matrix of one-hot encoded labels.
The matrix $\mathbf{A}$ representing $N$ samples of the Gaussian mixture can be written as a block diagonal matrix
\begin{equation*}
    \bA = \begin{bmatrix}\bZ_{1}\boldsymbol\Sigma_{1}^{1/2}&&&&\\ & \bZ_{2}\boldsymbol{\Sigma_{2}}^{1/2}&&&\\ &&...&&\\ &&&&\bZ_{K}\boldsymbol{\Sigma_{K}}^{1/2}\end{bmatrix} \in \mathbb{R}^{N \times Kd}
\end{equation*}
where the $\mathbf{Z}_{k} \in \mathbb{R}^{N_{k} \times d}$ are i.i.d. $\mathbf{N}(0,\frac{1}{d})$ independent matrices, with $N_{k}$ the number of samples coming from each cluster.
This type of matrix can be embedded into an AMP iteration using the spatial coupling technique to handle the block structure and the non-separable framework to deal with the covariances on each block. The validity of the SE equations for the combination of spatial coupling and non-separable effects is proven by this paper. This is also an example where the teacher distribution is independent of the Gaussian matrices that will appear in the AMP iteration, as the multinomial distribution prescribing cluster membership is independent of the Gaussian cloud of each cluster.
\section{Perspectives}
\label{sec:perspectives}
\paragraph{}
We have shown that AMP algorithms can be unified in an intuitive way by means of an oriented graph, and that this representation leads to a modular, effective and extended proof of state evolution equations. Several problems follow from the results presented here.

\paragraph{Connecting back to the factor graph.} We do not relate our proposed graphical representation of the AMP iterations with the factor graphs of the probabilistic inference problems that generated them. Understanding this relation would clarify the statistical inference problems that can be solved using AMP iterations. The applications that motivated this paper use our framework with only very simple graphs---line graphs, sometimes with a loop. However, the framework accepts much more complicated graphs, potentially with more loops. In future work, we hope to explore the new statistical problems and AMP iterations that can be analyzed using these graphs. 

\paragraph{Rotationally invariant matrices.}
As shown in \cite{rangan2019vector,fletcher2018inference,pandit2020inference,fan2020approximate}, the Gaussian conditioning method at the core of AMP proofs can be reproduced with right rotationally invariant matrices with generic spectrum. Extending the results of the present paper to this family of matrices requires finding the appropriate form of the graph iteration and is an open problem.

\paragraph{Universality and finite size corrections.}
State evolution proofs are amenable to both finite size analysis \cite{rush2018finite,ma2017analysis} and universality proofs \cite{bayati2015universality,chen2021universality}.  Although both problems were tackled in simpler settings in these papers, their techniques could be combined with the embedding proposed in the proof of Theorem \ref{thm:graph-AMP} to prove finite size rates and universality properties for any graph supported AMP.
\section{Acknowledgments}
We thank Florent Krzakala and Lenka Zdeborova for suggesting this project, insightful discussions, and for organizing the Ecole des Houches Summer Workshop on Statistical Physics and Machine Learning where this work was initiated. Rapha\"el Berthier acknowledges support from the DGA. 
\bibliographystyle{alpha}
\bibliography{References}
\newpage
\appendix
\etocdepthtag.toc{mtappendix}
\etocsettagdepth{mtmain}{none}
\etocsettagdepth{mtappendix}{subsection}
\tableofcontents
\section{Changing time indices}
\label{app:time}
Here we show how the time index convention usually encountered in earlier instances of the asymmetric AMP iteration can be recovered from the one used in this proof. Consider two successive iterations of the asymmetric AMP \eqref{eq:asy_amp}:
\begin{align}
\begin{split}
\begin{aligned}
    \mathbf{x}^{t+1}_{\overrightarrow{e}} &= \mathbf{A}_{\overrightarrow{e}} \mathbf{m}^t_{\overrightarrow{e}} - b^t_{\overrightarrow{e}} \mathbf{m}^{t-1}_{\overleftarrow{e}} \, ,  && \mathbf{x}^{t}_{\overrightarrow{e}} = \mathbf{A}_{\overrightarrow{e}} \mathbf{m}^{t-1}_{\overrightarrow{e}} - b^{t-1}_{\overrightarrow{e}} \mathbf{m}^{t-2}_{\overleftarrow{e}} \, ,  \\
    &\bm^t_{\overrightarrow{e}} = f^t_{\overrightarrow{e}}\left(\bx^t_{\overleftarrow{e}}\right) \, , && \hspace{0.7cm}\bm^{t-1}_{\overrightarrow{e}} = f^{t-1}_{\overrightarrow{e}}\left(\bx^{t-1}_{\overleftarrow{e}}\right) \, , \\
        \bx^{t+1}_{\overleftarrow{e}} &= \bA_{\overrightarrow{e}}^\top\bm^t_{\overleftarrow{e}} - b^t_{\overleftarrow{e}} \bm^{t-1}_{\overrightarrow{e}} \, ,  && \bx^{t}_{\overleftarrow{e}} = \bA_{\overrightarrow{e}}^\top\bm^{t-1}_{\overleftarrow{e}} - b^{t-1}_{\overleftarrow{e}} \bm^{t-2}_{\overrightarrow{e}} \, ,  \\
    &\bm^t_{\overleftarrow{e}} = f^t_{\overleftarrow{e}}\left(\bx^t_{\overrightarrow{e}}\right) && \hspace{0.7cm} \bm^{t-1}_{\overleftarrow{e}} = f^{t-1}_{\overleftarrow{e}}\left(\bx^{t-1}_{\overrightarrow{e}}\right) \\ 
\end{aligned}
 \end{split}
\end{align}
which requires initializing both $\mathbf{x}_{\overrightarrow{e}}$ and $\mathbf{x}_{\overleftarrow{e}}$, and updates them simultaneously at each iteration.
We see that to evaluate $\mathbf{x}^{t+1}_{\overrightarrow{e}}$ (resp. $\mathbf{x}^{t+1}_{\overleftarrow{e}}$), we only need the previous value of $\bx^{t}_{\overleftarrow{e}}$ (resp. $\bx^{t}_{\overrightarrow{e}}$) and $\bx_{\overrightarrow{e}^{t-1}}$ (resp. $\bx^{t-1}_{\overleftarrow{e}}$). Thus only half of the iterates can be computed, independently of the other half, using the following formulae (setting the other update functions to zero):
\begin{align}
\begin{split}
\bx^{2t+1}_{\overleftarrow{e}} &= \bA_{\overrightarrow{e}}^\top\bm^{2t}_{\overleftarrow{e}} - b^{2t}_{\overleftarrow{e}} \bm^{2t-1}_{\overrightarrow{e}} \, ,  \\
    &\bm^{2t}_{\overleftarrow{e}} = f^{2t}_{\overleftarrow{e}}\left(\bx^{2t}_{\overrightarrow{e}}\right) \, , \\
    \mathbf{x}^{2t}_{\overrightarrow{e}} &= \mathbf{A}_{\overrightarrow{e}} \mathbf{m}^{2t-1}_{\overrightarrow{e}} - b^{2t-1}_{\overrightarrow{e}} \mathbf{m}^{2t-2}_{\overleftarrow{e}} \, ,   \\
    &\bm^{2t-1}_{\overrightarrow{e}} = f^{2t-1}_{\overrightarrow{e}}\left(\bx^{2t-1}_{\overleftarrow{e}}\right) \, 
    \end{split}
\end{align}
which only requires one value at initialization and at each iteration. The usual time indices found in , e.g., \cite{berthier2020state} are then recovered with the following mapping:
\begin{align*}
    \bx^{2t+1}_{\overleftarrow{e}} &= \bu^{t+1} \\
    \bx^{2t}_{\overrightarrow{e}} &= \bv^{t} \\
    f^{2t}_{\overleftarrow{e}}(.) &= g_{t}(.) \\
    f^{2t-1}_{\overrightarrow{e}}(.) &= e_{t}(.)
\end{align*}
Note that this simplification is specific to the graph structure underlying the asymmetric AMP iteration.
\section{Matrix-valued symmetric AMP iterations with non-separable non-linearities}
\label{SE_sym_AMP}
\subsection{State evolution description}

In this section, we present the state evolution equations for a symmetric AMP iteration with non-separable non-linearities and matrix-valued variables. This is an extension of the results of \cite{javanmard2013state,berthier2020state}. This result underlies the proof of state evolution equations for graph-based AMP iterations.

Consider an initial (deterministic) matrix $\bX^{0} \in \mathbb{R}^{N \times q}$ and a sequence of deterministic functions $\{f^{t} : \mathbb{R}^{N \times q} \to \mathbb{R}^{N \times q}\}_{t \in \mathbb{N}}$. For the reader's convenience, we recall here the symmetric AMP iteration \eqref{eq:sym-amp-iteration-1}-\eqref{eq:sym-amp-iteration-2}.

\paragraph{Symmetric AMP iteration.} Let $\bX^0 \in \R^{N\times q}$ and define recursively,
\begin{align}
\label{eq:sym_AMP1}
\bX^{t+1} &= \mathbf{A}\bM^{t}-\bM^{t-1}(\bb^{t})^{\top} && \in \R^{N\times q} \, ,  \\
\bM^{t} &=f^{t}(\bX^{t}) && \in \R^{N\times q} \, , \\
\bb^t &= \frac{1}{N} \sum_{i=1}^N \frac{\partial f^t_i}{\partial \bX_i}(\bX^t) && \in 
\R^{q\times q}\, . 
\label{eq:sym_AMP2}
\end{align}
where $\bb^{t}$ is the Onsager correction term.
\noindent 
We now list the necessary assumptions.
\paragraph{Assumptions.}
\begin{enumerate}[font={\bfseries},label={(B\arabic*)}]
\item\label{it:ass-sym-1} $\mathbf{A} \in \mathbb{R}^{N \times N}$ is a GOE(N) matrix, i.e., $\mathbf{A} = \mathbf{G}+\mathbf{G}^{\top}$ for $\mathbf{G} \in \mathbb{R}^{N \times N}$ with i.i.d.~entries $G_{ij} \sim \mathbf{N}(0,1/(2N))$. 
\item For each $t \in \mathbb{N}, f^{t} : \mathbb{R}^{N \times q} \to \mathbb{R}^{N \times q}$ is pseudo-Lipschitz of order $k$, uniformly in $N$.
\item \label{it:ass-sym-3} $\norm{\bX^{0}}_{F}/\sqrt{N}$
converges to a finite constant as $N \to \infty$.
\item\label{it:ass-sym-4}  The following limit exists and is finite:
\begin{equation}
    \lim_{N \to \infty} \frac{1}{N}  f^{0}(\bX^{0})^{\top}f^{0}(\bX^{0}) \in \mathbb{R}^{q \times q}
\end{equation}
\item\label{it:ass-sym-5} For any $t \in \mathbb{N}_{>0}$ and any $\boldsymbol{\kappa} \in \cS_{q}^{+}$, the following limit exists and is finite:
\begin{equation}
    \lim_{N \to \infty} \frac{1}{N}  \mathbb{E}\left[f^{0}(\bX^{0})^{\top}f^{t}(\bZ)\right] \in \mathbb{R}^{q \times q}
\end{equation}
where $\mathbf{Z} \in \mathbb{R}^{N \times q}$, $\mathbf{Z} \sim \mathbf{N}\left(0,\boldsymbol{\kappa} \otimes \mathbf{I}_{N} \right)$.
\item\label{it:ass-sym-6} For any $s,t \in \mathbb{N}_{>0}$ and any $\boldsymbol{\kappa} \in \cS_{2q}^{+}
$, the following limit exists and is finite:
\begin{equation}
    \lim_{N \to \infty} \frac{1}{N}  \mathbb{E}\left[f^{s}(\bZ^{s})^{\top}f^{t}(\bZ^{t})\right] \in \mathbb{R}^{q \times q}
\end{equation}
where $(\mathbf{Z}^{s},\mathbf{Z}^{t}) \in (\mathbb{R}^{N \times q})^{2}$,$(\mathbf{Z}^{s},\mathbf{Z}^{t}) \sim \mathbf{N}(0,\boldsymbol{\kappa} \otimes \mathbf{I}_{N})$.
\end{enumerate}
Under these assumptions, we define the $\emph{state evolution}$ iteration related to the AMP iteration \eqref{eq:sym_AMP1}-\eqref{eq:sym_AMP2}.
\begin{definition}[state evolution iterates]
\label{def:se-sym}
The state evolution iterates are composed of one infinite-dimensional array $(\boldsymbol{\kappa}^{s,r})_{r,s>0}$ of real matrices. This array is generated as follows. Define the first state evolution iterate
\begin{equation}
    \boldsymbol{\kappa}^{1,1} = \lim_{N \to \infty} \frac{1}{N}  f^{0}(\bX^{0})^{\top}f^{0}(\bX^{0}) 
\end{equation}
Recursively, once $\boldsymbol{\kappa}^{s,r}, 0\leq s,r \leq t$ are defined for some $t \geq 1$, take $\bZ^0 = \bX^0$ and $(\bZ^1, \dots, \bZ^t) \in (\R^{n \times q})^t$ a centered Gaussian vector of covariance $(\boldsymbol{\kappa}^{s,r})_{s,r\leqslant t}\otimes \mathbf{I}_{N}$. We then define new state evolution iterates
\begin{equation*}
    \boldsymbol{\kappa}^{t+1, s+1} = \boldsymbol{\kappa}^{s+1, t+1} = \lim_{N \to \infty} \frac{1}{N} \E\left[ f^s(\bZ^s)^\top f^t(\bZ^t) \right] \, , \qquad s \in \{ 0, \dots, t \} \, .
\end{equation*}
\end{definition}
The following property then holds for the AMP iteration \eqref{eq:sym_AMP1}-\eqref{eq:sym_AMP2}. 
\begin{theorem}
\label{thm:symmetric}
Assume \ref{it:ass-sym-1}-\ref{it:ass-sym-6}. Define, as above, $\bZ^0 = \bX^0$ and $(\bZ^1, \dots, \bZ^t) \in (\R^{N \times q})^t$ a centered Gaussian vector of covariance $ (\boldsymbol{\kappa}^{s,r})_{s,r\leqslant t}\otimes \mathbf{I}_{N}$. Then for any sequence $\Phi_N : (\R^{N \times q})^{t+1} \to \R$ of pseudo-Lipschitz functions, 
\begin{equation*}
  \Phi_N\left(\bX^0, \bX^1, \dots, \bX^t\right) \approxP \E\left[ \Phi_N\left(\bZ^0, \bZ^1, \dots, \bZ^t\right) \right] \, . 
\end{equation*}
\end{theorem}
Given the above result, we can expect the Onsager correction $\bb^{t}$ to verify
\begin{equation}
\label{eq:emp_ons}
\bb^{t} \stackrel{P}\simeq \frac{1}{N} \mathbb{E}\left[\sum_{i=1}^N \frac{\partial f^t_i}{\partial \bZ_i}(\bZ^t)\right] \in 
\R^{q\times q}\, .
\end{equation}
where $\bZ^{t} \sim \mathbf{N}(0,\boldsymbol{\kappa}_{t,t}\otimes \mathbf{I}_{n})$. In fact, similarly to \cite{berthier2020state}, Theorem \ref{thm:symmetric} can be shown to hold for the AMP iteration (\eqref{eq:sym_AMP1}-\eqref{eq:sym_AMP2}) with any estimator $\hat{\bb}^{t}$ satisfying
\begin{equation}
    \hat{\bb}^{t}(\bX^{0},\hat{\bM}^{0},...,\hat{\bM}^{t-1},\hat{\bX}^{t}) \stackrel{P}\simeq \frac{1}{N} \mathbb{E}\left[\sum_{i=1}^N \frac{\partial f^t_i}{\partial \bZ_i}(\bZ^t)\right] \in 
\R^{q\times q}\, .
\end{equation}.
\begin{comment}
The state evolution actually hold for the AMP Eq.(\ref{eq:sym-amp-iteration-1}-\ref{eq:sym-amp-iteration-2}) with any estimator $\hat{\bb}_{t}$ converging in probability to the expectation on the r.h.s.~of Eq.(\eqref{eq:emp_ons}). This is formalized in the following corollary:
\begin{theorem}
\label{emp_Onsager}
Consider the AMP iteration
\begin{align}
\hat{\bX}^{t+1} &= \mathbf{A}\hat{\bM}^{t}-\hat{\bM}^{t-1}\hat{\bb}_{t}^{\top} && \in \R^{N\times q} \,\\
{\hat{\bM}}^{t} &=f^{t}({\hat{\bX}}^{t}) && \in \R^{N\times q}
\end{align}
initialized with $\bX^{0}$ as Eq. (\ref{eq:sym-amp-iteration-1}-\ref{eq:sym-amp-iteration-2}), and where $\hat{\bb}_{t}(\bX^{0},\hat{\bM}^{0},...,\hat{\bM}^{t-1},\hat{\bX}^{t})$ is an estimator of $\bb^{t}$. Under the set of assumptions (A1-A6), and provided the estimator $\hat{b}_{t}$ verifies
\begin{equation}
    \hat{\bb}_{t}(\bX^{0},\hat{\bM}^{0},...,\hat{\bM}^{t-1},\hat{\bX}^{t}) \stackrel{P}\simeq \frac{1}{N} \mathbb{E}\left[\sum_{i=1}^N \frac{\partial f_t^i}{\partial \bZ_i}(\bZ^t)\right] \in 
\R^{q\times q}\, .
\end{equation}
then for any $t \in \mathbb{N}$
\begin{equation}
    \lim_{N \to \infty} \frac{1}{\sqrt{N}}\norm{\hat{\bX}^{t+1}-\bX^{t+1}}_{F} \stackrel{P}\simeq 0, \quad \lim_{N \to \infty} \frac{1}{\sqrt{N}}\norm{\hat{\bM}^{t}-\bM^{t}}_{F} \stackrel{P}\simeq 0
\end{equation}
and the iterates $\hat{\bM}^{t},\hat{\bX}^{t}$ verify the state evolution equations.
\end{theorem}
The proof of this corollary is also provided in Appendix \ref{proof:sym_AMP}.
\end{comment}
\subsection{Application: proof of Theorem \ref{thm:graph-AMP}}
\label{ap:proof-graph-amp}

In Section \ref{ap:reduction-graph}, we have seen that the graph AMP iteration \eqref{eq:graph-amp-1}-\eqref{eq:graph-amp-3} can be rewritten as a symmetric AMP iteration of the form \eqref{eq:sym-amp-iteration-1}-\eqref{eq:sym-amp-iteration-2}. Here, we check that applying Theorem \ref{thm:symmetric} on the symmetric iteration after performing the reduction indeed gives Theorem \ref{thm:graph-AMP}.

Define the state evolution iterates as in Definition \ref{def:se-sym}. Here, due to the expression \eqref{eq:def-non-linearity} of the non-linearities, the state evolution iterates are diagonal: 
\begin{equation}
    \boldsymbol{\kappa}^{1,1} = \lim_{N \to \infty} \frac{1}{N}  \begin{pmatrix}
    \left\Vert f^0_{\overrightarrow{e}_1}((\bx^0_{\overrightarrow{e}})_{\overrightarrow{e}:\overrightarrow{e} \rightarrow \overrightarrow{e}_1}) \right\Vert^2 & & 0 \\
    & \ddots &  \\
     0 & & \left\Vert f^0_{\overleftarrow{e}_m}((\bx^0_{\overrightarrow{e}})_{\overrightarrow{e}:\overrightarrow{e} \rightarrow \overleftarrow{e}_m}) \right\Vert^2  
    \end{pmatrix}
\end{equation}
and 
\begin{equation*}
    \boldsymbol{\kappa}^{t+1, s+1} = \boldsymbol{\kappa}^{s+1, t+1} = \lim_{N \to \infty} \frac{1}{N}  \begin{pmatrix}
    \E  f^s_{\overrightarrow{e}_1}(\dots)^\top f^t_{\overrightarrow{e}_1}(\dots)  & & 0 \\
    & \ddots &  \\
     0 & & \E  f^s_{\overleftarrow{e}_m}(\dots)^\top f^t_{\overleftarrow{e}_m}(\dots) 
    \end{pmatrix} \, .
\end{equation*}
Let $\bZ^t\in \R^{N\times q}$ be the variable from Definition \ref{def:se-sym}. Decompose 
\begin{align*}
    \bZ^t= 
    \begin{pmatrix}
    \bZ^t_{\overrightarrow{e}_1} & & * \\
    & \ddots & \\
    * & & \bZ^t_{\overleftarrow{e}_m}
    \end{pmatrix} \, .
\end{align*}
where $\bZ^t_{(v,w)} \in \R^{n_w}$. The diagonal structure of the state evolution iterates means that $\bZ^t_{\overrightarrow{e}}$ and $\bZ^t_{\overrightarrow{e}'}$ are independent when $\overrightarrow{e} \neq \overrightarrow{e}'$. We thus find that 
\begin{align*}
    \boldsymbol{\kappa}^{s,t} = \begin{pmatrix}
    \boldsymbol{\kappa}^{s,t}_{\overrightarrow{e}_1} & & 0 \\
    & \ddots &  \\
     0 & & \boldsymbol{\kappa}^{s,t}_{\overleftarrow{e}_m} \, ,
    \end{pmatrix}
\end{align*}
where the $\boldsymbol{\kappa}^{s,t}_{\overrightarrow{e}}$ are those defined in Section \ref{sec:se-graph} and the variables $\bZ^t_{\overrightarrow{e}} $ are the same as those defined in Section \ref{sec:se-graph}. 

These elements show that Theorem \ref{thm:graph-AMP} follows from the application of Theorem \ref{thm:symmetric}.
 
\section{Proof of Theorem \ref{thm:symmetric}}
\label{proof:sym_AMP}
Once the concentration lemmas of Appendix \ref{app:sec_conc} are established for matrix valued-variables, the proof follows closely that of \cite{berthier2020state}. We include the main steps (with minor changes) for completeness nonetheless. \\
As an intermediate step, we introduce the following AMP iteration initialized with $X^{0} \in \mathbb{R}^{N \times q}$ :
\begin{align}
\bX^{t+1} &= \mathbf{A}\bM^{t}-\bM^{t-1}(\bb^{t})^{\top} && \in \R^{N\times q} \label{eq:sym-esp-1}\, \\
{\bM}^{t} &=f^{t}({\bX}^{t}) && \in \R^{N\times q} \, , \\ 
\bb_t &= \frac{1}{N} \mathbb{E}\left[\sum_{i=1}^N \frac{\partial f^t_i}{\partial \bZ_i}(\bZ^t)\right] && \in 
\R^{q\times q} \, . \label{eq:sym-esp-3}
\end{align}
where the Onsager term has been replaced by the expectation in Eq.\eqref{eq:emp_ons} using the state evolution recursion, i.e., $\bZ^{t} \in \mathbb{R}^{N \times q} \sim \mathbf{N}(0, \boldsymbol{\kappa}_{t,t}\otimes \mathbf{I}_{N})$. We denote this recursion with the shorthand $\{\bX^{t},\bM^{t} \vert f^{t},\bX^{0}\}$. The following lemma is an analog of Theorem \ref{thm:symmetric} for the iteration \eqref{eq:sym-esp-1}-\eqref{eq:sym-esp-3}.
\begin{lemma}
\label{lemma:symmetric_asy}
Define, as above, $\bZ^0 = \bX^0$ and $(\bZ^1, \dots, \bZ^t) \in (\R^{N \times q})^t$ a centered Gaussian vector of covariance $ \begin{pmatrix} \boldsymbol{\kappa}^{1,1} & \cdots & \boldsymbol{\kappa}^{1,t} \\
\vdots & \ddots & \vdots \\
\boldsymbol{\kappa}^{t,1} & \cdots & \boldsymbol{\kappa}^{t,t}
\end{pmatrix}\otimes \mathbf{I}_N $. Then for any sequence $\Phi_N : (\R^{N \times q})^{t+1} \to \R$ of pseudo-Lipschitz functions, the iterates of \eqref{eq:sym-esp-1}-\eqref{eq:sym-esp-3} satisfy
\begin{equation*}
  \Phi_N\left(\bX^0, \bX^1, \dots, \bX^t\right) \approxP \E\left[ \Phi_N\left(\bZ^0, \bZ^1, \dots, \bZ^t\right) \right] \, . 
\end{equation*}
\end{lemma}

\subsection{Proof outline and intermediate lemmas}
 The main idea is to analyze an iteration that behaves well under Gaussian conditioning and that asymptotically approximates \eqref{eq:sym-esp-1}-\eqref{eq:sym-esp-3}.
\paragraph{Matrix LoAMP.} We consider the following iteration, a matrix-valued version of the LoAMP iteration introduced in \cite{berthier2020state}. The sequence of functions $f^{t}$ and initialization $\bX^{0}$ are the same as for the AMP orbit $\{\bX^{t},\bM^{t}\vert f^{t},\bX^{0}\}$. Initialize $\bQ^{0} = f^{0}(\bX^{0})$, and recursively define
\begin{align}
\label{Lo-AMP}
    \bH^{t+1} &= \mathbf{P}^{\perp}_{\boldsymbol{\mathcal{Q}}_{t-1}}\mathbf{A}\mathbf{P}^{\perp}_{\boldsymbol{\mathcal{Q}}_{t-1}}\bQ^{t}+\boldsymbol{\mathcal{H}}_{t-1}\boldsymbol{\alpha}^{t} \quad \in \mathbb{R}^{N \times q} \, , \\
    \bQ^{t} &= f^{t}(\bH^{t}) \quad \in \mathbb{R}^{N \times q} \, ,
\end{align}

where at each step, the matrices $\boldsymbol{\mathcal{Q}}_{t-1}, \boldsymbol{\alpha}^{t}, \boldsymbol{\mathcal{H}}_{t-1}$ are defined as

\begin{align}
    \boldsymbol{\mathcal{Q}}_{t-1} &= \left[\bQ^{0} \vert \bQ^{1} \vert ... \vert \bQ^{t-1}\right] \quad \in \mathbb{R}^{N \times tq} \, ,  \\
    \boldsymbol{\alpha}^{t} &= (\boldsymbol{\mathcal{Q}}_{t-1}^{\top}\boldsymbol{\mathcal{Q}}_{t-1})^{-1}\boldsymbol{\mathcal{Q}}^{\top}_{t-1}\bQ^{t} \quad \in \mathbb{R}^{tq \times q} \, , \\
    \boldsymbol{\mathcal{H}}_{t-1} &= \left[\bH^{1} \vert \bH^{2} \vert ... \vert \bH^{t}\right] \quad \in \mathbb{R}^{N \times tq} \, , 
\end{align}
 $\mathbf{P}_{\boldsymbol{\mathcal{Q}}_{t-1}} = \boldsymbol{\mathcal{Q}}_{t-1}(\boldsymbol{\mathcal{Q}}_{t-1}^{\top}\boldsymbol{\mathcal{Q}}_{t-1})^{-1}\boldsymbol{\mathcal{Q}}_{t-1}^{\top}$ is the orthogonal projector on the subspace spanned by the columns of $\boldsymbol{\mathcal{Q}}_{t-1}$, and $\mathbf{P}^{\perp}_{\boldsymbol{\mathcal{Q}}_{t-1}} = \mathbf{I}_{N}-\mathbf{P}_{\boldsymbol{\mathcal{Q}}_{t-1}}$. We denote this recursion with the shorthand $\{\bH^{t},\bQ^{t} \vert f^{t},\bX^{0}\}$.
The inverse $(\boldsymbol{\mathcal{Q}}_{t-1}^{\top}\boldsymbol{\mathcal{Q}}_{t-1})^{-1}$ in the projector may not always be properly defined if $\boldsymbol{\mathcal{Q}}_{t-1}$ is either rank-deficient or has vanishing singular values. We thus introduce the following assumption as in \cite{berthier2020state}, which ensures the proper definition of the projector. \\ 
\quad \\
\begin{assumption}[Non-degeneracy]
We say that the LoAMP iterates satisfy the non-degeneracy assumption if :
\begin{itemize}
    \item almost surely, for all $t$ and all $N \geqslant t$, $\boldsymbol{\mathcal{Q}}_{t-1}$ has full column rank.
    \item for all $t $, there exists some constant $c_{t}>0$---independent of N---such that almost surely, there exists $N_{0}$ (random) such that, for $N \geqslant N_{0}$, $\sigma_{\min}(\boldsymbol{\mathcal{Q}}_{t-1})/\sqrt{N}\geqslant c_{t} >0$.
\end{itemize}
\end{assumption}
We now study the LoAMP iteration, starting with the non-degenerate case.
\paragraph{The non-degenerate case.}
The following lemma gives the distribution of the Long-AMP iterates when conditioned on the previous ones. 
\begin{lemma} 
\label{lemma:cond_dist_LoAMP}
Consider the LoAMP iteration $\{\bH^{t},\bQ^{t} \vert f_{t}, \bX^{0}\}$ and assume it satisfies the non-degeneracy assumption. For any $t \in \mathbb{N}$, let $\mathfrak{S}_{t}$ be the $\sigma$-algebra generated by the collection of random variables $\bH^{1},\bH^{2}, ... , \bH^{t}$. Then
\begin{equation}
\bH^{t+1}\vert_{\mathfrak{S}_{t}} \stackrel{d}= \mathbf{P}^{\perp}_{\mathcal{Q}_{t-1}}\tilde{\mathbf{A}}\mathbf{P}^{\perp}_{\mathcal{Q}_{t-1}}\bQ^{t}+\boldsymbol{\mathcal{H}}_{t-1}\boldsymbol{\alpha}^{t} 
\end{equation}
where $\tilde{\mathbf{A}}$ is a copy of $\mathbf{A}$ independent of $\mathfrak{S}_{t}$.
\end{lemma}
The next lemma characterizes the high-dimensional geometry and distribution of the LoAMP iterates, notably that they verify the state evolution equations.
\begin{lemma}
\label{lemma:LoAMP_SE}
Consider the LoAMP recursion $\{\bH^{t},\bQ^{t} \vert f_{t}, \bX^{0}\}$ and suppose it satisfies the non-degeneracy assumption. Then
\begin{enumerate}[label=\alph*)]
\item for all $0 \leqslant s,r \leqslant t$ ,
\begin{equation}
    \frac{1}{N} (\bH^{s+1})^{\top}\bH^{r+1}  \stackrel{P}\simeq \frac{1}{N} (\bQ^{s})^{\top}\bQ^{r} \quad \in \mathbb{R}^{q \times q} \, ,
\end{equation}
\item for any $t \in \mathbb{N}$, for any sequence of uniformly order-k pseudo-Lipschitz functions $\{\phi_{N}:(\mathbb{R}^{N \times q})^{t+2} \to \mathbb{R}\}$,
\begin{equation}
    \Phi_{N}(\bX^{0},\bH^{1},...,\bH^{t+1})\stackrel{P}\simeq \mathbb{E}[\Phi_{N}(\bX^{0},\bZ^{1},...,\bZ^{t+1})]
\end{equation}
where 
\begin{equation}
    (\bZ^{1},...,\bZ^{t+1}) \sim \mathbf{N}(0, (\boldsymbol{\kappa}^{s,r})_{s,r\leqslant t} \otimes \mathbf{I}_{N})
\end{equation}
 \end{enumerate}
\end{lemma}
The next two lemmas show that the iterates of the Long-AMP recursion are arbitrary close to those of the original symmetric AMP in the high-dimensional limit.
\begin{lemma}
\label{lemma:LoAMP_approx1}
    For each iteration t of the LoAMP iteration $\{\bH^{t},\bQ^{t} \vert f^{t}, \bX^{0}\}$, consider the recursion
    \begin{align}
    \label{approx_AMP}
        \hat{\bH}^{t+1} = \mathbf{A}\bQ^{t}-\bQ^{t-1}(\bb^{t})^{\top} \quad \mbox{where} \quad \bb^t &= \frac{1}{N}\mathbb{E}\left[ \sum_{i=1}^N \frac{\partial f^t_i}{\partial \bZ_i}(\bZ^t)\right] \quad \in 
\R^{q\times q}  \\
        \bQ^{t} = f^{t}(\bH^{t})
    \end{align}
where we take $\hat{\bH}^{1} = \mathbf{A}\bQ^{0}$ and $\bZ^{t} \sim \mathbf{N}(0,\mathbf{K}_{t,t}\otimes\mathbf{I}_{N})$with $\mathbf{K}_{t,t}$ defined by the state evolution. Then for any $t \in \mathbb{N}$, $\frac{1}{\sqrt{N}}\norm{\bH^{t+1}-\hat{\bH}^{t+1}}_{F} \xrightarrow[N \to \infty]{P} 0$.
\end{lemma}
\begin{lemma}
\label{lemma:LoAMP_approx2}
    Consider the symmetric AMP iteration $\{\bX^{t},\bM^{t} \vert f_{t}, \bX^{0}\}$ and the LongAMP iteration $\{\bH^{t},\bQ^{t} \vert f_{t},\bX^{0}\}$. Suppose that LongAMP satisfies the non-degeneracy assumption. Then for any $t\in \mathbb{N}$,
    \begin{equation}
        \frac{1}{\sqrt{N}}\norm{\bH^{t+1}-\bX^{t+1}}_{F} \xrightarrow[N \to \infty]{P} 0 \quad \mbox{and} \quad \frac{1}{\sqrt{N}}\norm{\bQ^{t}-\bM^{t}}_{F} \xrightarrow[N \to \infty]{P} 0 
    \end{equation}
\end{lemma}
Combining the previous results, and assuming the non-degeneracy is verified, Lemma \ref{lemma:symmetric_asy} holds true.

\paragraph{Relaxing the non-degeneracy hypothesis}
This paragraph shows how the non-degeneracy assumption is relaxed using a perturbative argument as done in \cite{berthier2020state}. Define the randomly perturbed functions 
\begin{equation}
\label{pert_func}
    f^{t}_{\epsilon \bY^{t}} = f^{t}(.)+\epsilon \bY^{t} 
\end{equation}
where $\bY^{t}\in \mathbb{R}^{N\times q}$ is a matrix with i.i.d.~$\mathbf{N}(0,1)$ entries independent of the original matrix $\mathbf{A}$. We denote $\mathbf{Y}$ the set of random matrices $(\bY^{0},\bY^{1},...,\bY^{t}) \in (\mathbb{R}^{N\times q})^{t+1}$. 
\begin{lemma}
\label{lemma:SE_pert_def}
The AMP iteration defined with the functions $f^{t}_{\epsilon \mathbf{Y}}$ and initialized with $\mathbf{X}^{0}$ verifies Assumptions $\ref{it:ass-sym-4}-\ref{it:ass-sym-6}$.
Furthermore, define the associated state evolution iteration $\{\boldsymbol{\kappa}^{s,t}_{\epsilon}\vert f^{t}_{\epsilon \bY},\bX^{0}\}$, initialized with
\begin{equation}
\label{eq:pert_SE}
    \boldsymbol{\kappa}^{1,1}_{\epsilon} = \lim_{N \to \infty} \frac{1}{N}(f_{e\bY}^{0}(\bX^{0}))^{\top}(f_{e\bY}^{0}(\bX^{0}))
\end{equation}
and
\begin{equation}
    \boldsymbol{\kappa}^{s+1,t+1}_{\epsilon} = \lim_{N \to \infty} \frac{1}{N} \mathbb{E}\left[(f_{\epsilon \bY}^{s}(\bZ^{\epsilon,s})^{\top}f_{\epsilon \bY}^{t}(\bZ^{\epsilon,t})\right]
\end{equation}
where $(\bZ^{\epsilon,1},...,\bZ^{\epsilon,t}) \sim \mathbf{N}(0,(\boldsymbol{\kappa}^{s,r})_{s,r\leqslant t}^{\epsilon}\otimes \mathbf{I}_{N})$ and the expectations are taken w.r.t.~$\bZ^{\epsilon,1},...,\bZ^{\epsilon,t}$ but not on $\bY$.
Then the state evolution $\{\boldsymbol{\kappa}^{s,t}_{\epsilon}\vert f^{t}_{\epsilon \bY},\bX^{0}\}$ is almost surely non-random.
\end{lemma}
\begin{lemma}
\label{lemma:pert_full_rank}
Denote $\boldsymbol{\mathcal{Q}}_{t-1}^{\epsilon \bY}$ the $N \times tq$ matrix associated with the LoAMP iterates $\{\bH^{\epsilon \bY,t},\bQ^{\epsilon \bY,t} \vert f^{t}_{\epsilon\bY},\bX^{0}\}$. Assume $\epsilon>0$. Then for $N \geqslant t$, the matrix $\boldsymbol{\mathcal{Q}}_{t-1}^{\epsilon \bY}$ almost surely has full column-rank. Furthermore, there exists a constant $c_{t, \epsilon}$, independent of n, such that, almost surely, there exists $N_{0}$ (random) such that, for $N \geqslant N_{0}$, $\sigma_{min}(\boldsymbol{\mathcal{Q}}_{t-1}^{\epsilon \bY})/\sqrt{N} \geqslant c_{t, \epsilon}>0$.
\end{lemma}
The next two lemmas show uniform convergence of the perturbed state evolution averages to the original one when the perturbation vanishes.
\begin{lemma}
\label{lemma:unif_conv}
Let $\{\Phi_{N} :\mathbb{R}^{N \times tq} \to \mathbb{R}^{q\times q}\}_{N >0}$ be a sequence of uniformly pseudo-Lipschitz functions of order k. Let $\boldsymbol{\kappa},\tilde{\boldsymbol{\kappa}}$ be two $tq \times tq$ covariance matrices and $\mathbf{Z} \sim \mathbf{N}(0,\boldsymbol{\kappa} \otimes \mathbf{I}_{N})$, $\tilde{\mathbf{Z}} \sim \mathbf{N}(0,\tilde{\boldsymbol{\kappa}} \otimes \mathbf{I}_{N})$.
Then 
\begin{equation}
    \lim_{\tilde{\boldsymbol{\kappa}} \to \boldsymbol{\kappa}} \sup_{N \geqslant 1} \mathbb{E}[\Phi_{N}(\bZ)]-\mathbb{E}[\Phi_{N}(\tilde{\bZ})] = 0 \, .
\end{equation}
\end{lemma}
\begin{lemma}
\label{lemma:conv_SE}
For any $s,t \geqslant 1$, $\boldsymbol{\kappa}_{\epsilon}^{s,t} \xrightarrow[\epsilon \to 0]{} \boldsymbol{\kappa}^{s,t}$.
\end{lemma}
This last lemma shows that the iterates of the AMP orbit defined with the randomly perturbed functions (\ref{pert_func}), denoted $\{\bX^{\epsilon\mathbf{Y},t},\bM^{\epsilon\mathbf{Y},t}\vert f_{\epsilon \bY}^{t},\bX^{0}\}$, is arbitrarily close to the original AMP orbit $\{\bX^{t},\bM^{t}\vert f^{t},\bX^{0}\}$ when the perturbation is taken to zero.
\begin{lemma}
\label{lemma:conv_AMP}
Consider the symmetric AMP orbit defined by $\{\bX^{t},\bM^{t}\vert f^{t},\bX^{0}\}$ and the corresponding perturbed orbit defined by $\{\bX^{\epsilon\mathbf{Y},t},\bM^{\epsilon\mathbf{Y},t}\vert f^{t}_{\epsilon \bY},\bX^{0}\}$. Assume that, for some $t\in \mathbb{N}$. Then there exist functions $h_{t}(\epsilon)$, $h'_{t}(\epsilon)$, independent of $N$, such that
\begin{equation}
    \lim_{\epsilon \to 0} h_{t}(\epsilon) = \lim_{\epsilon \to 0} h'_{t}(\epsilon) = 0
\end{equation}
and for all $\epsilon \leqslant 1$, with high probability,
\begin{align}
    \frac{1}{\sqrt{N}}\norm{\bM^{\epsilon\mathbf{Y} ,t}-\bM^{t}}_{F} &\leqslant h'_{t}(\epsilon) \, , \\
    \frac{1}{\sqrt{N}}\norm{\bX^{\epsilon\mathbf{Y} ,t+1}-\bX^{t+1}}_{F} &\leqslant h_{t}(\epsilon) \, .
\end{align}
\end{lemma}
Combining these lemmas, we now prove Lemma \ref{lemma:symmetric_asy}. 

\subsection{Proof of Lemma \ref{lemma:symmetric_asy} and Theorem \ref{thm:symmetric}}
Theorem \ref{thm:symmetric} follows from Lemma \ref{lemma:symmetric_asy} similarly to the proof of Corollary 2 from \cite{berthier2020state}.
\begin{proof}[Proof of Lemma \ref{lemma:symmetric_asy}]
The lemmas presented in the previous section ensure the following: 
\begin{itemize}
\item Lemma \ref{lemma:pert_full_rank} and \ref{lemma:symmetric_asy} ensure the AMP iteration defined with randomly perturbed functions verifies the non-degeneracy assumptions and the perturbed state evolution equations, i.e.,
\begin{equation*}
  \Phi_N\left(\bX^{0}, \bX^{\epsilon ,1}, \dots, \bX^{\epsilon\mathbf{Y},t}\right) \approxP \E\left[ \Phi_N\left(\bZ^{\epsilon ,0}, \bZ^{\epsilon ,1}, \dots, \bZ^{\epsilon ,t}\right) \right] \, . 
\end{equation*}
for any sequence of pseudo-Lispchitz functions $\Phi_{N}$, where $\left(\bZ^{\epsilon ,0}, \bZ^{\epsilon ,1}, \dots, \bZ^{\epsilon ,t}\right)$ are defined as in Eq.\eqref{eq:pert_SE}.
\item We have shown that the perturbed state evolution converges to the original one for vanishing perturbations, i.e.,
\begin{equation*}
    \sup_{N \geqslant 1}\abs{\E\left[ \Phi_N\left(\bZ^0, \bZ^1, \dots, \bZ^t\right) \right]-\E\left[ \Phi_N\left(\bZ^{\epsilon,0}, \bZ^{\epsilon,1}, \dots, \bZ^{\epsilon,t}\right) \right]} \xrightarrow[\epsilon \to 0]{} 0
\end{equation*}
using Lemma \ref{lemma:unif_conv} and \ref{lemma:conv_SE}.
\item Lemma \ref{lemma:conv_AMP} ensures the AMP orbit   $\{\bX^{\epsilon\mathbf{Y},t},\bM^{\epsilon\mathbf{Y},t}\vert f_{\epsilon \bY}^{t},\bX^{0}\}$ uniformly approximates the $\{\bX^{t},\bM^{t}\vert f^{t},\bX^{0}\}$ one.
\end{itemize}
In light of these results, consider the following decomposition: for any $\eta \geqslant 0$:
\begin{align*}
 &\mathbb{P}\left(\abs{\Phi_{N}\left(\bX^{0},\bX^{1},...,\bX^{t}\right)-\mathbb{E}\left[\Phi_{N}\left(\bX^{0},\bZ^{1},...,\bZ^{t}\right)\right]}\geqslant \eta\right) \notag \\
 &\leqslant \mathbb{P}\left(\abs{\Phi_{N}\left(\bX^{0},\bX^{1},...,\bX^{t}\right)-\Phi_{N}\left(\bX^{0},\bX^{\epsilon \bY,1},...,\bX^{\epsilon \bY,t}\right)}\geqslant \frac{\eta}{3}\right) \notag \\
 &\hspace{1cm}+\mathbb{P}\left(\abs{\Phi_{N}\left(\bX^{0},\bX^{\epsilon \bY,1},...,\bX^{\epsilon \bY,t}\right)-\mathbb{E}\left[\Phi_{N}\left(\bX^{0},\bZ^{\epsilon,1},...,\bZ^{\epsilon,t}\right)\right]}\geqslant \frac{\eta}{3}\right) \notag \\
 &\hspace{1cm}+\mathbb{P}\left(\abs{\mathbb{E}\left[\Phi_{N}\left(\bX^{0},\bZ^{\epsilon,1},...,\bZ^{\epsilon,t}\right)\right]-\mathbb{E}\left[\Phi_{N}\left(\bX^{0},\bZ^{1},...,\bZ^{t}\right)\right]}\geqslant \frac{\eta}{3}\right)
\end{align*}
Starting with the first term of the r.h.s., the pseudo-Lipschitz property and the triangle inequality give
\begin{align*}
    &\abs{\Phi_{N}(\bX^{0},\bX^{1},...,\bX^{t})-\Phi_{N}(\bX^{0},\bX^{\epsilon \bY,1},...,\bX^{\epsilon \bY,t})}\leqslant \notag\\
    &L\bigg(1+2\frac{\norm{\bX^{0}}^{k-1}_{F}}{n^{k-1}}+\sum_{i=1}^{t}\frac{\norm{\bX^{i}}^{k-1}_{F}}{n^{(k-1)/2}}+\sum_{i=1}^{t}\frac{\norm{\bX^{\epsilon,i}}^{k-1}_{F}}{N^{(k-1)/2}}\bigg)\sum_{i=1}^{t}\frac{\norm{\bX^{\epsilon,i}-\bX^{i}}_{F}}{\sqrt{N}} \notag\\
    &\leqslant L\bigg(1+2\frac{\norm{\bX^{0}}^{k-1}_{F}}{n^{(k-1)/2}}+\sum_{i=1}^{t}\frac{\norm{\bX^{i}-\bX^{\epsilon,i}+\bX^{\epsilon,i}}^{k-1}_{F}}{n^{(k-1)/2}}+\sum_{i=1}^{t}\frac{\norm{\bX^{\epsilon,i}}^{k-1}_{F}}{n^{(k-1)/2}}\bigg)\sum_{i=1}^{t}\frac{\norm{\bX^{\epsilon,i}-\bX^{i}}_{F}}{\sqrt{N}} \notag\\
    &\leqslant L\bigg(1+2\frac{\norm{\bX^{0}}^{k-1}_{F}}{n^{(k-1)/2}}+\sum_{i=1}^{t}\frac{\norm{\bX^{i}-\bX^{\epsilon,i}}^{k-1}_{F}}{n^{(k-1)/2}}+2\sum_{i=1}^{t}\frac{\norm{\bX^{\epsilon,i}}^{k-1}_{F}}{n^{(k-1)/2}}\bigg)\sum_{i=1}^{t}\frac{\norm{\bX^{\epsilon,i}-\bX^{i}}_{F}}{\sqrt{N}} \notag\\
    &\leqslant L\bigg(1+2C_{0}^{k-1}+\sum_{i=1}^{t}h_{i}(\epsilon)^{k-1}+2\sum_{i=1}^{t}C_{\epsilon \bY,t}^{k-1}\bigg)\sum_{i=1}^{t}h_{i}(\epsilon) \quad \mbox{w.h.p.}
\end{align*}
where we used assumption \ref{it:ass-sym-3} for the convergence of $\norm{\bX_{0}}_{F}/\sqrt{N}$ to a finite constant, the well-defined state evolution of the perturbed orbit $\{\bX^{\epsilon\mathbf{Y},t},\bM^{\epsilon\mathbf{Y},t}\vert f_{\epsilon \bY}^{t},\bX^{0}\}$ for convergence of $\norm{\bX^{\epsilon,i}}/\sqrt{N}$ to finite constants $C_{\epsilon \bY,t}$ and Lemma \ref{lemma:conv_AMP} to replace the differences $\norm{\bX^{\epsilon,i}-\bX^{i}}_{F}$ by the functions $h_{i}(\epsilon)$ with high probability. This gives, for any $\eta>0$:
\begin{equation}
    \lim_{\epsilon \to 0}\limsup_{N \to \infty} \mathbb{P}\left(\abs{\Phi_{N}\left(\bX^{0},\bX^{1},...,\bX^{t}\right)-\Phi_{N}\left(\bX^{0},\bX^{\epsilon \bY,1},...,\bX^{\epsilon \bY,t}\right)}\geqslant \frac{\eta}{3}\right) = 0
\end{equation}
The state evolution for the perturbed AMP then gives
\begin{equation}
    \lim_{\epsilon \to 0}\limsup_{N \to \infty} \mathbb{P}\left(\abs{\Phi_{N}\left(\bX^{0},\bX^{\epsilon \bY,1},...,\bX^{\epsilon \bY,t}\right)-\mathbb{E}\left[\Phi_{N}\left(\bX^{0},\bZ^{\epsilon,1},...,\bZ^{\epsilon,t}\right)\right]}\geqslant \frac{\eta}{3}\right) = 0
\end{equation}
and Lemma \ref{lemma:unif_conv} guarantees:
\begin{equation}
   \lim_{\epsilon \to 0} \mathbb{P}\left(\abs{\mathbb{E}\left[\Phi_{N}\left(\bX^{0},\bZ^{\epsilon \bY,1},...,\bZ^{\epsilon \bY,t}\right)\right]-\mathbb{E}\left[\Phi_{N}\left(\bX^{0},\bZ^{1},...,\bZ^{t}\right)\right]}\geqslant \frac{\eta}{3}\right) = 0
\end{equation}
for all N. From this we deduce
\begin{equation}
 \mathbb{P}\left(\abs{\Phi_{N}\left(\bX^{0},\bX^{1},...,\bX^{t}\right)-\mathbb{E}\left[\Phi_{N}\left(\bX^{0},\bZ^{1},...,\bZ^{t}\right)\right]}\geqslant \eta\right) \xrightarrow[N \to \infty]{} 0
\end{equation}
which is the desired result.
\end{proof}
\subsection{Proof of intermediate lemmas}
Those proofs which are too close to the ones appearing in \cite{berthier2020state} are not reminded.
\begin{proof}[Proof of Lemma \ref{lemma:cond_dist_LoAMP}]
Recall the $\sigma$-algebra $\mathfrak{S}_{t} = \sigma(\bH^{1},\bH^{2},...,\bH^{t})$. The LongAMP iteration verifies:
\begin{align}
    \bH^{t+1} &= (\mathbf{Id}-\mathbf{P}_{\boldsymbol{\mathcal{Q}}_{t-1}})\mathbf{A}\mathbf{P}_{\boldsymbol{\mathcal{Q}}_{t-1}}^{\perp}\bQ^{t}+\boldsymbol{\mathcal{H}}_{t-1}\boldsymbol{\alpha}^{t} \\
    &= \mathbf{A}\bQ^{t}_{\perp}-\mathbf{P}_{\boldsymbol{\mathcal{Q}}_{t-1}}\mathbf{A}\bQ^{t}_{\perp}+\boldsymbol{\mathcal{H}}_{t-1}\boldsymbol{\alpha}^{t} 
\end{align}
where $\bQ^{t}_{\perp} = \mathbf{P}_{\boldsymbol{\mathcal{Q}}_{t-1}}^{\perp}\bQ^{t}$. We now show by an induction that conditioning on $\mathfrak{S}_{t}$ is equivalent to conditioning on the linear observations $\mathbf{A}\bQ^{0},\mathbf{A}\bQ^{1},...,\mathbf{A}\bQ^{t}$, and thus to conditioning on $\mathbf{A}\bQ_{t-1}$. Consider the first iteration which initializes the induction:
\begin{align}
    \bH^{1} = \mathbf{A}\bQ^{0}
\end{align}
thus $\bH^{1}$ is $\sigma(\mathbf{A}\bQ^{0})$-measurable. Suppose now that $\boldsymbol{\mathcal{H}}_{t-1}$ is $\sigma(\mathbf{A}\boldsymbol{\mathcal{Q}}_{t-1})$-measurable. The LongAMP iteration then gives, remembering that $\bQ^{t}_{\parallel}=P_{\bQ_{t-1}}\bQ^{t}$ :
\begin{equation}
    \bH^{t+1} = \mathbf{A}\bQ^{t}-\underbrace{\mathbf{A}\bQ^{t}_{\parallel}-\mathbf{P}_{\boldsymbol{\mathcal{Q}}_{t-1}}\mathbf{A}\bQ^{t}_{\perp}+\boldsymbol{\mathcal{H}}_{t-1}\boldsymbol{\alpha}^{t}}_{\sigma(\mathbf{A}\boldsymbol{\mathcal{Q}}_{t-1})-\mbox{measurable}}
\end{equation}
where the highlighted term is $\sigma(\mathbf{A}\boldsymbol{\mathcal{Q}}_{t-1})-\mbox{measurable}$ by definition of $\bQ^{t}_{\parallel}$ and the induction hypothesis. This gives that $\boldsymbol{\mathcal{H}}_{t}$ is $\sigma(\mathbf{A}\bQ_{t})$-measurable. We can now condition on the linear observation $\mathbf{A}\boldsymbol{\mathcal{Q}}_{t-1}$ at each iteration. We thus have:
\begin{equation}
     \bH^{t+1}\vert_{\mathfrak{S}_{t}} \stackrel{d}= \mathbf{A}\vert_{\mathfrak{S}_{t}}\bQ^{t}_{\perp}-\mathbf{P}_{\boldsymbol{\mathcal{Q}}_{t-1}}\mathbf{A}\bQ^{t}_{\perp}+\boldsymbol{\mathcal{H}}_{t-1}\boldsymbol{\alpha}^{t}
\end{equation}
which amounts to condition the Gaussian space generated by the entries of $\mathbf{A}$ on its subspace defined by the linear combinations $\mathbf{A}\boldsymbol{\mathcal{Q}}_{t-1}$. Conditioning in Gaussian spaces amounts to doing orthogonal projections, which gives
\begin{equation}
    \mathbf{A} \vert_{\mathfrak{S}_{t}} = \mathbb{E}\left[\mathbf{A}\vert\mathfrak{S}_{t}\right]+\mathcal{P}_{t}(\tilde{\mathbf{A}}) \\
\end{equation}
as shown in \cite{bayati2011dynamics},\cite{javanmard2013state}, where $\tilde{\mathbf{A}}$ is a copy of $\mathbf{A}$, independent of $\mathfrak{S}_{t}$ and $\mathcal{P}_{t}$ is the projector onto the subspace $\{\hat{\mathbf{A}} \in \mathbb{R}^{N \times N} \vert \hat{\mathbf{A}}\boldsymbol{\mathcal{Q}}_{t-1} = 0, \hat{\mathbf{A}} = \hat{\mathbf{A}}^{\top}\}$  :
\begin{align}
    \mathbb{E}\left[\mathbf{A}\vert\mathfrak{S}_{t}\right] &= \mathbf{A}-\mathbf{P}^{\perp}_{\boldsymbol{\mathcal{Q}}_{t-1}}\mathbf{A}\mathbf{P}^{\perp}_{\boldsymbol{\mathcal{Q}}_{t-1}} \\
    \mathcal{P}_{t}(\tilde{A}) &= \mathbf{P}^{\perp}_{\boldsymbol{\mathcal{Q}}_{t-1}}\tilde{\mathbf{A}}\mathbf{P}^{\perp}_{\boldsymbol{\mathcal{Q}}_{t-1}}
\end{align}
where $\tilde{\mathbf{A}}$ is an independent copy of $\mathbf{A}$. Replacing in the original LongAMP iteration, we get : 
\begin{equation}
    \bH^{t+1}\vert_{\mathfrak{S}_{t}} \stackrel{d}= \mathbf{P}^{\perp}_{\boldsymbol{\mathcal{Q}}_{t-1}}\tilde{\mathbf{A}}\mathbf{P}^{\perp}_{\boldsymbol{\mathcal{Q}}_{t-1}}\bQ^{t}+\boldsymbol{\mathcal{H}}_{t-1}\boldsymbol{\alpha}^{t}
\end{equation}
where we used $\mathbf{P}^{\perp}_{\boldsymbol{\mathcal{Q}}_{t-1}}\mathbb{E}\left[\mathbf{A}\vert\mathfrak{S}_{t}\right]\mathbf{P}^{\perp}_{\boldsymbol{\mathcal{Q}}_{t-1}} = 0$.
\end{proof}
\begin{proof}[Proof of Lemma \ref{lemma:LoAMP_SE}]
We proceed by induction over t. Let $S_{t}$ be the property at time $t$.
\paragraph{Initialization.}
\begin{enumerate}[label=\alph*)]
\item We have $\bH^{1} = \mathbf{A}\bQ^{0}$. Then:
\begin{align}
    \frac{1}{N}(\bH^{1})^{\top}\bH^{1} &= \frac{1}{N}(\mathbf{A}\bQ^{0})^{\top}(\mathbf{A}\bQ^{0}) \notag\\
    & \stackrel{P}\simeq \frac{1}{N}(\bQ^{0})^{\top}\bQ^{0}
\end{align}
using Lemma \ref{conv_lemmas_app}. We then define $\boldsymbol{\kappa}^{1,1} = \frac{1}{N}(\bQ^{0})^{\top}\bQ^{0}$.
\item We want to show that $\Phi_{N}(\bX^{0},\bH^{1}) \stackrel{P}\simeq \mathbb{E}\left[\Phi_{N}(\bX^{0},\bZ^{1})]\right]$ where $\bZ^{1} \sim \mathbf{N}(0,\boldsymbol{\kappa}^{1,1})$, where
\begin{align}
    \boldsymbol{\kappa}^{1,1} = \frac{1}{N}(\bQ^{0})^{\top}\bQ^{0}=\frac{1}{N}\left(f^{0}(\bX^{0})\right)^{\top}f^{0}(\bX^{0})
\end{align}
For any sequence $\{\Phi_{N}\}_{N\in \mathbb{N}}$ of order k pseudo-Lipschitz function
\begin{align}
    &\norm{\Phi_{N}(\bX_{0},\mathbf{A}\bQ^{0})-\mathbb{E}[\Phi_{N}(\bZ^{1})]}_{2} \leqslant \norm{\Phi_{N}(\mathbf{A}\bQ^{0})-\Phi_{N}(\bZ^{1})}_{2}+\norm{\Phi_{N}(\bZ^{1})-\mathbb{E}[\Phi_{N}(\bZ^{1})]}_{2} \notag\\
    &\leqslant L_{n}\left(1+\left(\frac{\norm{\mathbf{A}\bQ^{0}}_{2}}{\sqrt{N}}\right)^{k-1}+\left(\frac{\norm{\bZ^{1}}}{\sqrt{N}}\right)^{k-1}\right)\frac{\norm{\mathbf{A}\bQ^{0}-\bZ^{1}}_{2}}{\sqrt{N}}+\norm{\Phi_{N}(\bZ^{1})-\mathbb{E}[\Phi_{N}(\bZ^{1})]}_{2}
\end{align}
where the large $n$ limit of $\left(\frac{\norm{\mathbf{A}\bQ^{0}}_{2}}{\sqrt{N}}\right)^{k-1}+\left(\frac{\norm{\bZ^{1}}}{\sqrt{N}}\right)^{k-1}$ being bounded, $\frac{\norm{\mathbf{A}\bQ^{0}-\bZ^{1}}_{2}}{\sqrt{N}} \xrightarrow[n \to \infty]{a.s} 0$  and $\norm{\Phi_{N}(\bZ^{1})-\mathbb{E}[\Phi_{N}(\bZ^{1})]}_{2} \xrightarrow[n \to \infty]{P} 0$ follow from Lemmas \ref{pseudo-lip-conv} and \ref{conv_lemmas_app} .
\end{enumerate}
\paragraph{Induction.} 
Here we assume that $S_{0},S_{1},...,S_{t-1}$ are verified, and we prove $S_{t}$.
\begin{enumerate}[label=\alph*)]
\item Consider the case $s<t$. Since $\bH^{s+1}$ and $\langle \bQ^{s},\bQ^{r}\rangle$ are $\mathfrak{S}_{t}$ measurable, using the conditioning lemma, we have :
\begin{align}
    \left((\bH^{s+1})^{\top}\bH^{t+1}-(\bQ^{s})^{\top}\bQ^{t}\right) \vert_{\mathfrak{S}_{t}} \stackrel{d} = \left((\bH^{s+1})^{\top}\bH^{t+1}\vert_{\mathfrak{S}_{t}}-(\bQ^{s})^{\top}\bQ^{t}\right) \notag \\
    = (\bH^{s+1})^{\top}(\mathbf{P}^{\perp}_{\boldsymbol{\mathcal{Q}}_{t-1}}\tilde{\mathbf{A}}\mathbf{P}^{\perp}_{\boldsymbol{\mathcal{Q}}_{t-1}}\bQ^{t}+\boldsymbol{\mathcal{H}}_{t-1}\boldsymbol{\alpha}^{t})-(\bQ^{s})^{\top}\bQ^{t} \notag\\
    = (\bH^{s+1})^{\top}\mathbf{P}^{\perp}_{\boldsymbol{\mathcal{Q}}_{t-1}}\tilde{\mathbf{A}}\bQ^{t}_{\perp}+(\bH^{s+1})^{\top}\boldsymbol{\mathcal{H}}_{t-1}\boldsymbol{\alpha}^{t}-(\bQ^{s})^{\top}\bQ^{t}
\end{align}
We thus have :
\begin{align}
    \frac{1}{N}\norm{\left((\bH^{s+1})^{\top}\bH^{t+1}-(\bQ^{s})^{\top}\bQ^{t}\right) \vert_{\mathfrak{S}_{t}}}_{F} \leqslant &\frac{1}{N} \norm{(\bH^{s+1})^{\top}\mathbf{P}^{\perp}_{\boldsymbol{\mathcal{Q}}_{t-1}}\tilde{\mathbf{A}}\bQ^{t}_{\perp}}_{F} \notag\\
    &+\frac{1}{N}\norm{(\bH^{s+1})^{\top}\boldsymbol{\mathcal{H}}_{t-1}\boldsymbol{\alpha}^{t}-(\bQ^{s})^{\top}\bQ^{t}}_{F}
\end{align}
Starting with the term
\begin{equation}
    \frac{1}{N} \norm{(\bH^{s+1})^{\top}\mathbf{P}^{\perp}_{\boldsymbol{\mathcal{Q}}_{t-1}}\tilde{\mathbf{A}}\bQ^{t}_{\perp}}_{F} = \frac{1}{N} \norm{(\mathbf{P}^{\perp}_{\boldsymbol{\mathcal{Q}}_{t-1}}\bH^{s+1})^{\top}\tilde{\mathbf{A}}\bQ^{t}_{\perp}}_{F}
\end{equation}
the induction ensires that $\frac{1}{\sqrt{N}}\norm{\bH^{s+1}}_{F},\frac{1}{\sqrt{N}}\norm{\bQ^{t}_{\perp}}_{F}$ concentrate to finite values. Furthermore, $\norm{\mathbf{P}^{\perp}_{\boldsymbol{\mathcal{Q}}_{t-1}}\bH^{s+1}}_{F} \leqslant \norm{\bH^{s+1}}_{F}$, so according to Lemma \ref{conv_lemmas_app}, the first term on the right-hand-side will concentrate to zero.\\
Moving to the second term, since $s<t$,  $\mathbf{P}_{\boldsymbol{\mathcal{Q}}_{t-1}}\bQ^{s} = \bQ^{s}$. Then:
\begin{align}
    \frac{1}{N}\norm{(\bH^{s+1})^{\top}\boldsymbol{\mathcal{H}}_{t-1}\boldsymbol{\alpha}^{t}-(\bQ^{s})^{\top}\bQ^{t}}_{F} &= \frac{1}{N}\norm{(\bH^{s+1})^{\top}\boldsymbol{\mathcal{H}}_{t-1}\boldsymbol{\alpha}^{t}-(\mathbf{P}_{\boldsymbol{\mathcal{Q}}_{t-1}}\bQ^{s})^{\top}\bQ^{t}}_{F} \notag\\
    &=\frac{1}{N}\norm{(\bH^{s+1})^{\top}\boldsymbol{\mathcal{H}}_{t-1}\boldsymbol{\alpha}^{t}-(\bQ^{s})^{\top}\boldsymbol{\mathcal{Q}}_{t-1}(\boldsymbol{\mathcal{Q}}_{t-1}^{\top}\boldsymbol{\mathcal{Q}}_{t-1})^{-1}\boldsymbol{\mathcal{Q}}_{t-1}^{\top}\bQ^{t}}_{F} \notag\\
    &= \frac{1}{N}\norm{(\bH^{s+1})^{\top}\boldsymbol{\mathcal{H}}_{t-1}\boldsymbol{\alpha}^{t}-(\bQ^{s})^{\top}\boldsymbol{\mathcal{Q}}_{t-1}\boldsymbol{\alpha}_{t}}_{F} \notag\\
    &\leqslant \frac{1}{N}\norm{(\bH^{s+1})^{\top}\boldsymbol{\mathcal{H}}_{t-1}-(\bQ^{s})^{\top}\boldsymbol{\mathcal{Q}}_{t-1}}_{F}\norm{\boldsymbol{\alpha_{t}}}_{F} 
\end{align}
Here we consider $s<t$ thus $s+1\leqslant t$. Hence the induction hypothesis includes the concentration properties of $\bH^{s+1}$ and $\boldsymbol{\alpha}_{t}$. We then have $\lim_{N \to \infty}\frac{1}{N}\norm{(\bH^{s+1})^{\top}\boldsymbol{\mathcal{H}}_{t-1}-(\bQ^{s})^{\top}\boldsymbol{\mathcal{Q}}_{t-1}}_{F} \to 0$ and  $\norm{\boldsymbol{\alpha}_{t}}_{F}$ has a finite and well-defined limit using the non-degeneracy assumption. Indeed:
\begin{align}
   \norm{\boldsymbol{\alpha}_{t}}_{F} &= \norm{(\boldsymbol{\mathcal{Q}}_{t-1}^{\top}\boldsymbol{\mathcal{Q}}_{t-1})^{-1}\boldsymbol{\mathcal{Q}}_{t-1}^{\top}\bQ^{t}}_{F} \notag \\
   &\leqslant \frac{1}{Nc_{t}^{2}}\boldsymbol{\mathcal{Q}}_{t-1}^{\top}\bQ^{t}
\end{align}
using the induction hypothesis, $\lim_{n \to +\infty} \frac{1}{N}\boldsymbol{\mathcal{Q}}_{t-1}^{\top}\bQ^{t}$ is finite.
This proves the property for $s<t$.
Now consider the case $s=t$. We then have: 
\begin{align}
    &\left(\norm{\bH^{t+1}}_{F}^{2}-\norm{\bQ^{t}}_{F}^{2}\right)\vert_{\mathfrak{S}_{t}} = \left(\norm{\bH^{t+1}\vert_{\mathfrak{S}_{t}}}_{F}^{2}-\norm{\bQ^{t}}_{F}^{2}\right) \notag \\
    &= \norm{\mathbf{P}^{\perp}_{\boldsymbol{\mathcal{Q}}_{t-1}}\tilde{\mathbf{A}}\bQ^{t}_{\perp}}_{F}^{2}+2\mbox{Tr}\left(\left(\mathbf{P}^{\perp}_{\boldsymbol{\mathcal{Q}}_{t-1}}\tilde{\mathbf{A}}\bQ^{t}_{\perp}\right)^{\top}\boldsymbol{\mathcal{H}}_{t-1}\boldsymbol{\alpha}^{t}\right)+\norm{\boldsymbol{\mathcal{H}}_{t-1}\boldsymbol{\alpha}^{t}}_{F}^{2}-\norm{\bQ^{t}}_{F}^{2}
\end{align}
We then have
\begin{align}
\frac{1}{N}\norm{\mathbf{P}^{\perp}_{\boldsymbol{\mathcal{Q}}_{t-1}}\tilde{\mathbf{A}}\bQ^{t}_{\perp}}_{F}^{2} = \frac{1}{N}\norm{\tilde{\mathbf{A}}\bQ^{t}_{\perp}}_{F}^{2}-\frac{1}{N}\norm{\mathbf{P}_{\boldsymbol{\mathcal{Q}}_{t-1}}\tilde{\mathbf{A}}\bQ^{t}_{\perp}}_{F}^{2} \stackrel{P} \simeq \frac{1}{N}\norm{\bQ^{t}_{\perp}}_{F}^{2}
\end{align}
where we used
\begin{align}
    \frac{1}{N}\norm{\tilde{\mathbf{A}}\bQ^{t}_{\perp}}_{F}^{2} \stackrel{P} \simeq \frac{1}{N}\norm{\bQ^{t}_{\perp}}_{F}^{2} \quad \mbox{and} \quad 
    \frac{1}{N}\norm{\mathbf{P}_{\boldsymbol{\mathcal{Q}}_{t-1}}\tilde{\mathbf{A}}\bQ^{t}_{\perp}}_{F}^{2} \xrightarrow[n \to \infty]{P} 0
\end{align}
which follows from Lemma \ref{conv_lemmas_app} and the independence of $\tilde{\mathbf{A}}$.
The second term then reads
\begin{align}
    \left(\mathbf{P}^{\perp}_{\boldsymbol{\mathcal{Q}}_{t-1}}\tilde{\mathbf{A}}\bQ^{t}_{\perp}\right)^{\top}\boldsymbol{\mathcal{H}}_{t-1}\boldsymbol{\alpha}^{t} = (\bQ_{\perp}^{t})^{\top}\tilde{\mathbf{A}}\mathbf{P}_{\boldsymbol{\mathcal{Q}}_{t-1}}^{\perp}\boldsymbol{\mathcal{H}}_{t-1}\boldsymbol{\alpha}^{t}
\end{align}
From the induction hypothesis, we know that $\boldsymbol{\alpha}^{t}$ has finite norm when $N \to \infty$. Moreover, $\norm{\mathbf{P}_{\boldsymbol{\mathcal{Q}}_{t-1}}^{\perp}\boldsymbol{\mathcal{H}}_{t-1}\boldsymbol{\alpha}^{t}}_{F} \leqslant \norm{\boldsymbol{\mathcal{H}}_{t-1}\boldsymbol{\alpha}^{t}}_{F}$, and $\norm{\bQ^{t}_{\perp}}_{F} \leqslant \norm{\bQ^{t}}_{F}$. Also $\frac{1}{\sqrt{N}}\norm{\boldsymbol{\mathcal{H}}_{t-1}}_{F}$ and $\frac{1}{\sqrt{N}}\norm{\bQ^{t}}_{F}$ converge to finite constants, again according to the induction hypothesis. Using Lemma \ref{conv_lemmas_app}, we get 
\begin{equation}
    \frac{1}{N}\mbox{Tr}\left(\left(\mathbf{P}^{\perp}_{\boldsymbol{\mathcal{Q}}_{t-1}}\tilde{\mathbf{A}}\bQ^{t}_{\perp}\right)^{\top}\boldsymbol{\mathcal{H}}_{t-1}\boldsymbol{\alpha}^{t}\right) \xrightarrow[n\to \infty]{P}0
\end{equation}
Finally the third term can be decomposed
\begin{align}
    \norm{\boldsymbol{\mathcal{H}}_{t-1}\boldsymbol{\alpha}^{t}}_{F}^{2} &= \mbox{Tr}(\left(\boldsymbol{\mathcal{H}}_{t-1}\boldsymbol{\alpha}^{t}\right)^{\top}\boldsymbol{\mathcal{H}}_{t-1}\boldsymbol{\alpha}^{t}) \notag\\
    &= \mbox{Tr}((\boldsymbol{\alpha}^{t})^{\top}\boldsymbol{\mathcal{H}}_{t-1}^{\top}\boldsymbol{\mathcal{H}}_{t-1}\boldsymbol{\alpha}^{t}) \notag\\
    &= \mbox{Tr}((\boldsymbol{\alpha}^{t})^{\top}\left(\boldsymbol{\mathcal{H}}_{t-1}^{\top}\boldsymbol{\mathcal{H}}_{t-1}-\boldsymbol{\mathcal{Q}}_{t-1}^{\top}\boldsymbol{\mathcal{Q}}_{t-1}\right)\boldsymbol{\alpha}^{t})+\Tr((\boldsymbol{\alpha}^{t})^{\top}\boldsymbol{\mathcal{Q}}_{t-1}^{\top}\boldsymbol{\mathcal{Q}}_{t-1}\boldsymbol{\alpha}^{t}) \notag\\
    & \leqslant \norm{\boldsymbol{\mathcal{H}}_{t-1}^{\top}\boldsymbol{\mathcal{H}}_{t-1}-\boldsymbol{\mathcal{Q}}_{t-1}^{\top}\boldsymbol{\mathcal{Q}}_{t-1}}_{F}\norm{\boldsymbol{\alpha}_{t}}_{F}+\norm{\boldsymbol{\mathcal{Q}}_{t-1}\boldsymbol{\alpha}_{t}}_{F}
\end{align} 
Using the induction hypothesis and the non-degeneracy assumption, $\lim_{N \to \infty} \norm{\boldsymbol{\alpha}^{t}}_{F}$ is a finite constant, and $ \frac{1}{N}\norm{\boldsymbol{\mathcal{H}}_{t-1}^{\top}\boldsymbol{\mathcal{H}}_{t-1}-\boldsymbol{\mathcal{Q}}_{t-1}^{\top}\boldsymbol{\mathcal{Q}}_{t-1}}_{F} \xrightarrow[n \to \infty]{P} 0$. Furthermore, by definition of $\boldsymbol{\alpha}_{t}$, $\boldsymbol{\mathcal{Q}}_{t-1}\boldsymbol{\alpha}_{t} = \bQ_{t}^{\parallel}$.

Grouping all the terms, we get 
\begin{align}
\frac{1}{N}\left(\norm{\bH^{t+1}}_{F}^{2}-\norm{\bQ^{t}}_{F}^{2}\right)\vert_{\mathfrak{S}_{t}} &\stackrel{P} \simeq \frac{1}{N}\norm{\bQ^{t}_{\perp}}_{F}^{2}+\frac{1}{N}\norm{\bQ^{t}_{\parallel}}_{F}^{2}-\frac{1}{N}\norm{\bQ^{t}}_{F}^{2} \notag \\
&=0
\end{align}

\item Using the conditioning lemma :
\begin{align}
    \Phi_{N}\left(\bX^{0},\bH^{1},...,\bH^{t},\bH^{t+1}\right)\vert_{\mathfrak{S}_{t}} &\stackrel{d} = \Phi_{N}\left(\bX^{0},\bH^{1},...,\bH^{t},\mathbf{P}^{\perp}_{\boldsymbol{\mathcal{Q}}_{t-1}}\tilde{\mathbf{A}}\mathbf{P}^{\perp}_{\boldsymbol{\mathcal{Q}}_{t-1}}\bQ^{t}+\boldsymbol{\mathcal{H}}_{t-1}\boldsymbol{\alpha}^{t}\right) \notag\\
    &= \Phi_{N}\left(\bX^{0},\bH^{1},...,\bH^{t},\tilde{\mathbf{A}}\bQ^{t}_{\perp}-P_{\boldsymbol{\mathcal{Q}}_{t-1}}\tilde{\mathbf{A}}\bQ^{t}_{\perp}+\boldsymbol{\mathcal{H}}_{t-1}\boldsymbol{\alpha}^{t}\right)
\end{align}
Let $\Phi^{'}_{N}\left(\tilde{\mathbf{A}}\bQ^{t}_{\perp}-P_{\boldsymbol{\mathcal{Q}}_{t-1}}\tilde{\mathbf{A}}\bQ^{t}_{\perp}+\boldsymbol{\mathcal{H}}_{t-1}\boldsymbol{\alpha}^{t}\right) = \Phi_{N}\left(\bX^{0},\bH^{1},...,\bH^{t},\tilde{\mathbf{A}}\bQ^{t}_{\perp}-P_{\boldsymbol{\mathcal{Q}}_{t-1}}\tilde{\mathbf{A}}\bQ^{t}_{\perp}+\boldsymbol{\mathcal{H}}_{t-1}\boldsymbol{\alpha}^{t}\right)$ as a shorthand. Then, from the pseudo-Lipschitz property:
\begin{align}
    &\abs{\Phi^{'}_{N}\left(\tilde{\mathbf{A}}\bQ^{t}_{\perp}-P_{\boldsymbol{\mathcal{Q}}_{t-1}}\tilde{\mathbf{A}}\bQ^{t}_{\perp}+\boldsymbol{\mathcal{H}}_{t-1}\boldsymbol{\alpha}^{t}\right)-\Phi^{'}_{N}\left(\tilde{\mathbf{A}}\bQ^{t}_{\perp}+\boldsymbol{\mathcal{H}}_{t-1}\boldsymbol{\alpha}^{t}\right)} \notag\\
    &\hspace{1cm}\leqslant L_{N}C(k,t)\bigg(1+\left(\frac{\norm{\bX^{0}}_{F}}{\sqrt{N}}\right)^{k-1}+\sum_{s=1}^{t}\left(\frac{\norm{\bH^{s}}_{F}}{\sqrt{N}}\right)^{k-1} \notag \\
    &\hspace{1cm}+\left(\frac{\norm{\bH^{t+1}}_{F}}{\sqrt{N}}\right)^{k-1}+\left(\frac{\norm{\tilde{\mathbf{A}}\bQ^{t}_{\perp}}_{F}}{\sqrt{N}}\right)^{k-1}+\left(\frac{\norm{\boldsymbol{\mathcal{H}}_{t-1}\boldsymbol{\alpha}^{t}}_{F}}{\sqrt{N}}\right)^{k-1}\bigg)\frac{\norm{P_{\boldsymbol{\mathcal{Q}}_{t-1}}\tilde{\mathbf{A}}\bQ^{t}_{\perp}}_{F}}{\sqrt{N}}
\end{align}
where $C(k,t)$ is a constant depending only on k and t. The induction hypothesis ensures that $\left(\frac{\norm{\bX^{0}}_{F}}{\sqrt{N}}\right)^{k-1}+\sum_{s=1}^{t}\left(\frac{\norm{\bH^{s}}_{F}}{\sqrt{N}}\right)^{k-1}$ converges to a finite constant. Furthermore, 
\begin{align}
    \frac{1}{\sqrt{N}}\norm{\tilde{\mathbf{A}}}_{F} \leqslant \frac{1}{\sqrt{N}}\norm{\tilde{\mathbf{A}}}_{op}\norm{\bQ^{t}}_{F}
\end{align}
which, using Proposition \ref{op-norm-GOE} and the induction hypothesis, converges to a finite constant. Also, using the fact that rank$(\mathbf{P}_{\boldsymbol{\mathcal{Q}}_{t-1}})\leqslant tq$ with $t,q$ finite, and the independence of $\tilde{\mathbf{A}}$, Lemma \ref{conv_lemmas_app} gives
\begin{equation}
    \frac{1}{\sqrt{N}}\norm{P_{\boldsymbol{\mathcal{Q}}_{t-1}}\tilde{\mathbf{A}}\bQ^{t}_{\perp}}_{F} \xrightarrow[n \to \infty]{P} 0 \, .
\end{equation}
Ultimately, we obtain
\begin{align}
    \Phi^{'}_{N}\left(\tilde{\mathbf{A}}\bQ^{t}_{\perp}-P_{\boldsymbol{\mathcal{Q}}_{t-1}}\tilde{\mathbf{A}}\bQ^{t}_{\perp}+\boldsymbol{\mathcal{H}}_{t-1}\boldsymbol{\alpha}^{t}\right) \stackrel{P} \simeq \Phi^{'}_{N}\left(\tilde{\mathbf{A}}\bQ^{t}_{\perp}+\boldsymbol{\mathcal{H}}_{t-1}\boldsymbol{\alpha}^{t}\right) \notag\\
    \stackrel{P} \simeq \Phi^{'}_{N}\left(\tilde{\mathbf{A}}\bQ^{t}_{\perp}+\boldsymbol{\mathcal{H}}_{t-1}\boldsymbol{\alpha}^{t,*}\right)
\end{align}
where $\boldsymbol{\alpha}_{t}^{*} = \lim_{N \to \infty} \boldsymbol{\alpha}_{t}$ which are finite matrices, and $\boldsymbol{\alpha}_{t}^{*} \in \mathbb{R}^{tq \times q}$. We write :
\begin{equation}
    \begin{bmatrix}
    (\boldsymbol{\alpha}_{t}^{*})_{1} \\... \\(\boldsymbol{\alpha}_{t}^{*})_{t}
    \end{bmatrix}
\end{equation}
where $\forall 1 \leqslant i \leqslant t$, $(\boldsymbol{\alpha}_{t}^{*})_{i} \in \mathbb{R}^{q \times q}$. Then 
\begin{align}
    \Phi^{'}_{N}\left(\tilde{\mathbf{A}}\bQ^{t}_{\perp}+\boldsymbol{\mathcal{H}}_{t-1}\boldsymbol{\alpha}^{t,*}\right) &\stackrel{P} \simeq \Phi^{'}_{N}\left(\tilde{\mathbf{A}}\bQ^{t}_{\perp}+\boldsymbol{\mathcal{H}}_{t-1}\boldsymbol{\alpha}^{t,*}\right) \notag \\
    &\stackrel{P} \simeq \Phi(\bX^{0},\bH^{1},...,\bH^{t},\tilde{\mathbf{A}}\bQ^{t}_{\perp}+\boldsymbol{\mathcal{H}}_{t-1}\boldsymbol{\alpha}^{t,*})
\end{align}
Using Lemma \ref{pseudo-lip-conv}, there exists $\bZ^{t+1}_{\perp} \sim \mathbf{N}(0,\boldsymbol{\kappa}_{\perp}^{t+1} \otimes I_{N})$ independent of $\mathfrak{S}_{t}$, where
$\boldsymbol{\kappa}_{\perp}^{t+1} = \lim_{N \to \infty} \frac{1}{N}(\bQ^{t}_{\perp})^{\top}\bQ^{t}_{\perp}$, such that:
\begin{align}
    \Phi(\bX^{0},\bH^{1},...,\bH^{t},\tilde{\mathbf{A}}\bQ^{t}_{\perp}+\boldsymbol{\mathcal{H}}_{t-1}\boldsymbol{\alpha}^{t,*}) &\stackrel{P} \simeq \mathbb{E}_{\bZ}\left[\Phi(\bX^{0},\bH^{1},...,\bH^{t},\bZ^{t+1}_{\perp}+\boldsymbol{\mathcal{H}}_{t-1}\boldsymbol{\alpha}^{t,*})\right] \notag \\
    &\stackrel{P} \simeq \mathbb{E}\left[\Phi_{N}(\bX^{0},\bZ^{1},...,\bZ^{t},\bZ^{t+1}_{\perp}+\sum_{i=1}^{t}\bZ^{i}(\boldsymbol{\alpha}^{t,*})_{i})\right] 
\end{align}

We now need to match the covariance matrices defined by the prescription of $\bZ^{t+1}$ we obtained with the ones from the state evolution. Let $\bZ^{t+1} = \bZ^{t+1}_{\perp}+\sum_{i=1}^{t}\bZ^{i}(\boldsymbol{\alpha}^{t,*})_{i}) \in \mathbb{R}^{q\times q}$. We then write $\bZ^{t+1} \sim \mathbf{N}(0,\boldsymbol{\kappa}^{t+1,t+1} \otimes \mathbf{I}_{N})$ where $\boldsymbol{\kappa}^{t+1,t+1} = \lim_{N \to \infty} \frac{1}{N} (\bZ^{t+1})^{\top}\bZ^{t+1}$. Then, using the isometry proved above and remembering that, for any $1 \leqslant i \leqslant t, \bQ^{t} = f^{t}(\bH^{t})$:
\begin{equation}
    \frac{1}{N}(\bZ^{t+1})^{\top}\bZ^{t+1} \stackrel{P} \simeq \frac{1}{N}(\bH^{t+1})^{\top}\bH^{t+1} \stackrel{P} \simeq  \frac{1}{N}(\bQ^{t})^{\top}\bQ^{t} \xrightarrow[n \to \infty]{P} \boldsymbol{\kappa}^{t+1,t+1}
\end{equation}
similarly, for $s\geqslant 2$:
\begin{equation}
     \boldsymbol{\kappa}^{s} = \frac{1}{N}(\bZ^{s})^{\top}\bZ^{t+1} \stackrel{P} \simeq \frac{1}{N}(\bH^{s})^{\top}\bH^{t+1} \stackrel{P} \simeq  \frac{1}{N}(\bQ^{s-1})^{\top}\bQ^{t} \xrightarrow[n \to \infty]{P} \boldsymbol{\kappa}^{s,t+1}
\end{equation}

and for $s=1$:
\begin{equation}
     \boldsymbol{\kappa}^{s} = \frac{1}{N}(\bZ^{1})^{\top}\bZ^{t+1} \stackrel{P} \simeq \frac{1}{N}(\bH^{1})^{\top}\bH^{t+1} \stackrel{P} \simeq  \frac{1}{N}(\bQ^{0})^{\top}\bQ^{t} \xrightarrow[n \to \infty]{P} \boldsymbol{\kappa}^{1,t+1}
\end{equation}
\end{enumerate}
\end{proof}
\begin{proof}[Proof of Lemma \ref{lemma:LoAMP_approx1}]
This lemma is proven by induction.
\paragraph{Initialization.}
The first iterates read $\bH^{1} = \mathbf{A}\bQ^{0}$ and $\hat{\bH}^{1} = \mathbf{A}\bQ^{0}$. This concludes the initialization.
\paragraph{Induction.}
Assume the proposition is true up to time t. Define the $(t+1)q \times (t+1)q$ block-diagonal matrix $\boldsymbol{\boldsymbol{\mathcal{B}}}_{t} = \mbox{diag}\left(0_{q \times q},\bb^{1},...,\bb^{t}\right)$ and $\hat{\boldsymbol{\mathcal{H}}}_{t-1} = \left[\hat{\bH}^{1} \vert \hat{\bH}^{2} \vert ...\vert \hat{\bH}^{t}\right]$. We then have :
\begin{align}
    \bH^{t+1} &= \mathbf{P}^{\perp}_{\boldsymbol{\mathcal{Q}}_{t-1}}\mathbf{A}\mathbf{P}^{\perp}_{\boldsymbol{\mathcal{Q}}_{t-1}}\bQ^{t}+\boldsymbol{\mathcal{H}}_{t-1}\boldsymbol{\alpha}^{t} \notag \\
    &= \mathbf{A}\bQ^{t}_{\perp}-\mathbf{P}_{\boldsymbol{\mathcal{Q}}_{t-1}}\mathbf{A}\bQ^{t}_{\perp}+\boldsymbol{\mathcal{H}}_{t-1}\boldsymbol{\alpha}^{t}
\end{align}
and
\begin{align}
    \hat{\bH}^{t+1} &= \mathbf{A}\bQ^{t}-\bQ^{t-1}(\bb^{t})^{\top} \notag\\
    &= \mathbf{A}\bQ^{t}_{\perp}+\mathbf{A}\bQ^{t}_{\parallel}-\bQ^{t-1}(\bb^{t})^{\top} \notag\\
    \mbox{where} \quad \mathbf{A}\bQ^{t}_{\parallel} &= \mathbf{A}\boldsymbol{\mathcal{Q}}_{t-1}(\boldsymbol{\mathcal{Q}}_{t-1}^{\top}\boldsymbol{\mathcal{Q}}_{t-1})^{-1}\boldsymbol{\mathcal{Q}}_{t-1}^{\top}\bQ^{t} \notag\\
    &=\mathbf{A}\boldsymbol{\mathcal{Q}}_{t-1}\boldsymbol{\alpha}^{t}
\end{align}
which gives 
\begin{align}
    \hat{\bH}^{t+1}-\bH^{t+1} &= \mathbf{P}_{\boldsymbol{\mathcal{Q}}_{t-1}}\mathbf{A}\bQ^{t}_{\perp}-\bQ^{t-1}(\bb^{t})^{\top}+\mathbf{A}\boldsymbol{\mathcal{Q}}_{t-1}\boldsymbol{\alpha}^{t}-\boldsymbol{\mathcal{H}}_{t-1}\boldsymbol{\alpha}^{t}
\end{align}
using the definition of iteration (\ref{approx_AMP}), we have:
\begin{equation}
    \mathbf{A}\boldsymbol{\mathcal{Q}}_{t-1} = \hat{\boldsymbol{\mathcal{H}}}_{t-1}+\left[0_{N\times q}\vert \bQ^{0} \vert ... \vert \bQ^{t-2}\right]\boldsymbol{\mathcal{B}}_{t-1}^{\top}
\end{equation} 
\begin{align}
    \hat{\bH}^{t+1}-\bH^{t+1} &=\mathbf{P}_{\boldsymbol{\mathcal{Q}}_{t-1}}\mathbf{A}\bQ^{t}_{\perp}-\bQ^{t-1}(\bb^{t-1})^{\top}+\left[0_{N\times q}\vert \boldsymbol{\mathcal{Q}}_{t-2}\right]\boldsymbol{\mathcal{B}}_{t-1}^{\top}\boldsymbol{\alpha}^{t}+\left(\hat{\boldsymbol{\mathcal{H}}}^{t-1}-\boldsymbol{\mathcal{H}}^{t-1}\right)\boldsymbol{\alpha}^{t} \notag\\
    &= \boldsymbol{\mathcal{Q}}_{t-1}(\boldsymbol{\mathcal{Q}}_{t-1}^{\top}\boldsymbol{\mathcal{Q}}_{t-1})^{-1}\boldsymbol{\mathcal{Q}}_{t-1}^{\top}\mathbf{A}\bQ^{t}_{\perp}-\bQ^{t-1}(\bb_{t-1})^{\top}+\left[0_{N\times q}\vert \boldsymbol{\mathcal{Q}}_{t-2}\right]\boldsymbol{\mathcal{B}}_{t-1}^{\top}\boldsymbol{\alpha}^{t} \notag\\ &\hspace{6cm}+\left(\hat{\boldsymbol{\mathcal{H}}}^{t-1}-\boldsymbol{\mathcal{H}}^{t-1}\right)\boldsymbol{\alpha}^{t}
\end{align}

and
\begin{align}
    \boldsymbol{\mathcal{Q}}_{t-1}^{\top}\mathbf{A} &=(\mathbf{A}\boldsymbol{\mathcal{Q}}_{t-1})^{\top} \notag\\
    &= ((\hat{\boldsymbol{\mathcal{H}}}_{t-1}+\left[0_{N\times q} \vert \boldsymbol{\mathcal{Q}}_{t-2}\right]\boldsymbol{\mathcal{B}}_{t}^{\top}))^{\top} \notag\\
    &= \hat{\boldsymbol{\mathcal{H}}}_{t-1}^{\top}+\boldsymbol{\mathcal{B}}_{t}\left[0_{N\times q} \vert \boldsymbol{\mathcal{Q}}_{t-2}\right]^{\top}
\end{align}

since $\bQ^{t}_{\perp} = \mathbf{P}_{\bQ_{t-1}}^{\perp}\bQ^{t}$, it holds that:
\begin{align}
    \boldsymbol{\mathcal{Q}}_{t-1}^{\top}\mathbf{A}\bQ^{t}_{\perp} &= \left(\hat{\boldsymbol{\mathcal{H}}}_{t-1}^{\top}+\boldsymbol{\mathcal{B}}_{t}\left[0_{N\times q} \vert \boldsymbol{\mathcal{Q}}_{t-2}\right]^{\top}\right)\mathbf{P}_{\boldsymbol{\mathcal{Q}}_{t-1}}^{\perp}\bQ^{t} \notag \\
    &=\hat{\boldsymbol{\mathcal{H}}}_{t-1}^{\top}\mathbf{P}_{\boldsymbol{\mathcal{Q}}_{t-1}}^{\perp}\bQ^{t}
\end{align}
which in turn gives:
\begin{align}
    \hat{\bH}^{t+1}-\bH^{t+1} &= \boldsymbol{\mathcal{Q}}_{t-1}(\boldsymbol{\mathcal{Q}}_{t-1}^{\top}\boldsymbol{\mathcal{Q}}_{t-1})^{-1}\hat{\boldsymbol{\mathcal{H}}}_{t-1}^{\top}\bQ^{t}_{\perp}-\bQ^{t-1}(\bb^{t-1})^{\top}+\left[0_{N\times q}\vert \boldsymbol{\mathcal{Q}}_{t-2}\right]\boldsymbol{\mathcal{B}}_{t-1}^{\top}\boldsymbol{\alpha}^{t} \notag\\ &+\left(\hat{\boldsymbol{\mathcal{H}}}^{t-1}-\boldsymbol{\mathcal{H}}^{t-1}\right)\boldsymbol{\alpha}^{t} \notag\\
    &= \boldsymbol{\mathcal{Q}}_{t-1}(\boldsymbol{\mathcal{Q}}_{t-1}^{\top}\boldsymbol{\mathcal{Q}}_{t-1})^{-1}\boldsymbol{\mathcal{H}}_{t-1}^{\top}\bQ^{t}_{\perp}-\bQ^{t-1}(\bb^{t-1})^{\top}+\left[0_{N\times q} \vert \boldsymbol{\mathcal{Q}}_{t-2}\right]\boldsymbol{\mathcal{B}}_{t-1}^{\top}\boldsymbol{\alpha}^{t} \notag\\ &+\left(\hat{\boldsymbol{\mathcal{H}}}^{t-1}-\boldsymbol{\mathcal{H}}^{t-1}\right)\boldsymbol{\alpha}^{t}+\boldsymbol{\mathcal{Q}}_{t-1}(\boldsymbol{\mathcal{Q}}_{t-1}^{\top}\boldsymbol{\mathcal{Q}}_{t-1})^{-1}\left(\hat{\boldsymbol{\mathcal{H}}}_{t-1}-\boldsymbol{\mathcal{H}}_{t-1}\right)^{\top}\bQ^{t}_{\perp}
\end{align}
We now study the limiting behaviour of this quantity, starting with:
\begin{equation}
    \mathbf{C} = \boldsymbol{\mathcal{Q}}_{t-1}(\boldsymbol{\mathcal{Q}}_{t-1}^{\top}\boldsymbol{\mathcal{Q}}_{t-1})^{-1}\boldsymbol{\mathcal{H}}_{t-1}^{\top}\bQ^{t}_{\perp}-\bQ^{t-1}(\bb^{t-1})^{\top}+\left[0_{N\times q} \vert \boldsymbol{\mathcal{Q}}_{t-2}\right]\boldsymbol{\mathcal{B}}_{t-1}^{\top}\boldsymbol{\alpha}^{t} 
\end{equation}
We have :
\begin{align}
    \bQ^{t}_{\perp} &= \bQ^{t}-\bQ^{t}_{\parallel} \notag \\
    &= \bQ^{t}-\boldsymbol{\mathcal{Q}}_{t-1}\boldsymbol{\alpha}_{t}
\end{align}
and :
\begin{equation}
    \mathbf{C} = \boldsymbol{\mathcal{Q}}_{t-1}(\boldsymbol{\mathcal{Q}}_{t-1}^{\top}\boldsymbol{\mathcal{Q}}_{t-1})^{-1}\boldsymbol{\mathcal{H}}_{t-1}^{\top}(\bQ^{t}-\boldsymbol{\mathcal{Q}}_{t-1}\boldsymbol{\alpha}_{t} )-\bQ^{t-1}(\bb^{t-1})^{\top}+\left[0_{N\times q} \vert \boldsymbol{\mathcal{Q}}_{t-2}\right]\boldsymbol{\mathcal{B}}_{t-1}^{\top}\boldsymbol{\alpha}^{t}
\end{equation}

Using Lemma \ref{matrix-stein}, the state evolution, and the concentration properties of pseudo-Lipschitz functions Lemma \ref{pseudo-lip-conv}, we get, for all $1 \leqslant j \leqslant t-1$ and $1 \leqslant i \leqslant t$:
\begin{align}
    \frac{1}{N}(\bH^{i})^{\top}f^{j}(\bH^{j}) &\stackrel{P} \simeq \mathbb{E}\left[\frac{1}{N}(\bZ^{i})^{\top}f^{j}(\bZ^{j})\right] \notag\\
    &= \mathbf{K}_{i,j}\mathbb{E}\left[\frac{1}{N}\mbox{div}f^{j}(\bZ^{j})\right] \notag \\
    &\stackrel{P} \simeq \frac{1}{N}(\bQ^{i-1})^{\top}\bQ^{j-1}(\bb^{j})^{\top}
\end{align}
and for $j=0$ :
\begin{align}
    \frac{1}{N}(\bH^{i})^{\top}f(\bX^{0}) &\stackrel{P} \simeq \mathbb{E}\left[\frac{1}{N}(\bZ^{i})^{\top}f_{0}(\bX^{0})\right] \notag \\
    &= 0
\end{align}
which in turn gives
\begin{equation}
    \frac{1}{N}(\boldsymbol{\mathcal{H}}_{t-1}^{\top}\bQ^{t})=\frac{1}{N}\left[\bH^{1} \vert ... \vert \bH^{t}\right]^{\top}f_{t}(\bH^{t}) \stackrel{P} \simeq  \frac{1}{N}(\boldsymbol{\mathcal{Q}}_{t-1})^{\top}\bQ^{t-1}(\bb^{t-1})^{\top}
\end{equation}
and 
\begin{equation}
    \frac{1}{N}\boldsymbol{\mathcal{H}}_{t-1}^{\top}\boldsymbol{\mathcal{Q}}_{t-1} =\frac{1}{N}\left[\bH^{1} \vert ... \vert \bH^{t}\right]^{\top}\left[\bQ^{0} \vert f_{1}(\bH^{1}) ... \vert f_{t-1}(\bH^{t-1})\right] \stackrel{P} \simeq \frac{1}{N}\boldsymbol{\mathcal{Q}}_{t-1}^{\top}[0_{N\times q}\vert \boldsymbol{\mathcal{Q}}_{t-2}]\boldsymbol{\mathcal{B}}_{t-1}^{\top} 
\end{equation}
which gives :
\begin{align}
    \frac{1}{N}\mathbf{C} &\stackrel{P} \simeq \frac{1}{N}\bigg(\boldsymbol{\mathcal{Q}}_{t-1}(\boldsymbol{\mathcal{Q}}_{t-1}^{\top}\boldsymbol{\mathcal{Q}}_{t-1})^{-1}\boldsymbol{\mathcal{Q}}_{t-1}^{\top}(\bQ^{t-1}(\bb^{t-1})^{\top}-\left[0_{N\times q} \vert \boldsymbol{\mathcal{Q}}_{t-2}\right]\boldsymbol{\mathcal{B}}_{t-1}^{\top}\boldsymbol{\alpha}^{t} )\\
    &\hspace{6cm}-\bQ^{t-1}(\bb^{t-1})^{\top}+\left[0_{N\times q} \vert \boldsymbol{\mathcal{Q}}_{t-2}\right]\boldsymbol{\mathcal{B}}_{t-1}^{\top}\boldsymbol{\alpha}^{t}\bigg) \notag \\
    &=\frac{1}{N}\bigg(\boldsymbol{\mathcal{Q}}_{t-1}(\boldsymbol{\mathcal{Q}}_{t-1}^{\top}\boldsymbol{\mathcal{Q}}_{t-1})^{-1}\boldsymbol{\mathcal{Q}}_{t-1}^{\top}\underbrace{(\bQ^{t-1}(\bb^{t-1})^{\top}-\left[0_{N\times q} \vert \boldsymbol{\mathcal{Q}}_{t-2}\right]\boldsymbol{\mathcal{B}}_{t-1}^{\top}\boldsymbol{\alpha}^{t} )}_{\in \thickspace \mbox{span}(\boldsymbol{\mathcal{Q}}_{t-1})}\\
    &\hspace{6cm}-\bQ^{t-1}(\bb^{t-1})^{\top}+\left[0_{N\times q} \vert \boldsymbol{\mathcal{Q}}_{t-2}\right]\boldsymbol{\mathcal{B}}_{t-1}^{\top}\boldsymbol{\alpha}^{t}\bigg) \notag \\
    &=0
\end{align}
At this point, we have :
\begin{align}
     &\frac{1}{\sqrt{N}}\norm{\hat{\bH}^{t+1}-\bH^{t+1}}_{F} \leqslant \frac{1}{\sqrt{N}}\norm{\mathbf{C}}_{F}+\frac{1}{\sqrt{N}}\norm{\left(\hat{\boldsymbol{\mathcal{H}}}_{t-1}-\boldsymbol{\mathcal{H}}_{t-1}\right)\boldsymbol{\alpha}^{t} \notag \\
     &\hspace{3cm}+\boldsymbol{\mathcal{Q}}_{t-1}(\boldsymbol{\mathcal{Q}}_{t-1}^{\top}\boldsymbol{\mathcal{Q}}_{t-1})^{-1}\left(\hat{\boldsymbol{\mathcal{H}}}_{t-1}-\boldsymbol{\mathcal{H}}_{t-1}\right)^{\top}\bQ^{t}_{\perp}}_{F}
\end{align}
Where
\begin{align}
    \frac{1}{\sqrt{N}}\norm{\left(\hat{\boldsymbol{\mathcal{H}}}_{t-1}-\boldsymbol{\mathcal{H}}_{t-1}\right)\boldsymbol{\alpha}^{t}}_{F} \leqslant \frac{1}{\sqrt{N}}\norm{\hat{\boldsymbol{\mathcal{H}}}_{t-1}-\boldsymbol{\mathcal{H}}_{t-1}}_{F}\norm{\boldsymbol{\alpha}^{t}}_{F}
\end{align}
As previously discussed, $\norm{\boldsymbol{\alpha}^{t}}_{F}$ has a finite limit, and according to the induction hypothesis, $\frac{1}{\sqrt{N}}\norm{\hat{\boldsymbol{\mathcal{H}}}_{t-1}-\boldsymbol{\mathcal{H}}_{t-1}}_{F} \xrightarrow[N \to \infty]{P} 0$. Then 
\begin{align}
    \frac{1}{\sqrt{N}}\norm{\boldsymbol{\mathcal{Q}}_{t-1}(\boldsymbol{\mathcal{Q}}_{t-1}^{\top}\boldsymbol{\mathcal{Q}}_{t-1})^{-1}\left(\hat{\boldsymbol{\mathcal{H}}}_{t-1}-\boldsymbol{\mathcal{H}}_{t-1}\right)^{\top}\bQ^{t}_{\perp}}_{F} \leqslant \frac{1}{\sqrt{N}}\norm{\hat{\boldsymbol{\mathcal{H}}}_{t-1}-\boldsymbol{\mathcal{H}}_{t-1}}_{F}\frac{1}{Nc_{t}^{2}}\norm{\boldsymbol{\mathcal{Q}}_{t-1}}_{F}\norm{\bQ^{t}}_{F}
\end{align}
where $\frac{1}{Nc_{t}^{2}}\norm{\boldsymbol{\mathcal{Q}}_{t-1}}_{F}\norm{\bQ^{t}}_{F}$ converges to a finite limit due to the state evolution proved above. This ultimately shows that 
\begin{equation}
    \frac{1}{\sqrt{N}}\norm{\hat{\bH}^{t+1}-\bH^{t+1}}_{F} \xrightarrow[N \to \infty]{P}0
\end{equation}
and concludes the induction.
\end{proof}
\begin{proof}[Proof of Lemma \ref{lemma:LoAMP_approx2}]
This one is another induction. 
Let $S_{t}$ be the statement 
$\frac{1}{\sqrt{N}}\norm{\bQ^{t}-\bM^{t}}_{F} \xrightarrow[N \to \infty]{P} 0$ and $\frac{1}{\sqrt{N}}\norm{\bH^{t+1}-\bX^{t+1}}_{F} \xrightarrow[N \to \infty]{P} 0$.
\paragraph{Initialization.}
We have $\bQ^{0} = f^{0}(\bX^{0}) = \bM^{0}$ and $\bH^{1} = \mathbf{A}\bQ^{0}, \bX^{1}=\mathbf{A}\bM^{0}$.
\paragraph{Induction.}
We assume $S_{t-1}$ is true, and we prove $S_{t}$. We have
\begin{align}
    \frac{1}{\sqrt{N}}\norm{\bQ^{t}-\bM^{t}}_{F} &= \frac{1}{\sqrt{N}}\norm{f^{t}(\bH^{t})-f^{t}(\bX^{t})}_{F} \notag \\ &\leqslant L_{t}\left(1+\left(\frac{\norm{\bH^{t}}_{F}}{\sqrt{N}}\right)^{k-1}+\left(\frac{\norm{\bX^{t}}_{F}}{\sqrt{N}}\right)^{k-1}\right)\frac{\norm{\bH^{t}-\bX^{t}}_{F}}{\sqrt{N}}
\end{align}
which goes to zero as n goes to infinity from the induction hypothesis. We then prove that $\frac{1}{\sqrt{N}}\norm{\hat{\bH}^{t+1}-\bX^{t+1}}_{F} \xrightarrow[N \to \infty]{P}0$.
\begin{align}
    \hat{\bH}^{t+1}-\bX^{t+1} = \mathbf{A}\bQ^{t}-\bQ^{t-1}(\bb^{t})^{\top}-\mathbf{A}\bM^{t}+\bM^{t-1}(\bb^{t})^{\top}
\end{align}
and
\begin{equation}
    \frac{1}{\sqrt{N}}\norm{\hat{\bH}^{t+1}-\bX^{t+1}}_{F} \leqslant \norm{\mathbf{A}}_{op}\frac{1}{\sqrt{N}}\norm{\bQ^{t}-\bM^{t}}_{F}+\frac{1}{\sqrt{N}}\norm{\bQ^{t-1}-\bM^{t-1}}_{F}\norm{\bb^{t}}_{F}
\end{equation}
using Proposition \ref{op-norm-GOE}, $\norm{\mathbf{A}}_{op} \xrightarrow[N \to \infty]{P} 2$. Using the induction hypothesis, $\frac{1}{\sqrt{N}}\norm{\bQ^{t}-\bM^{t}}_{F}\xrightarrow[N \to \infty]{P} 0,\frac{1}{\sqrt{N}}\norm{\bQ^{t-1}-\bM^{t-1}}_{F}\xrightarrow[N \to \infty]{P} 0$, and $\norm{\bb^{t}}_{F}$ is finite. This concludes the induction step.
\end{proof}
\begin{proof}[Proof of Lemma \ref{lemma:SE_pert_def}]
In this proof, we will consider the $2q\times2q$ covariance matrix $\boldsymbol{\kappa} = \begin{bmatrix}
    \boldsymbol{\kappa}^{1,1} \thickspace \boldsymbol{\kappa}^{1,2} \\
    \boldsymbol{\kappa}^{1,2} \thickspace  \boldsymbol{\kappa}^{2,2}\end{bmatrix}$ and two matrices $\bZ^{1},\bZ^{2} \in (\mathbb{R}^{N\times q})^{2}$ following the distribution $\mathbf{N}\left(0,\boldsymbol{\kappa}\otimes \mathbf{I}_{N}\right)$  
, and we study the corresponding state evolution when the perturbed functions $f_{\epsilon \bY}^{t}$ are considered. We drop the $\epsilon$ exponent on the covariance matrices since we are just studying the well-definiteness of the perturbed SE as an induction. The link with the original SE will be studied in subsequent lemmas.
\begin{align*}
\mathbb{E}_{\bZ}\left[\frac{1}{N}(f^{s}_{\epsilon \bY}(\bZ^{s})^{\top}f^{t}_{\epsilon \bY}(\bZ^{t})\right] &= \mathbb{E}_{\bZ}\left[\frac{1}{N}(f^{s}(\bZ^{s})^{\top}f^{t}(\bZ^{t})\right]+\epsilon\mathbb{E}_{\bZ}\left[\frac{1}{N}(f^{s}(\bZ^{s}))^{\top}\bY^{t}\right]\\
&+\epsilon\mathbb{E}_{\bZ}\left[\frac{1}{N}(f^{t}(\bZ^{t}))^{\top}\bY^{s}\right]+\epsilon^{2}\frac{1}{N}(\bY^{s})^{\top}\bY^{t} \\
&= \mathbb{E}_{\bZ}\left[\frac{1}{N}(f^{s}(\bZ^{s})^{\top}f^{t}(\bZ^{t})\right]+\frac{\epsilon}{N}\mathbb{E}_{\bZ}\left[f^{s}(\bZ^{s})\right]^{\top}\bY^{t}\\
&+\frac{\epsilon}{N}\mathbb{E}_{\bZ}\left[f^{t}(\bZ^{t})\right]^{\top}\bY^{s}+\frac{\epsilon^{2}}{N}(\bY^{s})^{\top}\bY^{t}
\end{align*}
\begin{itemize}
    \item the first term does not depend on the perturbation and is deterministic. Using assumptions (A6), this quantity has a finite limit.
    \item second term is a $q \times q$ matrix where each element have zero mean and variance 
    \begin{equation}
        \mbox{Var}\left[\frac{1}{N}\left(\mathbb{E}\left[f^{s}(\bZ^{s})^{\top}\bY^{t}\right]_{j}^{i}\right)\right] = \frac{1}{N^{2}}\norm{\mathbb{E}\left[f^{s}(\bZ^{s})\right]}_{2}^{2} \leqslant \frac{C}{N}
    \end{equation}
Using the Gaussian tail and the Borel-Cantelli lemma, this term converges almost surely to zero.
\item the third term is treated in the same way as the second one

\item the last term follows from the strong law of large numbers:
\begin{equation}
\lim_{N \to \infty} \frac{1}{N}(\bY^{s})^{\top}\bY^{t} \xrightarrow[n \to \infty]{a.s.}\mathbf{I}_{q \times q}\delta_{s=t}
\end{equation}
\end{itemize}
Putting things together, we get, almost surely:
\begin{equation}
\label{eq:induc_pert_SE}
    \lim_{N \to \infty}\mathbb{E}_{\bZ}\left[ \frac{1}{N}(f^{s}_{\epsilon \bY}(\bZ^{s}))^{\top}f^{t}_{\epsilon \bY}(\bZ^{t})\right] = \lim_{N \to \infty}\mathbb{E}_{\bZ}\left[ \frac{1}{N}(f^{s}(\bZ^{s}))^{\top}f^{t}(\bZ^{t})\right]+\epsilon^{2}\mathbf{I}_{q \times q}\delta_{s=t}
\end{equation}
Verifying the initialization assumptions (A4-A5) is very similar to the previous steps, thus we directly give the result. The initialization reads: 
\begin{align}
\label{eq:pertubed-SE-init}
    &\lim_{N \to \infty} \frac{1}{N}(f^{0}_{\epsilon \bY}(\bX^{0}))^{\top}f^{0}_{\epsilon \bY}(\bX^{0}) = \lim_{N \to \infty} \frac{1}{N}(f^{0}(\bX^{0}))^{\top}f^{0}(\bX^{0})+\epsilon^{2}\mathbf{I}_{q \times q} \\
     &\lim_{N \to \infty} \frac{1}{N}\mathbb{E}\left[(f^{0}_{\epsilon \bY}(\bX^{0}))^{\top}f^{t}_{\epsilon \bY}(\bZ^{t})\right] = \lim_{N \to \infty} \frac{1}{N}\mathbb{E}\left[(f^{0}(\bX^{0}))^{\top}f^{t}(\bZ^{t})\right]
\end{align}
It follows straightforwardly from these equations and a short induction that the resulting state evolution is almost surely non-random.
\end{proof}
\begin{proof}[Proof of Lemma \ref{lemma:pert_full_rank}]
By definition, for any $t\in \mathbb{N}$ :
\begin{equation}
 \bQ^{t,\epsilon \bY} = \bQ^{t}+\epsilon \bY^{t}
\end{equation}
Then
\begin{equation}
    \bQ^{\epsilon \bY,t}_{\perp} = \mathbf{P}^{\perp}_{\boldsymbol{\mathcal{Q}}_{t-1}^{\epsilon \bY}}f^{t}(\bH^{\epsilon \bY,t})+\epsilon \mathbf{P}^{\perp}_{\boldsymbol{\mathcal{Q}}_{t-1}^{\epsilon \bY}}\bY^{t}
\end{equation}
with the parallel term a linear combination of the previous ones.
Denote $\mathcal{F}_{t}$ the $\sigma$-algebra generated by $\bH^{\epsilon \bY,1}, ...,\bH^{\epsilon \bY,t},\bY^{1},...,\bY^{t-1}$. Since $\bY^{t}$ is generated independently of $\mathcal{F}_{t}$, each column $j$ of $\bQ^{\epsilon \bY,t}$ obeys the distribution:
\begin{equation}
    (\bQ^{\epsilon \bY,t}_{\perp})_{j}\vert_{\mathcal{F}_{t}} \sim \mathbf{N}(\mathbf{P}^{\perp}_{\boldsymbol{\mathcal{Q}}_{t-1}^{\epsilon \bY}}(f^{t}(\bH^{\epsilon \bY,t}))_{j},\epsilon^{2}\mathbf{P}^{\perp}_{\boldsymbol{\mathcal{Q}}_{t-1}^{\epsilon \bY}})
\end{equation}
the variance of which is almost surely non-zero whenever $N \geqslant tq$. Thus, when $N\geqslant tq$, the matrix $\boldsymbol{\mathcal{Q}}_{t-1}$ has full column rank.
We now need to control the minimal singular value of $\boldsymbol{\mathcal{Q}}_{t-1}$. Following \cite{bayati2011dynamics}, Lemma 9, we only need to check that, for any column $j$, almost surely, for N sufficiently large, there exists a constant $c_{\epsilon}>0$ such that:
\begin{equation}
\frac{1}{N}\norm{(\bQ_{\perp}^{\epsilon \bY,t})_{j}}^{2} \geqslant c_{\epsilon}
\end{equation}
which follows in almost identical fashion to \cite{berthier2020state}, Lemma 9 using the moments of a $N-tq$ chi-square variable, instead of $N-t$ in the original proof, which extends straightforwardly since $q$ is kept finite.
\end{proof}
\begin{proof}[Proof of Lemma \ref{lemma:unif_conv}]
This result is proven for $q=1$ in \cite{berthier2020state} and the proof for the case of finite, integer $q$ is identical.   
\end{proof}
\begin{proof}[Proof of Lemma \ref{lemma:conv_SE}]
This lemma is proven by induction.
\paragraph{Initialization.}
From equation (\ref{eq:pertubed-SE-init}), it holds that
\begin{equation}
    \mathbf{K}_{1,1}^{\epsilon} = \mathbf{K}_{1,1}+\epsilon^{2} \xrightarrow[\epsilon \to 0]{} \mathbf{K}_{1,1}
\end{equation}
\paragraph{Induction.}
Let t be a non-negative integer. Assume that, for any $r,s\leqslant t$, $\boldsymbol{\kappa}_{\epsilon}^{r,s} \to \boldsymbol{\kappa}^{r,s}$. Then:
\begin{equation}
    \boldsymbol{\kappa}_{\epsilon}^{s+1,t+1} = \lim_{N \to \infty} \mathbb{E}\left[\frac{1}{N}(f_{\epsilon \bY}^{s}(\bZ_{\epsilon \bY}^{s}))^{\top}f^{t}_{\epsilon \bY}(\bZ_{\epsilon \bY}^{t})\right]
\end{equation}
where $\bZ_{\epsilon \bY}^{s},\bZ_{\epsilon \bY}^{t}$ are $n\times q$ Gaussian random matrices whose distributions are specified by $\boldsymbol{\kappa}_{\epsilon}^{s,s},\boldsymbol{\kappa}^{t,t}_{\epsilon}$ and $\boldsymbol{\kappa}_{\epsilon}^{s,t}$ which are $q\times q$ deterministic matrices. Then, from equation (\ref{eq:induc_pert_SE}), we have 
\begin{equation}
    \boldsymbol{\kappa}_{\epsilon}^{s+1,t+1} = \lim_{N \to \infty} \frac{1}{N} \mathbb{E}_{\bZ}\left[\frac{1}{N}(f_{s}(\bZ^{\epsilon,s}))^{\top}f_{t}(\bZ^{\epsilon, t})\right]+\epsilon^{2}\mathbf{I}_{q\times q}\delta_{s=t}
\end{equation}
From Lemma \ref{lemma:pseudo-lip_product}, the function $(\bZ^{s},\bZ^{t})\to \frac{1}{N}f_{s}(\bZ^{s})^{\top}f_{t}(\bZ^{t})$ is uniformly pseudo-Lipschitz. Moreover, from the induction hypothesis, we have :
\begin{equation}
    \lim_{\epsilon \to 0}\boldsymbol{\kappa}_{\epsilon}^{s,t} = \boldsymbol{\kappa}^{s,t}
\end{equation}
thus, using the uniform convergence Lemma \ref{lemma:unif_conv}, we get :
\begin{equation}
    \lim_{\epsilon \to 0}\lim_{N \to \infty} \frac{1}{N}\mathbb{E}\left[f_{s}(\bZ^{\epsilon,s})^{\top}f_{t}(\bZ^{\epsilon,t})\right] = \lim_{N \to \infty}\frac{1}{N}\mathbb{E}\left[f_{s}(\bZ^{s})^{\top}f_{t}(\bZ^{t})\right] = \boldsymbol{\kappa}^{s+1,t+1}
\end{equation}
where $(\bZ^{s},\bZ^{t}) \sim \mathbb{N}(0,\boldsymbol{\kappa}\otimes \mathbf{I}_{n})$ and $\boldsymbol{\kappa} = \begin{bmatrix}\boldsymbol{\kappa}^{s,s},\boldsymbol{\kappa}^{s,t} \\\boldsymbol{\kappa}^{t,s}, \boldsymbol{\kappa}^{t,t}\end{bmatrix}$. This shows that
\begin{equation}
    \boldsymbol{\kappa}^{s+1,t+1}_{\epsilon} \xrightarrow[\epsilon \to 0]{} \boldsymbol{\kappa}^{s+1,t+1}
\end{equation}
which concludes the induction. Similar reasoning proves the convergence of correlations with the initial vector
\begin{equation}
    \boldsymbol{\kappa}_{\epsilon}^{1,t+1} \xrightarrow[\epsilon \to 0]{}\boldsymbol{\kappa}^{1,t+1}
\end{equation}.
\end{proof}
\begin{proof}[Proof of Lemma \ref{lemma:conv_AMP}] This Lemma is proven by induction.
\paragraph{Initialization.}
\begin{align}
\frac{1}{\sqrt{N}}\norm{\bM^{\epsilon \bY,0}-\bM^{0}}_{F} = f^{0}_{\epsilon \bY}(\bX^{0})-f^{0}(\bX^{0}) = \frac{1}{\sqrt{N}}\epsilon\norm{\bY^{0}}_{F}
\end{align}
Using the bound from Lemma $\ref{prop:norm_Gauss}$, there exists an absolute constant $C_{\bY}$ independent of $N$ such that, with high probability:
\begin{equation}
 \frac{\epsilon}{\sqrt{N}}\norm{\bY^{0}}_{F} \leqslant C_{\bY} \epsilon
\end{equation}
Note that $C_{\bY}$ is the same for all $\bY^{t}$. We thus choose $h'_{0}(\epsilon) = C_{\bY}\epsilon$. Then
\begin{align}
    \frac{1}{\sqrt{N}}\norm{\bX^{\epsilon \bY,1}-\bX^{1}}_{F} \leqslant \norm{\mathbf{A}}_{op}\frac{\epsilon}{\sqrt{N}}\norm{\bY^{0}}_{F} \leqslant 2C_{\bY}\epsilon
\end{align}
using the bound on the operator norm of GOE matrices Proposition \ref{op-norm-GOE}, and we can choose $h_{0}(\epsilon) = 2C_{\bY}\epsilon$.  
\paragraph{Induction}
Assume the property is verified up to time $t$, i.e., the functions $h_{0}(\epsilon),h'_{0}(\epsilon),...,h_{t-1}(\epsilon), h'_{t-1}(\epsilon)$ exist and are known. We now need to show $h_{t}(\epsilon),h'_{t}(\epsilon)$ exist. By definition of the iteration:
\begin{align}
\frac{1}{\sqrt{N}}\norm{\bM^{\epsilon \bY,t}-\bM^{t}}_{F} &= \frac{1}{\sqrt{N}}\norm{f^{t}_{\epsilon \bY}(\bX^{\epsilon \bY})-f^{t}(\bX^{t})}_{F} \notag \\
&= \frac{1}{\sqrt{N}}\norm{f^{t}(\bX^{\epsilon \bY})-f^{t}(\bX^{t})+\epsilon \bY^{t}}_{F} \notag \\
&\leqslant L_{t}\left(1+\left(\frac{\norm{\bX^{\epsilon \bY,t}}_{F}}{\sqrt{N}}\right)^{k-1}+\left(\frac{\norm{\bX^{t}}_{F}}{\sqrt{N}}\right)^{k-1}\right)\frac{\norm{\bX^{\epsilon \bY,t}-\bX^{t}}_{F}}{\sqrt{N}}+\frac{1}{\sqrt{N}}\epsilon\norm{\bY^{t}}_{F} \notag \\
&\leqslant L_{t}\left(1+\left(\frac{\norm{\bX^{\epsilon \bY,t}}_{F}}{\sqrt{N}}\right)^{k-1}+\left(\frac{\norm{\bX^{t}}_{F}}{\sqrt{N}}\right)^{k-1}\right)h_{t-1}(\epsilon)+C_{\bY}\epsilon \notag \\
&\leqslant L_{t}\left(1+C_{\epsilon \bY}(k)+\left(\frac{\norm{\bX^{\epsilon \bY,t}}_{F}}{\sqrt{N}}+\frac{\norm{\bX^{\epsilon \bY,t}-\bX^{t}}_{F}}{\sqrt{N}}\right)^{k-1}\right)h_{t-1}(\epsilon)+C_{\bY}\epsilon \notag \\
&\leqslant L_{t}\left(1+C_{\epsilon \bY}(k)+2^{k-2}C_{\epsilon \bY}(k)^{k-1}+2^{k-2}h_{t-1}^{k-1}(\epsilon)\right)h_{t-1}(\epsilon)+C_{\bY}\epsilon \notag \\
\end{align}
where we used the state evolution of the perturbed AMP orbit to show that $\frac{\norm{\bX^{\epsilon \bY,t}}_{F}}{\sqrt{N}}$ has a finite limit and Hölder's inequality. We can thus choose
\begin{equation}
h'_{t}(\epsilon) = L_{t}\left(1+C_{\epsilon \bY}(k)+2^{k-2}C_{\epsilon \bY}(k)^{k-1}+2^{k-2}h_{t-1}^{k-1}(\epsilon)\right)h_{t-1}(\epsilon)+C_{\bY}\epsilon
\end{equation}
which goes to zero when $\epsilon$ goes to zero. Then
\begin{align}
    \frac{1}{\sqrt{N}}\norm{\bX^{\epsilon \bY,t+1}-\bX^{t+1}}_{F} &\leqslant \norm{\mathbf{A}}_{op}\frac{1}{\sqrt{N}}\norm{\bM^{\epsilon \bY,t}-\bM^{t}}_{F}+\frac{1}{\sqrt{N}}\norm{\bM^{\epsilon \bY,t-1}(\bb_{\epsilon \bY}^{t})^{\top}-\bM^{t-1}(\bb^{t})^{\top}}_{F} \notag \\
    & \leqslant 2h'_{t}(\epsilon)+\frac{1}{\sqrt{N}}\norm{\bM^{\epsilon \bY,t-1}(\bb^{\epsilon \bY}_{t})^{\top}-\bM^{t-1}(\bb^{t})^{\top}}_{F} \notag \\
    &\leqslant 2h'_{t}(\epsilon)+\frac{1}{\sqrt{N}}\norm{\bM^{\epsilon \bY,t-1}-\bM^{t-1}}_{F}\norm{\bb^{t}}_{F}+\frac{1}{\sqrt{N}}\norm{\bb_{\epsilon \bY}^{t}-\bb^{t}}_{F}\norm{\bM^{t-1}}_{F}
\end{align}
and
\begin{align}
    \norm{\bb^{t}}_{F} &= \norm{\mathbb{E}\left[\frac{1}{N} \sum_{i=1}^N \frac{\partial f^t_i}{\partial \bZ_i}(\bZ^t)\right]}_{F} \notag \\
    &\leqslant \mathbb{E}\left[\frac{1}{N}\sum_{i=1}^{N}\norm{\frac{\partial f^t_i}{\partial \bZ_i}(\bZ^t)}_{F}\right] 
\end{align}
where $\bZ^{t} \sim \mathbf{N}(0,\boldsymbol{\kappa}_{t,t}\otimes \mathbf{I}_{n})$. Since the function $f^{t}:\mathbb{R}^{N \times q}\to \mathbb{R}^{N \times q}$ is pseudo-Lipschitz of order $k$, the components $f^{t}_{i} : \mathbb{R}^{N \times q} \to \mathbb{R}^{q}$ are pseudo-Lipschitz of order $k$ as well. So are the functions $f_{i,j}^{t} : \mathbb{R}^{N \times q}\to \mathbb{R}$ for $1 \leqslant j \leqslant q$ generating each component of $f^{t}_{i}(\bZ^{t}) \in \mathbb{R}^{q}$ and their $\mathbb{R}^{q} \to \mathbb{R}$ restrictions to the $i-th$ line of $\bZ^{t}$. 
Then
\begin{align}
    \norm{\mathbf{b}^{t}}_{F} \leqslant \frac{1}{N}\sum_{i=1}^{N}q\max_{j}\left\{\mathbb{E}\norm{\nabla_{\bZ^{t}_{i}}f_{i,j}^{t}(\bZ^{t})}_{2}\right\}
\end{align}
where $\max_{j}\left\{\mathbb{E}\norm{\nabla_{\bZ^{t}_{i}}f_{i,j}^{t}(\bZ^{t})}_{2}\right\}$ is bounded using the pseudo-Lipschitz property and a similar argument to the proof of lemma \ref{pseudo-lip-conv}. Let $C_{J}$ be this upper bound, then
\begin{align}
    \frac{1}{\sqrt{N}}\norm{\bM^{\epsilon \bY,t-1}(\bb_{\epsilon \bY}^{t})^{\top}-\bM^{t-1}(\bb^{t})^{\top}}_{F} \leqslant qC_{J}h'_{t-1}(\epsilon)+\frac{1}{\sqrt{N}}\norm{\bM^{t-1}}_{F}\norm{\bb^{t}_{\epsilon \bY}-\bb^{t}}_{F}
\end{align}
Using the same decomposition as before
\begin{align}
    \frac{1}{\sqrt{N}}\norm{\bM^{t-1}}_{F}\norm{\bb^{t}_{\epsilon \bY}-\bb^{t}}_{F}&\leqslant \left(\frac{1}{\sqrt{N}}\norm{\bM^{\epsilon \bY,t-1}-\bM^{t-1}}+\frac{1}{\sqrt{N}}\norm{\bM^{\epsilon \bY,t-1}}\right)\norm{\bb^{t}_{\epsilon \bY}-\bb^{t}}_{F} \notag \\
    &\leqslant \left(h'_{t-1}(\epsilon)+C_{\epsilon \bY,t-1}\right)\norm{\bb^{t}_{\epsilon \bY}-\bb^{t}}_{F}
\end{align}
The definition of the Onsager correction terms gives
\begin{align}
    \norm{\bb^{t}_{\epsilon \bY}-\bb^{t}}_{F} &= \left\Vert {\mathbb{E}\left[\frac{1}{N} \sum_{i=1}^N \frac{\partial f^t_i}{\partial \tilde{\bZ}^{\epsilon \bY,t}_i}(\tilde{\bZ}^{\epsilon \bY,t})\right]-\mathbb{E}\left[\frac{1}{N} \sum_{i=1}^N \frac{\partial f^t_i}{\partial \tilde{\bZ}^{t}_i}(\tilde{\bZ}^{t})\right]} \right\Vert_{F}
\end{align}
where $\tilde{\bZ}^{\epsilon \bY,t} = \bZ(\boldsymbol{\kappa}_{t,t}^{\epsilon \bY})^{1/2}$ where $\bZ \in \mathbb{R}^{N \times q}$ is an i.i.d.~standard normal matrix. Similarly $\tilde{\bZ}^{t} = \bZ(\boldsymbol{\kappa}_{t,t})^{1/2}$. Using the positive definiteness of $\boldsymbol{\kappa}_{t,t}$ along with Lemma \ref{matrix-stein}, we can write, keeping in mind that the perturbation $\epsilon \bY$ doesn't change the derivatives in the Onsager correction:
\begin{align}
    \norm{\bb^{t}_{\epsilon \bY}-\bb^{t}}_{F} &=\norm{(\boldsymbol{\kappa}^{t,t}_{\epsilon \bY})^{-1}\mathbb{E}\left[\frac{1}{N}(\tilde{\bZ}^{\epsilon \bY,t})^{\top}f^{t}(\tilde{\bZ}^{\epsilon \bY,t})\right]-(\boldsymbol{\kappa}^{t,t})^{-1}\mathbb{E}\left[\frac{1}{N}(\bZ^{t})^{\top}f^{t}(\bZ^{t})\right]}_{F} \notag \\
    &\leqslant \norm{(\boldsymbol{\kappa}^{t,t}_{\epsilon \bY})^{-1}-(\boldsymbol{\kappa}^{t,t})^{-1}}_{F}\mathbb{E}\left[\frac{1}{N}(\tilde{\bZ}^{\epsilon \bY,t})^{\top}f^{t}(\tilde{\bZ}^{\epsilon \bY,t})\right]+\notag\\
    &\hspace{3cm}(\boldsymbol{(\kappa}^{t,t})^{-1})\norm{\mathbb{E}\left[\frac{1}{N}(\tilde{\bZ}^{\epsilon \bY,t})^{\top}f^{t}(\tilde{\bZ}^{\epsilon \bY,t})\right]-\mathbb{E}\left[\frac{1}{N}(\bZ^{t})^{\top}f^{t}(\bZ^{t})\right]}_{F}\notag
\end{align}
The function $\mathbb{R}^{N\times q}\to \mathbb{R}^{q\times q}, \bZ \to \bZ^{\top}f^{t}(\bZ)$ is pseudo-Lipschitz of order $k+1$. Moreover, from Lemma 8, $\boldsymbol{\kappa}_{\epsilon \bY}^{t,t}\xrightarrow[\epsilon \to 0]{}\boldsymbol{\kappa}^{t,t}$. Thus using Lemma \ref{lemma:unif_conv}, we get 
\begin{equation}
    \lim_{\epsilon \to 0}\norm{\mathbb{E}\left[\frac{1}{N}(\tilde{\bZ}^{\epsilon \bY,t})^{\top}f_{t}(\tilde{\bZ}^{\epsilon \bY,t})\right]-\mathbb{E}\left[\frac{1}{N}(\bZ^{t})^{\top}f_{t}(\bZ^{t})\right]}_{F}=0 \\
\end{equation}
and Lemma \ref{lemma:conv_SE} gives $\lim_{\epsilon \to 0} \norm{(\boldsymbol{\kappa}_{t,t}^{\epsilon \bY})^{-1}-(\boldsymbol{\kappa}_{t,t})^{-1}}_{F} = 0$, which concludes the induction.
\end{proof}
\section{Low-rank perturbations and projections}
\label{app:low_rank_pert}
As mentioned in Section \ref{sec:extensions}, AMP iterations associated to inference problems often present 
non-trivial dependencies between the non-linearities and the random matrices of the corresponding graph. These dependencies typically 
take the form of low-rank linear perturbations, or an additional argument in the non-linearities composed of a non-linear transform involving the random matrices of the graph, see the examples of Section \ref{sec:applications}. In this appendix, we propose a generic way of dealing with these dependencies by leveraging on the matrix-valued iteration Eq.(\ref{eq:sym-amp-iteration-1}-\ref{eq:sym-amp-iteration-2}), in the form of two lemmas.
\subsection{Additive low-rank perturbation}
\begin{lemma}
    \label{lemma:spike_SE}
    Let $\mathbf{V}_{0}\in \mathbb{R}^{N \times q}$ be a given matrix such that the quantity $\frac{1}{\sqrt{N}}\norm{\mathbf{V}_{0}}_{F}$ converges to a finite constant as $N \to \infty$. Define the matrix 
    \begin{equation}
        \hat{\mathbf{A}} = \mathbf{A}+\frac{1}{N}\mathbf{V}_{0}\mathbf{V}_{0}^{\top} \quad \in \mathbb{R}^{N \times N},
    \end{equation}
consider the AMP iteration initialized with $\mathbf{X}^{0} \in \mathbb{R}^{N \times q}$
\begin{align}
    \bX^{t+1} &= \hat{\mathbf{A}}\bM^{t}-\bM^{t-1}(\bb^{t})^{\top} && \in \R^{N\times q} \, ,  \\
    \bM^{t} &=f^{t}(\bX^{t}) && \in \R^{N\times q} \, , \\
    \bb^t &= \frac{1}{N} \sum_{i=1}^N \frac{\partial f^t_i}{\partial \bX_i}(\bX^t) && \in 
    \R^{q\times q}\, . 
\end{align}
and the following state evolution recursion, initialized with $\boldsymbol{\mu}_{0} = 0_{q \times q}$,
    \begin{align}
        \label{eq:spike_SE}
        \boldsymbol{\mu}_{0}, \thickspace \boldsymbol{\kappa}^{1,1} &= \lim_{N \to \infty} \frac{1}{N}  f^{0}(\mathbf{V}_{0}\boldsymbol{\mu}_{0}+\bX^{0})^{\top}f^{0}(\mathbf{V}_{0}\boldsymbol{\mu}_{0}+\bX^{0})\\
        \boldsymbol{\mu}^{s+1} &= \lim_{N \to +\infty} \frac{1}{N}\mathbb{E}\left[(\mathbf{V}_{0})^{\top}f^{s}\left(\mathbf{V}_{0}\boldsymbol{\mu}^{s}+\bZ^s\right)\right]\\
        \boldsymbol{\kappa}^{t+1, s+1} &= \boldsymbol{\kappa}^{s+1, t+1} = \lim_{N \to \infty} \frac{1}{N} \E\left[ f^s(\mathbf{V}_{0}\boldsymbol{\mu}^{s}+\bZ^s)^\top f^t(\mathbf{V}_{0}\boldsymbol{\mu}^{t}+\bZ^t) \right] \, , \qquad s \in \{ 0, \dots, t \} \, .
    \end{align}
    where $(\bZ^{1}, ..., \bZ^{t}) \sim \mathbf{N}\left(0,\left(\kappa^{s,r}\right)_{s,r\leqslant t} \otimes \mathbf{I}_{N}\right)$. Assume $\ref{it:ass-sym-1}-\ref{it:ass-sym-6}$ and that for any $t \in \mathbb{N}$,
    any $1 \leqslant i \leqslant N$, the derivative $\frac{\partial f^t_i}{\partial \bX_i}$ is pseudo-Lipschitz of order $k$.
    Then for any sequence $\phi_{N}:(\mathbb{R}^{N \times q})^{t+1} \to \mathbb{R}$ of pseudo-Lipschitz functions
    \begin{equation}
        \phi_{N}\left(\mathbf{X}^{0}, \mathbf{X}^{1}, ..., \mathbf{X}^{t}\right) \approxP \mathbb{E}\left[\phi_{N}\left(\mathbf{V}_{0}\boldsymbol{\mu}^{0}+\mathbf{Z}^{0},\mathbf{V}_{0}\boldsymbol{\mu}^{1}+\mathbf{Z}^{1},...,\mathbf{V}_{0}\boldsymbol{\mu}^{t}+\mathbf{Z}^{t}\right)\right]
    \end{equation}
\end{lemma}
\begin{proof}[Proof of Lemma \ref{lemma:spike_SE}]
    The proof follows a similar argument to that of Lemma 3.4 from \cite{deshpande2017asymptotic}.
    Consider the following iteration 
    \begin{align}
        \label{eq:spike_aux_it}
        \bS^{t+1} &= \mathbf{A}\tilde{\bM}^{t}-\tilde{\bM}^{t-1}(\tilde{\bb}^{t})^{\top} && \in \R^{N\times q} \, ,  \\
        \tilde{\bM}^{t} &=f^{t}(\mathbf{V}_{0}\boldsymbol{\mu}^{t}+\bS^{t}) && \in \R^{N\times q} \, , \\
        \tilde{\bb}^t &= \frac{1}{N} \sum_{i=1}^N \frac{\partial f^t_i}{\partial \bS_i}(\mathbf{V}_{0}\boldsymbol{\mu}^{t}+\bS^{t}) && \in 
        \R^{q\times q}\, . 
    \end{align}
    initialized with $\mathbf{S}^{0} = \mathbf{X}^{0}-\boldsymbol{\mu}_{0}\mathbf{V}_{0}$. Under assumptions $\ref{it:ass-sym-1}-\ref{it:ass-sym-6}$, the iterates $\mathbf{S}^{t}$ obey the state evolution equations
    Eq.\eqref{eq:spike_SE} owing to Theorem \ref{thm:symmetric}. We now prove the following statement by induction.
    \begin{align}
        \label{eq:spike_equiv}
        \forall t\in \mathbb{N} \quad \frac{1}{\sqrt{N}}\norm{\mathbf{X}^{t}-\mathbf{S}^{t}-\mathbf{V}_{0}\boldsymbol{\mu}^{t}}_{F} \xrightarrow[N \to \infty]{P} 0
    \end{align}
    The statement is true at $t=0$ owing to the initialization of the sequences. Assume the statement is true up to time $t$. We can then write
    \begin{align}
&\mathbf{X}^{t+1}-\mathbf{S}^{t+1}-\mathbf{V}_{0}\boldsymbol{\mu}^{t+1} = \hat{\mathbf{A}}\bM^{t}-\bM^{t-1}(\bb^{t})^{\top}-\mathbf{A}\tilde{\bM}^{t}+\tilde{\bM}^{t-1}(\tilde{\bb}^{t})^{\top}-\mathbf{V}_{0}\boldsymbol{\mu}^{t+1} \\
&=\mathbf{A}\left(f^{t}(\mathbf{X}^{t})-f^{t}\left(\mathbf{V}_{0}\boldsymbol{\mu}^{t}+\bS^{t}\right)\right)+\frac{1}{N}\mathbf{V}_{0}\mathbf{V}_{0}^{\top}f^{t}\left(\mathbf{X}^{t}\right)-\mathbf{V}_{0}\boldsymbol{\mu}^{t+1} \notag \\
&+\left(f^{t-1}(\mathbf{V}_{0}\boldsymbol{\mu}^{t-1}+\mathbf{S}^{t-1})-f^{t-1}(\mathbf{X}^{t-1}))\right)(\tilde{\mathbf{b}}^{t})^{\top}+f^{t-1}(\mathbf{X}^{t-1})(\tilde{\mathbf{b}}^{t}-\mathbf{b}^{t})^{\top}
    \end{align}
The triangle inequality then gives
\begin{align}
    \label{eq:inter_spike_tri}
    &\frac{1}{\sqrt{N}}\norm{\mathbf{X}^{t+1}-\mathbf{S}^{t+1}-\mathbf{V}_{0}\boldsymbol{\mu}^{t+1}}_{N} \leqslant \frac{1}{\sqrt{N}}\norm{\mathbf{A}}_{op}\norm{f^{t}(\mathbf{X}^{t})-f^{t}\left(\mathbf{V}_{0}\boldsymbol{\mu}^{t}+\bS^{t}\right)}_{F}\notag\\
    &+\frac{1}{\sqrt{N}}\norm{\frac{1}{N}\mathbf{V}_{0}\mathbf{V}_{0}^{\top}f^{t}\left(\mathbf{X}^{t}\right)-\mathbf{V}_{0}\boldsymbol{\mu}^{t+1}}_{F} \notag \\
    &+\frac{1}{\sqrt{N}}\norm{\left(f^{t-1}(\mathbf{V}_{0}\boldsymbol{\mu}^{t-1}+\mathbf{S}^{t-1})-f^{t-1}(\mathbf{X}^{t-1}))\right)(\tilde{\mathbf{b}}^{t})^{\top}}_{F}+\frac{1}{\sqrt{N}}\norm{f^{t-1}(\mathbf{X}^{t-1})(\tilde{\mathbf{b}}^{t}-\mathbf{b}^{t})^{\top}}_{F}
\end{align}
and, owing to the pseudo-Lipschitz property
\begin{align}
    \frac{1}{\sqrt{N}}\norm{f^{t}(\mathbf{X}^{t})-f^{t}\left(\mathbf{V}_{0}\boldsymbol{\mu}^{t}+\bS^{t}\right)}_{F} &\leqslant \notag \\
    &\hspace{-1cm}L\left(1+\left(\frac{\norm{\mathbf{X}^{t}}_{F}}{\sqrt{N}}\right)^{k-1}+\left(\frac{\norm{\mathbf{V}_{0}\boldsymbol{\mu}^{t}+\mathbf{S}^{t}}_{F}}{\sqrt{N}}\right)^{k-1}\right)\frac{\norm{\mathbf{X}^{t}-\mathbf{V}^{0}\boldsymbol{\mu^{t}}-\mathbf{S}^{t}}_{F}}{\sqrt{N}},
\end{align}
where the state evolution verified by iteration Eq.\eqref{eq:spike_aux_it} ensures that $\frac{\norm{\mathbf{V}_{0}\boldsymbol{\mu}^{t}+\mathbf{S}^{t}}_{F}}{\sqrt{N}}$ is bounded with high probability. The induction hypothesis then gives that  
$\frac{\norm{\mathbf{X}^{t}-\mathbf{V}^{0}\boldsymbol{\mu^{t}}-\mathbf{S}^{t}}_{F}}{\sqrt{N}} \xrightarrow[N \to \infty]{P} 0$, which, together with the previous statement ensures that $\frac{\norm{\mathbf{X}^{t}}_{F}}{\sqrt{N}}$ is also bounded with high probability. Combining this with proposition
\ref{op-norm-GOE} shows that 
\begin{equation}
    \frac{1}{\sqrt{N}}\norm{\mathbf{A}}_{op}\norm{f^{t}(\mathbf{X}^{t})-f^{t}\left(\mathbf{V}_{0}\boldsymbol{\mu}^{t}+\bS^{t}\right)}_{F} \xrightarrow[N \to \infty]{P} 0.
\end{equation} 
Then
\begin{align}
    \frac{1}{\sqrt{N}}\norm{\frac{1}{N}\mathbf{V}_{0}\mathbf{V}_{0}^{\top}f^{t}\left(\mathbf{V}_{0}\boldsymbol{\mu}^{t}+\mathbf{S}^{t}\right)-\mathbf{V}_{0}\boldsymbol{\mu}^{t+1}}_{F} \leqslant \frac{\norm{\mathbf{V}_{0}}_{F}}{\sqrt{N}}\norm{\frac{1}{N}\mathbf{V}_{0}^{\top}f^{t}\left(\bX^{t}\right)-\boldsymbol{\mu}^{t+1}}_{F}
\end{align}
where $\norm{\mathbf{V}_{0}}_{F}/\sqrt{N}$ is bounded with high probability by assumption. Since the function $\mathbf{V}_{0}^{\top}f^{t}(.)$ is pseudo-Lipschitz, we can use the induction hypothesis and SE equations together with the definition of $\boldsymbol{\mu}^{t}$ show that the r.h.s. goes to zero with high probability. The third term of the sum in the r.h.s. of Eq.\eqref{eq:inter_spike_tri} can be bounded in similar fashion
to the first one using the pseudo-Lipschitz property, the induction hypothesis and the boundedness of the norm of the Onsager term $\tilde{\mathbf{b}}^{t}$, which can be expressed as 
a pseudo-Lipschitz function of $\mathbf{S}^{t}$ using the SE property of iteration Eq.\eqref{eq:spike_aux_it} and Lemma \ref{matrix-stein}. The last term then verifies
\begin{align}
    \frac{1}{\sqrt{N}}\norm{f^{t-1}(\mathbf{X}^{t-1})(\tilde{\mathbf{b}}^{t}-\mathbf{b}^{t})^{\top}}_{F} \leqslant \frac{1}{\sqrt{N}}\norm{f^{t-1}(\mathbf{X}^{t-1})}_{F}\norm{\tilde{\mathbf{b}}^{t}-\mathbf{b}^{t}}_{F}
\end{align}
where $\frac{1}{\sqrt{N}}\norm{f^{t-1}(\mathbf{X}^{t-1})}_{F}$ is bounded w.h.p. owing to the induction hypothesis, pseudo-Lipschitz property of $f^{t-1}$ and the SE equations of iteration Eq.\eqref{eq:spike_aux_it}, and the difference in Onsager terms verifies
\begin{align}
    \norm{\tilde{\mathbf{b}}^{t}-\mathbf{b}^{t}}_{F} &= \frac{1}{N}\norm{\sum_{i=1}^{N}\left(\frac{\partial f^t_i}{\partial \bS_i}(\mathbf{V}_{0}\boldsymbol{\mu}^{t}+\bS^{t})-\frac{\partial f^t_i}{\partial \bX_i}(\bX^t)\right)}_{F} \notag \\
    &\leqslant \sup_{1\leqslant i \leqslant N} \norm{\frac{\partial f^t_i}{\partial \bS_i}(\mathbf{V}_{0}\boldsymbol{\mu}^{t}+\bS^{t})-\frac{\partial f^t_i}{\partial \bX_i}(\bX^t)}_{F}
\end{align}
where we remind that $f_{i}^{t}:\mathbb{R}^{N \times q} \to \mathbb{R}^{q}$ and is therefore a low-dimensional observable, for which the pseudo-Lipschitz assumption implies that there exists a constant $L$ such that
\begin{equation}
    \norm{\tilde{\mathbf{b}}^{t}-\mathbf{b}^{t}}_{F} \leqslant L\left(1+\left(\frac{\norm{\mathbf{X}^{t}}_{F}}{\sqrt{N}}\right)^{k-1}+\left(\frac{\norm{\mathbf{V}_{0}\boldsymbol{\mu}^{t}+\mathbf{S}^{t}}_{F}}{\sqrt{N}}\right)^{k-1}\right)\frac{\norm{\mathbf{X}^{t}-\mathbf{V}^{0}\boldsymbol{\mu^{t}}-\mathbf{S}^{t}}_{F}}{\sqrt{N}}
\end{equation}
which converges to zero with high probability for large N using the induction hypthesis and the SE equations of iteration $\eqref{eq:spike_aux_it}$. This concludes the induction
and proves the statement Eq.\eqref{eq:spike_equiv}. The proof of Lemma \ref{lemma:spike_SE} follows immediately from the pseudo-Lipschitz property, the property Eq.\eqref{eq:spike_equiv} and the SE equations of iteration Eq.\eqref{eq:spike_aux_it}.
\end{proof}
\subsection{Dependence on an additional linear observation}
\begin{lemma}
    \label{lemma:proj_SE}
    Let $\mathbf{W}_{0}\in \mathbb{R}^{N \times q}$ be a matrix such that $\frac{1}{N}\norm{\mathbf{W}_{0}^{\top}\mathbf{W}_{0}}_{F}$ converges to a finite constant as $N \to \infty$, and a given pseudo-Lipschitz function $\varphi : \mathbb{R}^{N \times q} \to \mathbb{R}^{N}$. Consider the AMP iteration initialized with $\mathbf{X}^{0} \in \mathbb{R}^{N \times q}$
\begin{align}
    \label{eq:proj_AMP_it}
    \bX^{t+1} &= \mathbf{A}\bM^{t}-\bM^{t-1}(\bb^{t})^{\top} && \in \R^{N\times q} \, ,  \\
    \bM^{t} &=f^{t}(\varphi\left(\mathbf{A}\mathbf{W}_{0}\right),\bX^{t}) && \in \R^{N\times q} \, , \\
    \bb^t &= \frac{1}{N} \sum_{i=1}^N \frac{\partial f^t_i}{\partial \bX_i}(\varphi\left(\mathbf{A}\mathbf{W}_{0}\right),\bX^t) && \in 
    \R^{q\times q}\, . 
\end{align}
where the functions $f^{t} : \mathbb{R}^{N \times (q+1)} \to \mathbb{R}^{N \times q}$ are pseudo-Lipschitz. Consider the following state evolution recursion, initialized with $\boldsymbol{\nu}^{0},\hat{\boldsymbol{\nu}}^{0} = 0_{q \times q}$,
    \begin{align}
        \label{eq:proj_SE}
        &\boldsymbol{\nu}^{0}, \hat{\boldsymbol{\nu}}^{0}, \boldsymbol{\kappa}^{1,1} = \frac{1}{N}f^{0}(\bX^{0})^{\top}f^{0}(\bX^{0}) \\
        &\boldsymbol{\nu}^{t+1} = \lim_{N \to \infty} \frac{1}{N}\mathbb{E}\left[\mathbf{W}_{0}^{\top}f^{t}\left(\varphi(\mathbf{Z}_{\mathbf{W}_{0}}), \mathbf{Z}_{\mathbf{W}_{0}}\rho_{\mathbf{W}_{0}}^{-1}\boldsymbol{\nu}^{t}+\mathbf{W}_{0}\hat{\boldsymbol{\nu}}^{t}+\mathbf{Z}^{t}\right)\right] \\
        &\hat{\boldsymbol{\nu}}^{t+1} = \lim_{N \to \infty} \frac{1}{N}\mathbb{E}\left[\sum_{i=1}^{N}\frac{\partial f_{i}^{t}}{\partial \mathbf{Z}_{\mathbf{W}_{0},i},\varphi}\left(\varphi(\mathbf{Z}_{\mathbf{W}_{0}}), \mathbf{Z}_{\mathbf{W}_{0}}\rho_{\mathbf{W}_{0}}^{-1}\boldsymbol{\nu}^{t}+\mathbf{W}_{0}\hat{\boldsymbol{\nu}}^{t}+\mathbf{Z}^{t}\right)\right] \\
        &\boldsymbol{\kappa}^{t+1, s+1} = \boldsymbol{\kappa}^{s+1, t+1} = \notag \\
        &\lim_{N \to \infty} \frac{1}{N}\mathbb{E}\bigg[\left(f^{s}\left(\varphi(\mathbf{Z}_{\mathbf{W}_{0}}), \mathbf{Z}_{\mathbf{W}_{0}}\rho_{\mathbf{W}_{0}}^{-1}\boldsymbol{\nu}^{s}+\mathbf{W}_{0}\hat{\boldsymbol{\nu}}^{s}+\mathbf{Z}^{s}\right)-\mathbf{W}_{0}\rho_{\mathbf{W}_{0}}^{-1}\boldsymbol{\nu}^{s+1}\right)^{\top} \notag \\
        &\hspace{6cm}\left(f^{t}\left(\varphi(\mathbf{Z}_{\mathbf{W}_{0}}), \mathbf{Z}_{\mathbf{W}_{0}}\rho_{\mathbf{W}_{0}}^{-1}\boldsymbol{\nu}^{t}+\mathbf{W}_{0}\hat{\boldsymbol{\nu}}^{t}+\mathbf{Z}^{t}\right)-\mathbf{W}_{0}\rho_{\mathbf{W}_{0}}^{-1}\boldsymbol{\nu}^{t+1}\right)\bigg]
    \end{align}
    where the notation $\partial \mathbf{Z}_{\mathbf{W}_{0,i},\varphi}$ denotes a derivatives w.r.t. the argument of $\varphi$, $\rho_{\mathbf{W}_{0}} = \frac{1}{N}\mathbf{W}_{0}^{\top}\mathbf{W}_{0}$, and $\mathbf{Z}_{\mathbf{W}_{0}} \sim \mathbf{N}(0, \rho_{\mathbf{W}_{0}}\otimes \mathbf{I}_{N})$ is independent from the $(\bZ^{1}, ..., \bZ^{t}) \sim \mathbf{N}\left(0,\left(\kappa^{s,r}\right)_{s,r\leqslant t} \otimes \mathbf{I}_{N}\right)$. Assume $\ref{it:ass-sym-1}-\ref{it:ass-sym-6}$ and that for any $t \in \mathbb{N}$,
    any $1 \leqslant i \leqslant N$, the derivative $\frac{\partial f^t_i}{\partial \bX_i}$ is pseudo-Lipschitz of order $k$.
    Then for any sequence $\phi_{N}:(\mathbb{R}^{N \times q})^{t+1} \to \mathbb{R}$ of pseudo-Lipschitz functions
    \begin{equation}
        \phi_{N}\left(\mathbf{X}^{0}, \mathbf{X}^{1}, ..., \mathbf{X}^{t}\right) \approxP \mathbb{E}\left[\phi_{N}\left(\mathbf{Z}_{\mathbf{W}_{0}}\rho_{\mathbf{W}_{0}}^{-1}\boldsymbol{\nu}^{0}+\mathbf{W}_{0}\hat{\boldsymbol{\nu}}^{0}+\mathbf{Z}^{0},...,\mathbf{Z}_{\mathbf{W}_{0}}\rho_{\mathbf{W}_{0}}^{-1}\boldsymbol{\nu}^{t}+\mathbf{W}_{0}\hat{\boldsymbol{\nu}}^{t}+\mathbf{Z}^{t}\right)\right]
    \end{equation}
\end{lemma}
\begin{proof}[Proof of lemma \ref{lemma:proj_SE}]
Consider the following iteration
\begin{align}
    \label{eq:proj_aux_it}
    \bS^{t+1} &= \tilde{\mathbf{A}}\tilde{\bM}^{t}-\tilde{\bM}^{t-1}(\tilde{\bb}^{t})^{\top} && \in \R^{N\times q} \, ,  \\
    \tilde{\bM}^{t} &=f^{t}\left(\varphi(\mathbf{A}\mathbf{W}_{0}), \mathbf{A}\mathbf{W}_{0}\rho_{\mathbf{W}_{0}}^{-1}\boldsymbol{\nu}^{t}+\mathbf{W}_{0}\hat{\boldsymbol{\nu}}^{t}+\mathbf{S}^{t}\right)-\mathbf{W}_{0}\rho_{\mathbf{W}_{0}}^{-1}\boldsymbol{\nu}^{t+1} && \in \R^{N\times q} \, , \\
    \tilde{\bb}^t &= \frac{1}{N} \sum_{i=1}^N \frac{\partial f^t_i}{\partial \bS_i}\left(\varphi(\mathbf{A}\mathbf{W}_{0}), \mathbf{A}\mathbf{W}_{0}\rho_{\mathbf{W}_{0}}^{-1}\boldsymbol{\nu}^{t}+\mathbf{W}_{0}\hat{\boldsymbol{\nu}}^{t}+\mathbf{S}^{t}\right) && \in 
    \R^{q\times q}\, . 
\end{align}
where $\tilde{\mathbf{A}}$ is a copy of $\mathbf{A}$ independent on $\mathbf{Z}_{\mathbf{W}_{0}}$. Under assumptions $\ref{it:ass-sym-1}-\ref{it:ass-sym-6}$ and conditionally on $\mathbf{A}\mathbf{W}_{0}$, the iterates $\mathbf{S}^{t}$ obey the state evolution equations
Eq.\eqref{eq:proj_SE} where the $\mathbf{Z}_{\mathbf{W}_{0}}$ are replaced by fixed $\mathbf{A}\mathbf{W}_{0}$, owing to Theorem \ref{thm:symmetric}. For any $t$, the composition of $f^{t}$ and $\varphi$ is pseudo-Lipschitz of order $k$, and owing to Lemma \ref{conv_lemmas_app}, $\frac{1}{\sqrt{N}}\norm{\mathbf{A}\mathbf{W}_{0}-\mathbf{Z}_{\mathbf{W}_{0}}}_{F}\xrightarrow[N \to +\infty]{P} 0$. Using the pseudo-Lipschitz property, the assumption on $\mathbf{W}_{0}$ to bound the norms of $\frac{1}{\sqrt{N}}\mathbf{A}\mathbf{W}_{0}$ and $\frac{1}{\sqrt{N}}\mathbf{W}_{0}$ w.h.p., and Lemma \ref{pseudo-lip-conv}, we obtain that iteration Eq.\eqref{eq:proj_aux_it} verifies the SE equations Eq.\eqref{eq:proj_SE}, where the expectations are taken w.r.t. $\mathbf{Z}_{\mathbf{W}_{0}}$ and all the $\mathbf{Z}^{s}$ for $0\leqslant s\leqslant t$. We now prove the following statement by induction
\begin{equation}
    \forall t \in \mathbb{N} \quad \frac{1}{\sqrt{N}}\norm{\mathbf{X}^{t}-\mathbf{A}\mathbf{W}_{0}\rho_{\mathbf{W}_{0}}^{-1}\boldsymbol{\nu}^{t}-\mathbf{W}_{0}\hat{\boldsymbol{\nu}}^{t}-\mathbf{S}^{t}}_{F} \xrightarrow[N \to \infty]{P} 0
\end{equation}
The property is true at $t=0$ owing to the initialization of both sequences. Assume the property is verified up to time $t$. Then, denoting the increment $\Delta^{t} = \mathbf{X}^{t}-\mathbf{A}\mathbf{W}_{0}\rho_{\mathbf{W}_{0}}^{-1}\boldsymbol{\nu}^{t}-\mathbf{W}_{0}\hat{\boldsymbol{\nu}}^{t}-\mathbf{S}^{t}$
\begin{align}
    \label{eq:increment_proj}
    \Delta^{t} = \mathbf{A}\bM^{t}-\bM^{t-1}(\bb^{t})^{\top}-\left(\tilde{\mathbf{A}}\tilde{\bM}^{t}-\tilde{\bM}^{t-1}(\tilde{\bb}^{t})^{\top}\right)-\mathbf{A}\mathbf{W}_{0}\rho_{\mathbf{W}_{0}}^{-1}\boldsymbol{\nu}^{t+1}-\mathbf{W}_{0}\hat{\boldsymbol{\nu}}^{t+1}
\end{align}
Consider then the iteration Eq.\eqref{eq:proj_AMP_it}, where we condition on the value of $\mathbf{A}\mathbf{W}_{0}$ at each iteration. A straightforward induction starting from the initialization then shows that, for any $t \in \mathbb{N}$
\begin{align}
    \mathbf{X}^{t+1}_{\vert \mathbf{A}\mathbf{W}_{0}} = \mathbf{A}_{\vert\mathbf{A}\mathbf{W}_{0}}f^{t}(\varphi\left(\mathbf{A}\mathbf{W}_{0}\right),\bX^{t}_{\vert\mathbf{A}\mathbf{W}_{0}})-f^{t-1}(\varphi\left(\mathbf{A}\mathbf{W}_{0}\right),\bX^{t-1}_{\vert \mathbf{A}\mathbf{W}_{0}})\left(\frac{1}{N} \sum_{i=1}^N \frac{\partial f^t_i}{\partial \bX_i}(\varphi\left(\mathbf{A}\mathbf{W}_{0}\right),\bX^t_{\vert\mathbf{A}\mathbf{W}_{0}})\right)^{\top}
\end{align}
Using the same 
lemma from \cite{bayati2011dynamics,javanmard2013state} used in the proof of Lemma \ref{lemma:cond_dist_LoAMP}, we may write
\begin{align}
    \mathbf{A}_{\vert \mathbf{A}\mathbf{W}_{0}} &= \mathbf{A}-\mathbf{P}_{\mathbf{W}_{0}}\mathbf{A}\mathbf{P}_{\mathbf{W}_{0}}+\mathbf{P}^{\perp}_{\mathbf{W}_{0}}\tilde{\mathbf{A}}\mathbf{P}^{\perp}_{\mathbf{W}_{0}} \\
    &= \mathbf{A}\mathbf{P}_{\mathbf{W}_{0}}+\mathbf{P}_{\mathbf{W}_{0}}\mathbf{A}-\mathbf{P}_{\mathbf{W}_{0}}\mathbf{A}\mathbf{P}_{\mathbf{W}_{0}}+\mathbf{P}^{\perp}_{\mathbf{W}_{0}}\tilde{\mathbf{A}}\mathbf{P}^{\perp}_{\mathbf{W}_{0}}
\end{align}
where $\tilde{\mathbf{A}}$ is an independent copy of $\mathbf{A}$ and $\mathbf{P}_{\mathbf{W}_{0}} = \mathbf{W}_{0}\left(\mathbf{W}_{0}^{\top}\mathbf{W}_{0}^{\top}\right)^{-1}\mathbf{W}_{0}^{\top} = \frac{1}{N}\mathbf{W}_{0}\rho_{\mathbf{W}_{0}}^{-1}\mathbf{W}_{0}^{\top} $ is always well-defined for $n \geqslant q$. We can then lift the conditioning by considering the 
distribution of $\mathbf{A}\mathbf{W}_{0}$ (which is straightforward since there is no correlation between $\mathbf{A}$ and $\mathbf{W}_{0}$) in all subsequent expressions. The increment Eq.\eqref{eq:increment_proj} becomes
\begin{align}
    &\left(\mathbf{A}\mathbf{P}_{\mathbf{W}_{0}}+\mathbf{P}_{\mathbf{W}_{0}}\mathbf{A}-\mathbf{P}_{\mathbf{W}_{0}}\mathbf{A}\mathbf{P}_{\mathbf{W}_{0}}+\mathbf{P}^{\perp}_{\mathbf{W}_{0}}\tilde{\mathbf{A}}\mathbf{P}^{\perp}_{\mathbf{W}_{0}}\right)\bM^{t}-\bM^{t-1}(\bb^{t})^{\top}-\left(\tilde{\mathbf{A}}\tilde{\bM}^{t}-\tilde{\bM}^{t-1}(\tilde{\bb}^{t})^{\top}\right) \notag \\
    &-\mathbf{A}\mathbf{W}_{0}\rho_{\mathbf{W}_{0}}^{-1}\boldsymbol{\nu}^{t+1}-\mathbf{W}_{0}\hat{\boldsymbol{\nu}}^{t+1}
\end{align}
where we chose the matrix $\tilde{\mathbf{A}}$ coming from the decomposition of $\mathbf{A}$ to define the iteration Eq.\eqref{eq:proj_aux_it}, and
\begin{align}
    &\Delta^{t} = \mathbf{A}\mathbf{P}_{\mathbf{W}_{0}}f^{t}\left(\varphi(\mathbf{A}\mathbf{W}_{0}),\mathbf{X}^{t}\right)+\mathbf{P}_{\mathbf{W}_{0}}\mathbf{A}f^{t}\left(\varphi(\mathbf{A}\mathbf{W}_{0}),\mathbf{X}^{t}\right)-\mathbf{P}_{\mathbf{W}_{0}}\mathbf{A}\mathbf{P}_{\mathbf{W}_{0}}f^{t}\left(\varphi(\mathbf{A}\mathbf{W}_{0}),\mathbf{X}^{t}\right) \notag \\
    &+\mathbf{P}^{\perp}_{\mathbf{W}_{0}}\tilde{\mathbf{A}}\mathbf{P}^{\perp}_{\mathbf{W}_{0}}f^{t}\left(\varphi(\mathbf{A}\mathbf{W}_{0}),\mathbf{X}^{t}\right)- f^{t-1}\left(\varphi(\mathbf{A}\mathbf{W}_{0}),\mathbf{X}^{t-1}\right)(\mathbf{b}^{t})^{\top}-\mathbf{A}\mathbf{W}_{0}\rho_{\mathbf{W}_{0}}^{-1}\boldsymbol{\nu}^{t+1}-\mathbf{W}_{0}\hat{\boldsymbol{\nu}}^{t+1} \notag \\
    &-\tilde{\mathbf{A}}\left(f^{t}\left(\varphi(\mathbf{A}\mathbf{W}_{0}), \mathbf{A}\mathbf{W}_{0}\rho_{\mathbf{W}_{0}}^{-1}\boldsymbol{\nu}^{t}+\mathbf{W}_{0}\hat{\boldsymbol{\nu}}^{t}+\mathbf{S}^{t}\right)-\mathbf{W}_{0}\rho_{\mathbf{W}_{0}}^{-1}\boldsymbol{\nu}^{t+1}\right) \notag \\
    &+\left(f^{t-1}\left(\varphi(\mathbf{A}\mathbf{W}_{0}), \mathbf{A}\mathbf{W}_{0}\rho_{\mathbf{W}_{0}}^{-1}\boldsymbol{\nu}^{t-1}+\mathbf{W}_{0}\hat{\boldsymbol{\nu}}^{t-1}+\mathbf{S}^{t-1}\right)-\mathbf{W}_{0}\rho_{\mathbf{W}_{0}}^{-1}\boldsymbol{\nu}^{t}\right)\left(\tilde{\mathbf{b}}^{t}\right)^{\top} \\
    &=\mathbf{A}\mathbf{P}_{\mathbf{W}_{0}}f^{t}\left(\varphi(\mathbf{A}\mathbf{W}_{0}),\mathbf{X}^{t}\right)-\mathbf{A}\mathbf{W}_{0}\rho_{\mathbf{W}_{0}}^{-1}\boldsymbol{\nu}^{t+1}+\mathbf{P}_{\mathbf{W}_{0}}\mathbf{A}f^{t}\left(\varphi(\mathbf{A}\mathbf{W}_{0}),\mathbf{X}^{t}\right)-\mathbf{W}_{0}\hat{\boldsymbol{\nu}}^{t+1}-\mathbf{W}_{0}\rho_{\mathbf{W}_{0}}^{-1}\boldsymbol{\nu}^{t}(\tilde{\mathbf{b}}^{t})^{\top} \notag \\
    &-f^{t-1}\left(\varphi(\mathbf{A}\mathbf{W}_{0}),\mathbf{X}^{t-1}\right)(\mathbf{b}^{t})^{\top}+f^{t-1}\left(\varphi(\mathbf{A}\mathbf{W}_{0}), \mathbf{A}\mathbf{W}_{0}\rho_{\mathbf{W}_{0}}^{-1}\boldsymbol{\nu}^{t-1}+\mathbf{W}_{0}\hat{\boldsymbol{\nu}}^{t-1}+\mathbf{S}^{t-1}\right)\left(\tilde{\mathbf{b}}^{t}\right)^{\top} \notag \\
    &-\tilde{\mathbf{A}}\left(f^{t}\left(\varphi(\mathbf{A}\mathbf{W}_{0}), \mathbf{A}\mathbf{W}_{0}\rho_{\mathbf{W}_{0}}^{-1}\boldsymbol{\nu}^{t}+\mathbf{W}_{0}\hat{\boldsymbol{\nu}}^{t}+\mathbf{S}^{t}\right)-\mathbf{W}_{0}\rho_{\mathbf{W}_{0}}^{-1}\boldsymbol{\nu}^{t+1}\right)+\mathbf{P}^{\perp}_{\mathbf{W}_{0}}\tilde{\mathbf{A}}\mathbf{P}^{\perp}_{\mathbf{W}_{0}}f^{t}\left(\varphi(\mathbf{A}\mathbf{W}_{0}),\mathbf{X}^{t}\right) \notag \\
    &-\mathbf{P}_{\mathbf{W}_{0}}\mathbf{A}\mathbf{P}_{\mathbf{W}_{0}}f^{t}\left(\varphi(\mathbf{A}\mathbf{W}_{0}),\mathbf{X}^{t}\right)
\end{align}
where, the second equality is only a reorganization of the terms. We now study the asymptotic behaviour of each component of the previous sum. We have
\begin{align}
    \label{eq:182}
    &\frac{1}{\sqrt{N}}\norm{\mathbf{A}\mathbf{P}_{\mathbf{W}_{0}}f^{t}\left(\varphi(\mathbf{A}\mathbf{W}_{0}),\mathbf{X}^{t}\right)-\mathbf{A}\mathbf{W}_{0}\rho_{\mathbf{W}_{0}}^{-1}\boldsymbol{\nu}^{t+1}}_{F} \leqslant \notag \\
    &\hspace{5cm}\norm{\mathbf{A}}_{op}\frac{1}{\sqrt{N}}\norm{\mathbf{W}_{0}\rho_{\mathbf{W}_{0}}^{-1}}_{F}\norm{\frac{1}{N}\mathbf{W}_{0}^{\top}f^{t}\left(\varphi(\mathbf{A}\mathbf{W}_{0}),\mathbf{X}^{t}\right)-\boldsymbol{\nu}^{t+1}}_{F}
\end{align}
where $\norm{\mathbf{A}}_{op}$ is bounded w.h.p. owing to lemma \ref{op-norm-GOE} and $\frac{1}{\sqrt{N}}\norm{\mathbf{W}_{0}\rho_{\mathbf{W}_{0}}^{-1}}_{F}$ is bounded w.h.p. by assumption. Then, using the pseudo-Lipschitz property, the induction hypothesis and Lemma \ref{pseudo-lip-conv}, it holds that 
\begin{align}
    \label{eq:inter_pseudo_lip}
    &\frac{1}{\sqrt{N}}\norm{f^{t}\left(\phi(\mathbf{A}\mathbf{W}_{0},\mathbf{X}^{t})\right)-f^{t}\left(\varphi(\mathbf{Z}_{\mathbf{W}_{0}}), \mathbf{Z}_{\mathbf{W}_{0}}\rho_{\mathbf{W}_{0}}^{-1}\boldsymbol{\nu}^{t}+\mathbf{W}_{0}\hat{\boldsymbol{\nu}}^{t}+\mathbf{S}^{t}\right)}_{F} \xrightarrow[N \to \infty]{P} 0
\end{align}
The triangle inequality then gives 
\begin{align}
    &\norm{\frac{1}{N}\mathbf{W}_{0}^{\top}f^{t}\left(\varphi(\mathbf{A}\mathbf{W}_{0}),\mathbf{X}^{t}\right)-\boldsymbol{\nu}^{t+1}}_{F} \leqslant \notag \norm{\frac{1}{N}\mathbf{W}_{0}^{\top}f^{t}\left(\varphi(\mathbf{Z}_{\mathbf{W}_{0}}), \mathbf{Z}_{\mathbf{W}_{0}}\rho_{\mathbf{W}_{0}}^{-1}\boldsymbol{\nu}^{t}+\mathbf{W}_{0}\hat{\boldsymbol{\nu}}^{t}+\mathbf{S}^{t}\right)-\boldsymbol{\nu}^{t+1}}_{F} \times \\
    &\frac{1}{\sqrt{N}}\norm{\mathbf{W}_{0}}_{F}\frac{1}{\sqrt{N}}\norm{f^{t}\left(\varphi(\mathbf{A}\mathbf{W}_{0},\mathbf{X}^{t})\right)-f^{t}\left(\varphi(\mathbf{Z}_{\mathbf{W}_{0}}), \mathbf{Z}_{\mathbf{W}_{0}}\rho_{\mathbf{W}_{0}}^{-1}\boldsymbol{\nu}^{t}+\mathbf{W}_{0}\hat{\boldsymbol{\nu}}^{t}+\mathbf{S}^{t}\right)}_{F}.
\end{align}
Using the definition of $\boldsymbol{\mu}^{t+1}$, the assumption on $\mathbf{W}_{0}$ and Eq.\eqref{eq:inter_pseudo_lip}, we conclude that, with high probability 
\begin{equation}
    \frac{1}{\sqrt{N}}\norm{\mathbf{A}\mathbf{P}_{\mathbf{W}_{0}}f^{t}\left(\varphi(\mathbf{A}\mathbf{W}_{0}),\mathbf{X}^{t}\right)-\mathbf{A}\mathbf{W}_{0}\rho_{\mathbf{W}_{0}}^{-1}\boldsymbol{\nu}^{t+1}}_{F} \xrightarrow[N \to \infty]{} 0
\end{equation}
The term
\begin{align}
    &\frac{1}{\sqrt{N}}\norm{f^{t-1}\left(\varphi(\mathbf{A}\mathbf{W}_{0}), \mathbf{A}\mathbf{W}_{0}\rho_{\mathbf{W}_{0}}^{-1}\boldsymbol{\nu}^{t-1}+\mathbf{W}_{0}\hat{\boldsymbol{\nu}}^{t-1}+\mathbf{S}^{t-1}\right)\left(\tilde{\mathbf{b}}^{t}\right)^{\top}-f^{t-1}\left(\varphi(\mathbf{A}\mathbf{W}_{0}),\mathbf{X}^{t-1}\right)(\mathbf{b}^{t})^{\top}}_{F} \notag \\
    &\leqslant \frac{1}{\sqrt{N}}\norm{\left(f^{t-1}\left(\varphi(\mathbf{A}\mathbf{W}_{0}), \mathbf{A}\mathbf{W}_{0}\rho_{\mathbf{W}_{0}}^{-1}\boldsymbol{\nu}^{t-1}+\mathbf{W}_{0}\hat{\boldsymbol{\nu}}^{t-1}+\mathbf{S}^{t-1}\right)-f^{t-1}\left(\varphi(\mathbf{A}\mathbf{W}_{0}),\mathbf{X}^{t-1}\right)\right)\left(\tilde{\mathbf{b}}^{t}\right)^{\top}}_{F} \notag \\
    &\hspace{1cm}+\frac{1}{\sqrt{N}}\norm{f^{t-1}\left(\varphi(\mathbf{A}\mathbf{W}_{0}),\mathbf{X}^{t-1}\right)\left(\tilde{\mathbf{b}}^{t}-\mathbf{b}^{t}\right)}_{F},
\end{align}
is similar to the third term of Eq.\eqref{eq:inter_spike_tri} in the proof of Lemma \ref{lemma:spike_SE} and converges to zero with high probability for large $N$ using similar arguments.
Then, letting
\begin{align}
    \Delta_{1}^{t} = \tilde{\mathbf{A}}\left(f^{t}\left(\varphi(\mathbf{A}\mathbf{W}_{0}), \mathbf{A}\mathbf{W}_{0}\rho_{\mathbf{W}_{0}}^{-1}\boldsymbol{\nu}^{t}+\mathbf{W}_{0}\hat{\boldsymbol{\nu}}^{t}+\mathbf{S}^{t}\right)-\mathbf{W}_{0}\rho_{\mathbf{W}_{0}}^{-1}\boldsymbol{\nu}^{t+1}\right)-\mathbf{P}^{\perp}_{\mathbf{W}_{0}}\tilde{\mathbf{A}}\mathbf{P}^{\perp}_{\mathbf{W}_{0}}f^{t}\left(\varphi(\mathbf{A}\mathbf{W}_{0}),\mathbf{X}^{t}\right),
\end{align}
the defintion of $\mathbf{P}^{\perp}_{\mathbf{W}_{0}} = \mathbf{I}-\mathbf{P}_{\mathbf{W}_{0}}$ and the triangle inequality yield
\begin{align}
    &\frac{1}{\sqrt{N}}\norm{\Delta_{1}^{t}}_{F} \leqslant \norm{\tilde{\mathbf{A}}}_{op}\frac{1}{\sqrt{N}}\norm{\mathbf{P}_{\mathbf{W}_{0}}f^{t}\left(\varphi(\mathbf{A}\mathbf{W}_{0}),\mathbf{X}^{t}\right)-\mathbf{W}_{0}\rho_{\mathbf{W}_{0}}^{-1}\boldsymbol{\nu}^{t+1}}_{F} \notag \\
    &+\norm{\tilde{\mathbf{A}}}_{op}\frac{1}{\sqrt{N}}\norm{f^{t}\left(\varphi(\mathbf{A}\mathbf{W}_{0}), \mathbf{A}\mathbf{W}_{0}\rho_{\mathbf{W}_{0}}^{-1}\boldsymbol{\nu}^{t}+\mathbf{W}_{0}\hat{\boldsymbol{\nu}}^{t}+\mathbf{S}^{t}\right)-f^{t}\left(\varphi(\mathbf{A}\mathbf{W}_{0}),\mathbf{X}^{t}\right)}_{F} \notag \\
    &+\frac{1}{\sqrt{N}}\norm{\mathbf{P}_{\mathbf{W}_{0}}\tilde{\mathbf{A}}\mathbf{P}^{\perp}_{\mathbf{W}_{0}}f^{t}\left(\varphi(\mathbf{A}\mathbf{W}_{0}),\mathbf{X}^{t}\right)}_{F}
\end{align}
where the first term converges to zero w.h.p. using the same argument as the one used for Eq.\eqref{eq:182}. For the second term, the operator norm of $\tilde{\mathbf{A}}$ is bounded w.h.p. using Lemma \ref{op-norm-GOE}, and the 
diffence goes to zero w.h.p. using the pseudo-Lipschitz property, the induction hypothesis and the SE equations Eq.\eqref{eq:proj_SE} of iteration Eq.\eqref{eq:proj_aux_it}. Finally, since $\mathbf{P}_{\mathbf{W}_{0}}$ has finite rank and $\frac{1}{\sqrt{N}}\norm{\mathbf{P}^{\perp}_{\mathbf{W}_{0}}f^{t}\left(\varphi(\mathbf{A}\mathbf{W}_{0}),\mathbf{X}^{t}\right)}_{F}$ is bounded w.h.p. using the induction hypothesis and SE equations of iteration Eq.\eqref{eq:proj_aux_it},the last term goes to zero w.h.p. using Lemma \ref{conv_lemmas_app}. \\
Moving to the term $\mathbf{P}_{\mathbf{W}_{0}}\mathbf{A}f^{t}\left(\varphi(\mathbf{A}\mathbf{W}_{0}),\mathbf{X}^{t}\right)-\mathbf{W}_{0}\hat{\boldsymbol{\nu}}^{t+1}-\mathbf{W}_{0}\rho_{\mathbf{W}_{0}}^{-1}\boldsymbol{\nu}^{t}(\tilde{\mathbf{b}}^{t})^{\top}$, which we denote $\Delta_{2}^{t}$, we may write
\begin{align}
    \label{eq:189}
    \mathbf{P}_{\mathbf{W}_{0}}\mathbf{A}f^{t}\left(\varphi(\mathbf{A}\mathbf{W}_{0}),\mathbf{X}^{t}\right) = \frac{1}{N}\mathbf{W}_{0}\rho_{\mathbf{W}_{0}}^{-1}(\mathbf{A}\mathbf{W}_{0})^{\top}f^{t}\left(\varphi\left(\mathbf{A}\mathbf{W}_{0}\right),\mathbf{X}^{t}\right)
\end{align}
since the function $\mathbf{A}\mathbf{W}_{0},\mathbf{X}^{t} \to (\mathbf{A}\mathbf{W}_{0})^{\top}f^{t}\left(\varphi\left(\mathbf{A}\mathbf{W}_{0}\right),\mathbf{X}^{t}\right)$ is pseudo-Lipschitz, Lemma \ref{conv_lemmas_app} and the induction hypothesis give 
\begin{equation}
    \label{eq:190}
    \norm{\frac{1}{N}\left(\mathbf{A}\mathbf{W}_{0}\right)^{\top}f^{t}\left(\varphi\left(\mathbf{A}\mathbf{W}_{0}\right),\mathbf{X}^{t}\right)-\frac{1}{N}\mathbf{Z}_{\mathbf{W}_{0}}^{\top}f^{t}\left(\varphi\left(\mathbf{Z}_{\mathbf{W}_{0}}\right),\mathbf{Z}_{\mathbf{W}_{0}}\rho_{\mathbf{W}_{0}}^{-1}\boldsymbol{\nu}^{t}+\mathbf{W}_{0}\hat{\boldsymbol{\nu}}^{t}+\mathbf{S}^{t}\right)}_{F} \xrightarrow[N \to \infty]{P} 0,
\end{equation}
where the SE equations for iteration Eq.\eqref{eq:proj_aux_it} yield
\begin{align}
   &\frac{1}{N}\mathbf{Z}_{\mathbf{W}_{0}}^{\top}f^{t}\left(\varphi\left(\mathbf{Z}_{\mathbf{W}_{0}}\right),\mathbf{Z}_{\mathbf{W}_{0}}\rho_{\mathbf{W}_{0}}^{-1}\boldsymbol{\nu}^{t}+\mathbf{W}_{0}\hat{\boldsymbol{\nu}}^{t}+\mathbf{S}^{t}\right) \approxP \notag \\ 
   &\hspace{4cm}\frac{1}{N}\mathbb{E}\left[\mathbf{Z}_{\mathbf{W}_{0}}^{\top}f^{t}\left(\varphi\left(\mathbf{Z}_{\mathbf{W}_{0}}\right),\mathbf{Z}_{\mathbf{W}_{0}}\rho_{\mathbf{W}_{0}}^{-1}\boldsymbol{\nu}^{t}+\mathbf{W}_{0}\hat{\boldsymbol{\nu}}^{t}+\mathbf{Z}^{t}\right)\right]
\end{align}
An application of Lemma \ref{matrix-stein} and the chain rule gives
\begin{align}
    &\frac{1}{N}\mathbb{E}\left[\mathbf{Z}_{\mathbf{W}_{0}}^{\top}f^{t}\left(\varphi\left(\mathbf{Z}_{\mathbf{W}_{0}}\right),\mathbf{Z}_{\mathbf{W}_{0}}\rho_{\mathbf{W}_{0}}^{-1}\boldsymbol{\nu}^{t}+\mathbf{W}_{0}\hat{\boldsymbol{\nu}}^{t}+\mathbf{Z}^{t}\right)\right] = \notag \\
    &\hspace{1.2cm}\frac{1}{N}\rho_{\mathbf{W}_{0}}\mathbb{E}\left[\sum_{i=1}^{N}\frac{\partial f_{i}^{t}}{\partial \mathbf{Z}_{\mathbf{W}_{0},i},\varphi}\left(\varphi(\mathbf{Z}_{\mathbf{W}_{0}}), \mathbf{Z}_{\mathbf{W}_{0}}\rho_{\mathbf{W}_{0}}^{-1}\boldsymbol{\nu}^{t}+\mathbf{W}_{0}\hat{\boldsymbol{\nu}}^{t}+\mathbf{Z}^{t}\right)\right] \notag \\
    &\hspace{2.5cm}+\frac{1}{N}\mathbf{m}^{t}\mathbb{E}\left[\sum_{i=1}^N \frac{\partial f^t_i}{\partial \bZ_i}\left(\varphi(\mathbf{Z}_{\mathbf{W}_{0}}), \mathbf{Z}_{\mathbf{W}_{0}}\rho_{\mathbf{W}_{0}}^{-1}\boldsymbol{\nu}^{t}+\mathbf{W}_{0}\hat{\boldsymbol{\nu}}^{t}+\mathbf{Z}^{t}\right)\right].
\end{align}
The SE equations of iteration Eq.\eqref{eq:proj_aux_it} and the pseudo-Lipschitz assumptions on the Jacobians of the $f^{t}$ then show that
\begin{equation}
    \tilde{\mathbf{b}}^{\top} \approxP \frac{1}{N}\mathbb{E}\left[\sum_{i=1}^N \frac{\partial f^t_i}{\partial \bZ_i}\left(\varphi(\mathbf{Z}_{\mathbf{W}_{0}}), \mathbf{Z}_{\mathbf{W}_{0}}\rho_{\mathbf{W}_{0}}^{-1}\boldsymbol{\nu}^{t}+\mathbf{W}_{0}\hat{\boldsymbol{\nu}}^{t}+\mathbf{Z}^{t}\right)\right],
\end{equation}
which, combined with the definition of $\hat{\boldsymbol{\nu}}^{t}$, shows that
\begin{equation}
    \frac{1}{N}\mathbb{E}\left[\mathbf{Z}_{\mathbf{W}_{0}}^{\top}f^{t}\left(\varphi\left(\mathbf{Z}_{\mathbf{W}_{0}}\right),\mathbf{Z}_{\mathbf{W}_{0}}\rho_{\mathbf{W}_{0}}^{-1}\boldsymbol{\nu}^{t}+\mathbf{W}_{0}\hat{\boldsymbol{\nu}}^{t}+\mathbf{Z}^{t}\right)\right] \approxP \rho_{\mathbf{W}_{0}}\hat{\boldsymbol{\nu}}^{t+1}+\boldsymbol{\nu}^{t}\left(\tilde{\mathbf{b}}^{t}\right)^{\top}
\end{equation}
combining this with Eq.\eqref{eq:189} and Eq.\eqref{eq:190}, a straightforward application of the triangle inequality allows to show that
\begin{equation}
    \frac{1}{\sqrt{N}}\norm{\Delta_{2}^{t}}_{F} \xrightarrow[N \to \infty]{P} 0.
\end{equation}
The only remaining term in $\Delta^{t}$ is $\mathbf{P}_{\mathbf{W}_{0}}\mathbf{A}\mathbf{P}_{\mathbf{W}_{0}}f^{t}\left(\varphi(\mathbf{A}\mathbf{W}_{0}),\mathbf{X}^{t}\right)$. Expanding the projectors, we obtain
\begin{align}
    &\frac{1}{\sqrt{N}}\norm{\mathbf{P}_{\mathbf{W}_{0}}\mathbf{A}\mathbf{P}_{\mathbf{W}_{0}}f^{t}\left(\varphi(\mathbf{A}\mathbf{W}_{0}),\mathbf{X}^{t}\right)}_{F} = \frac{1}{\sqrt{N}}\norm{\frac{1}{N}\mathbf{W}_{0}\rho_{\mathbf{W}_{0}}^{-1}\mathbf{W}_{0}^{\top}\mathbf{A}\frac{1}{N}\mathbf{W}_{0}\rho_{\mathbf{W}_{0}}^{-1}\mathbf{W}_{0}^{\top}f^{t}\left(\varphi(\mathbf{A}\mathbf{W}_{0}),\mathbf{X}^{t}\right)}_{F} \notag \\
    &\leqslant \frac{1}{\sqrt{N}}\norm{\mathbf{W}_{0}\rho_{\mathbf{W}_{0}}^{-1}}_{F}\frac{1}{N}\norm{\mathbf{W}_{0}^{\top}\mathbf{A}\mathbf{W}_{0}}_{F}\frac{1}{N}\norm{\mathbf{W}_{0}^{\top}f^{t}\left(\varphi(\mathbf{A}\mathbf{W}_{0}),\mathbf{X}^{t}\right)}_{F}\norm{\rho_{\mathbf{W}_{0}}^{-1}}_{F}.
\end{align}
Lemma \ref{conv_lemmas_app} then shows that $\frac{1}{N}\norm{\mathbf{W}_{0}^{\top}\mathbf{A}\mathbf{W}_{0}}_{F} \xrightarrow[N \to \infty]{P} 0$, and the other terms are bounded w.h.p. We have now treated all 
the terms in $\Delta^{t}$, and the triangle inequality gives
\begin{align}
    \frac{1}{\sqrt{N}}\norm{\Delta^{t}}_{F} \xrightarrow[N \to \infty]{P} 0.
\end{align}
which concludes the induction. Combining this with the pseudo-Lipschitz property and the SE equations to ensure all iterates have bounded scaled norms, we conclude the proof of Lemma \ref{lemma:proj_SE}.
\end{proof}
\paragraph{Application to graph-based AMP iterations : proof of Lemma \ref{lemma:teacher_SE}}
Consider the AMP iteration \eqref{eq:AMP_teach1}-\eqref{eq:AMP_teach2}.To obtain the SE equations for this iteration, 
we follow a similar argument as the proof of Theorem \ref{thm:graph-AMP} and embed the iteration indexed on the graph $G = (V,E)$ into a large, symmetric iteration of the form of that of Lemma \ref{lemma:spike_SE} and Lemma \ref{lemma:proj_SE}. We may then write the $N \times N$ GOE matrix 
corresponding to the symmetric AMP iteration
\begin{align*}
    \hat{\bA} &= 
    \begin{pmatrix}
    \hat{\bA}_{\overrightarrow{e}_1} & & & & & & & \\
    & \ddots& & & & & \ast & \\
    & & \hat{\bA}_{\overrightarrow{e}_l} & & & & & \\
    & & & \ast & \hat{\bA}_{\overrightarrow{e}_{l+1}} & & & \\
    & & & \hat{\bA}_{\overleftarrow{e}_{l+1}} & \ast & & & \\
    & & & & & \ddots & & \\
    & \ast& & & &  & \ast & \hat{\bA}_{\overrightarrow{e}_{m}} \\
    & & & & & & \hat{\bA}_{\overleftarrow{e}_{m}} & \ast 
    \end{pmatrix} \, \\
\end{align*}
where, using the definition of each $\hat{\bA}_{\overrightarrow{e}}$, we may write 
\begin{align*}
    &=\begin{pmatrix}
        \bA_{\overrightarrow{e}_1} & & & & & & & \\
        & \ddots& & & & & \ast & \\
        & & \bA_{\overrightarrow{e}_l} & & & & & \\
        & & & \ast & \bA_{\overrightarrow{e}_{l+1}} & & & \\
        & & & \bA_{\overleftarrow{e}_{l+1}} & \ast & & & \\
        & & & & & \ddots & & \\
        & \ast& & & &  & \ast & \bA_{\overrightarrow{e}_{m}} \\
        & & & & & & \bA_{\overleftarrow{e}_{m}} & \ast \\
        \end{pmatrix} \, \\
        &\hspace{8cm}+\begin{pmatrix}
            \frac{1}{N}\mathbf{v}_{\overrightarrow{e}_1}\mathbf{v}_{\overrightarrow{e}_1}^{\top} & & & & & & & \\
            & \ddots& & & & & 0 & \\
            & & \frac{1}{N}\mathbf{v}_{\overrightarrow{e}_l}\mathbf{v}_{\overrightarrow{e}_l}^{\top} & & & & & \\
            & & & 0 & 0 & & & \\
            & & & 0 & 0 & & & \\
            & & & & & \ddots & & \\
            & 0& & & &  & 0 & 0 \\
            & & & & & & 0 & 0
            \end{pmatrix} \,
\end{align*}
where the second term gives the form of the matrix $\mathbf{V}_{0}$ from Lemma \ref{lemma:spike_SE}, i.e.
\begin{equation}
    \mathbf{V}_{0} = \begin{pmatrix}
        \mathbf{v}_{\overrightarrow{e}_1}& & & & & & & \\
        & \ddots& & & & & 0 & \\
        & & \mathbf{v}_{\overrightarrow{e}_l} & & & & & \\
        & & & 0 & 0 & & & \\
        & & & 0 & 0 & & & \\
        & & & & & \ddots & & \\
        & 0& & & &  & 0 & 0 \\
        & & & & & & 0 & 0
        \end{pmatrix}
\end{equation}. \\
Furthermore, we may write the update function of the symmetric AMP iteration as 
\begin{align}
    &\tilde{f}^t
    \begin{pmatrix}
    \bx_{\overrightarrow{e}_1} & & & & & & & \\
    & \ddots & & & & &\ast  & \\
    & & \bx_{\overrightarrow{e}_l} & & & & & \\
    & & & \bx_{\overrightarrow{e}_{l+1}} & & & & \\
    & & & &\bx_{\overleftarrow{e}_{l+1}} & & & \\
    & & & & & \ddots & & \\
     & \ast & & & & & \bx_{\overrightarrow{e}_{m}}  & \\
     & & & & & & &\bx_{\overleftarrow{e}_{m}}
    \end{pmatrix} \label{eq:def-non-linearity}\\
    &\qquad= 
    \begin{pmatrix}
    \tilde{f}^t_{\overrightarrow{e}_1}\left(\left(\bx_{\overrightarrow{e}}\right)_{\overrightarrow{e}:\overrightarrow{e} \rightarrow \overrightarrow{e}_1}\right) & & & & & & & \\
    & \ddots & & & & & 0 & \\
    & & \tilde{f}^t_{\overrightarrow{e}_l}\left(\dots\right) & & & & &\\
    & & & 0 & \tilde{f}^t_{\overleftarrow{e}_{l+1}}(\dots) & & & \\
    & & & \tilde{f}^t_{\overrightarrow{e}_{l+1}}(\dots) & 0 & & & \\
    & & & & & \ddots & & \\
     & 0 & & & & &  0 & \tilde{f}^t_{\overleftarrow{e}_{m}}(\dots)\\
     & & & & & & \tilde{f}^t_{\overrightarrow{e}_{m}}(\dots) & 0
    \end{pmatrix} \nonumber \\
    &\qquad= f^t\left(\Phi\left(\mathbf{A}\mathbf{W}_{0}\right),
    \begin{pmatrix}
    \bx_{\overrightarrow{e}_1} & & & & & & & \\
    & \ddots & & & & &\ast  & \\
    & & \bx_{\overrightarrow{e}_l} & & & & & \\
    & & & \bx_{\overrightarrow{e}_{l+1}} & & & & \\
    & & & &\bx_{\overleftarrow{e}_{l+1}} & & & \\
    & & & & & \ddots & & \\
     & \ast & & & & & \bx_{\overrightarrow{e}_{m}}  & \\
     & & & & & & &\bx_{\overleftarrow{e}_{m}}
    \end{pmatrix}\right)
\end{align}
where 
\begin{equation}
    \mathbf{W}_{0} = 
        \begin{pmatrix}
        0 & & & & & & & \\
        & \ddots& & & & & 0 & \\
        & & 0 & & & & & \\
        & & & 0 & \mathbf{w}_{\overleftarrow{e}_{l+1}} & & & \\
        & & & \mathbf{w}_{\overrightarrow{e}_{l+1}} & 0 & & & \\
        & & & & & \ddots & & \\
        & 0 & & & &  & 0 & \mathbf{w}_{\overleftarrow{e}_{m}} \\
        & & & & & & \mathbf{w}_{\overrightarrow{e}_{m}} & 0 \\
        \end{pmatrix}
\end{equation}
and the function $\Phi$ contains the functions $\varphi_{\overrightarrow{e}}$
\begin{align}
    &\Phi\begin{pmatrix}
        0 & & & & & & & \\
        & \ddots & & & & &0  & \\
        & & 0 & & & & & \\
        & & & \bA_{\overrightarrow{e}_{l+1}}\mathbf{w}_{\overrightarrow{e}_{l+1}} & & & & \\
        & & & &\bA_{\overleftarrow{e}_{l+1}}\mathbf{w}_{\overleftarrow{e}_{l+1}} & & & \\
        & & & & & \ddots & & \\
         & 0 & & & & & \bA_{\overrightarrow{e}_{m}}\mathbf{w}_{\overrightarrow{e}_{m}}  & \\
         & & & & & & &\bA_{\overleftarrow{e}_{m}}\mathbf{w}_{\overleftarrow{e}_{m}} 
        \end{pmatrix}= \notag \\
        &\begin{pmatrix}
            0 & & & & & & & \\
            & \ddots & & & & &0  & \\
            & & 0 & & & & & \\
            & & & \varphi_{\overrightarrow{e}_{l+1}}\left(\bA_{\overrightarrow{e}_{l+1}}\mathbf{w}_{\overrightarrow{e}_{l+1}}\right) & & & & \\
            & & & &\varphi_{\overleftarrow{e}_{l+1}}\left(\bA_{\overleftarrow{e}_{l+1}}\mathbf{w}_{\overleftarrow{e}_{l+1}}\right) & & & \\
            & & & & & \ddots & & \\
             & 0 & & & & & \varphi_{\overrightarrow{e}_{m}}\left(\bA_{\overrightarrow{e}_{m}}\mathbf{w}_{\overrightarrow{e}_{m}}\right)  & \\
             & & & & & & & \varphi_{\overleftarrow{e}_{m}}\left(\bA_{\overleftarrow{e}_{m}}\mathbf{w}_{\overleftarrow{e}_{m}}\right)
        \end{pmatrix}
\end{align}
Under the condition that the matrices $\mathbf{V}_{0}, \mathbf{W}_{0}$ and the function $\Phi$ verify the assumptions of Lemma \ref{lemma:spike_SE} and Lemma \ref{lemma:proj_SE}, 
we may use those results to obtain the SE equations for the iteration Eq.\eqref{eq:AMP_teach1}-\eqref{eq:AMP_teach2}. Evaluating the matrix products defining the parameters $\boldsymbol{\mu}^{t}, \boldsymbol{\nu}^{t}, \hat{\boldsymbol{\nu}}^{t}$ then leads to the SE equations of 
Lemma \ref{lemma:teacher_SE}.
\section{Useful definitions and probability lemmas}
\label{app:sec_conc}
In this section, we compile useful definitions and lemmas that appear throughout the proof. Most of those results are finite-width matrix generalizations of those appearing in \cite{berthier2020state} and some are the same.
\begin{proposition}(Norm of matrices with Gaussian entries \cite{vershynin2018high})
\label{prop:norm_Gauss}
Let $\bY$ be an $M\times N$ random matrix with independent $\mathbf{N}(0,1)$ entries. Then, for any $t>0$, we have:
\begin{equation}
\mathbb{P}\left(\norm{\bY}_{F}\leqslant C\left(\sqrt{M}+\sqrt{N}+t\right)\right) \geqslant 1-2\exp(-t^{2})
\end{equation}
where $C$ is an absolute constant.
\end{proposition}
\begin{proposition}(Operator norm of GOE(N) \cite{boucheron2013concentration}) \\
\label{op-norm-GOE}
Consider a sequence of matrices $\mathbf{A} \sim$ GOE(N). Then $\norm{\mathbf{A}}_{op} \to 2$ almost surely as $N \to \infty$.
\end{proposition}
%\begin{proposition}(Gaussian concentration of Lipschitz function \cite{boucheron2013concentration})
%Let $\mathbf{\bZ} \in \mathbb{R}^{N}$ be a $\mathbf{N}(0, \mathbf{I}_{N})$ random vector. Then for any continuous, L-Lipschitz function $\Phi$, there exists an absolute constant c such that:
%\begin{equation}
%    P\left(\abs{\Phi(\mathbf{\bZ})-\mathbb{E}[\Phi(\mathbf{\bZ})]} \geqslant t\right) \leqslant %2\exp^{-c\frac{t^{2}}{L^{2}}}
%\end{equation}
%\end{proposition}
\begin{proposition}(Gaussian Poincaré inequality \cite{boucheron2013concentration}) \\
\label{prop:gauss_poinc}
Let $\mathbf{\bZ} \in \mathbb{R}^{N}$ be a $\mathbf{N}(0, \mathbf{I}_{N})$ random vector. Then for any continuous, weakly differentiable $\varphi$, there exists a constant $c\geqslant 0$ such that:
\begin{equation}
    \mbox{Var}[\varphi(\mathbf{\bZ})] \leqslant c\mathbb{E}\left[\norm{\nabla \varphi(\bZ)}_{2}^{2}\right]
\end{equation}
\end{proposition}
The next result is a matrix version of Gaussian integration by parts, or Stein's lemma.
\begin{lemma}(Stein's lemma, matrix version)
\label{matrix-stein}
Let $(\mathbf{\bZ}_{1},\mathbf{\bZ}_{2}) \in \left(\mathbb{R}^{N \times q}\right)^{2}$ be two $\mathbf{N}(0, \boldsymbol{\kappa} \otimes \mathbf{I}_{N})$ random vectors, where $\boldsymbol{\kappa}\in \mathbb{R}^{(2q) \times (2q)}$.
\begin{equation}
    \boldsymbol{\kappa} = \begin{bmatrix}
    \boldsymbol{\kappa}_{11} \thickspace \boldsymbol{\kappa}_{12} \\
    \boldsymbol{\kappa}_{12} \thickspace  \boldsymbol{\kappa}_{22}\end{bmatrix}
\end{equation}
Consider an almost everywhere differentiable function $f:\mathbb{R}^{N \times q} \to \mathbb{R}^{N \times q}$. For any $\mathbf{\bZ} \in \mathbb{R}^{N \times q}$ we can write:
\begin{equation}
    f\left(\begin{bmatrix}
    \bZ_{11}, ..., \bZ_{1q} \\
    ...
    \\
    \bZ_{n1}, ...,\bZ_{nq}
    \end{bmatrix}\right) = \begin{bmatrix}
    f_1(\mathbf{\bZ}) \\
    ... \\
    f_{n}(\mathbf{\bZ})\end{bmatrix} = \begin{bmatrix}
    f_{1}^{1}(\mathbf{\bZ}), ...f_{1}^{q}(\mathbf{\bZ})\\
    ...\\
    f_{n}^{1}(\mathbf{\bZ}), ..., f_{n}^{q}(\mathbf{\bZ})
    \end{bmatrix}
\end{equation}
Then 
\begin{equation}
    \mathbb{E}\left[(\mathbf{\bZ}_{1})^{\top}f(\mathbf{\bZ}_{2})\right] = \boldsymbol{\kappa}_{1,2}\left(\sum_{k=1}^{N}\mathbb{E}\left[\frac{\partial f_{k}(\mathbf{\bZ}_{2})}{\partial \bZ_{k}}\right]\right)^{\top}
\end{equation}
where $\frac{\partial f_{k}(\mathbf{\bZ}_{2})}{\partial \bZ_{k}} \in \mathbb{R}^{q \times q}$ is the Jacobian containing the partial derivatives of $f_{k}$ w.r.t. the line $\mathbf{\bZ}_{k} \in \mathbb{R}^{q}$.
\end{lemma}
\begin{proof}
\begin{align}
\mathbb{E}\left[(\mathbf{\bZ}_{1})^{\top}f(\mathbf{\bZ}_{2})\right]_{ij}&= \sum_{k=1}^{N}\mathbb{E}\left[((\bZ_{1})_{ki}f_{kj}(\mathbf{\bZ}_{2})\right] \notag\\
&=\sum_{k=1}^{N}\sum_{l=1}^{q}\mathbb{E}[\bZ^{1}_{ki}\bZ^{2}_{kl}]\mathbb{E}\left[\frac{\partial f_{kj}}{\partial (\bZ_{2})_{kl}}(\mathbf{\bZ}_{2})\right] \quad \mbox{since} \thickspace (\mathbf{\bZ}_{1}, \mathbf{\bZ}_{2}) \sim \mathbf{N}(0, \boldsymbol{\kappa} \otimes I_{N}) \notag\\
&= \sum_{l=1}^{q}(\boldsymbol{\kappa}_{12})_{il}\sum_{k=1}^{N}\mathbb{E}\left[\frac{\partial f_{kj}}{\partial (\bZ_{2})_{kl}}(\mathbf{\bZ}_{2})\right] \notag\\
&=\sum_{l=1}^{q}(\boldsymbol{\kappa}_{12})_{il}\left(\sum_{k=1}^{N}\mathbb{E}\left[\frac{\partial f_{k}(\mathbf{\bZ}_{2})}{\partial \bZ_{k}}\right]\right)_{jl} \notag\\
&= \left(\boldsymbol{\kappa}_{12}\left(\sum_{k=1}^{N}\mathbb{E}\left[\frac{\partial f_{k}(\mathbf{\bZ}_{2})}{\partial \bZ_{k}}\right]\right)^{\top}\right)_{ij}
\end{align}
where the second step is obtained by iteratively conditioning on the entries of $\bZ_{2}$ and applying one dimensional Gaussiaan integration by parts, see e.g. \cite{vershynin2018high} Lemma 7.2.5.
\end{proof}
\begin{definition}[pseudo-Lipschitz function]
\label{def:pseudo-lip}
For $k \in \mathbb{N}^{*}$ and any $N,m \in \mathbb{N}^{*}$, a function $\Phi : \mathbb{R}^{N \times q} \to \mathbb{R}^{m \times q}$ is said to be \emph{pseudo-Lipschitz of order k} if there exists a constant L such that for any $\mathbf{x},\mathbf{y} \in \mathbb{R}^{N \times q}$, 
\begin{equation}
    \frac{\norm{\Phi(\mathbf{x})-\Phi(\mathbf{y})}_{F}}{\sqrt{m}} \leqslant L \left(1+\left(\frac{\norm{\mathbf{x}}_{F}}{\sqrt{N}}\right)^{k-1}+\left(\frac{\norm{\mathbf{y}}_{F}}{\sqrt{N}}\right)^{k-1}\right)\frac{\norm{\mathbf{x}-\mathbf{y}}_{F}}{\sqrt{N}}
\end{equation}

A family of pseudo-Lipschitz functions is said to be \emph{uniformly} pseudo-Lipschitz if all functions of the family are pseudo-Lipschitz with the same order $k$ and the same constant $L$. 
\end{definition}
We now remind useful properties of pseudo-Lipschitz functions from \cite{berthier2020state}.
\begin{lemma}
\label{lemma:pseudo-lip_product}
Let k be any positive integer. Consider two sequences $f : \mathbb{R}^{N} \to \mathbb{R}^{N}, N \geqslant 1$ and $g:\mathbb{R}^{N} \to \mathbb{R}^{N}, N \geqslant 1$ of uniformly pseudo-Lipschitz functions of order k. The sequence of functions $\Phi_{N}: \mathbb{R}^{N} \times \mathbb{R}^{N} \to \mathbb{R},N\geqslant 1$ such that $\Phi_{N}(\mathbf{x},\mathbf{y}) = \langle f(\mathbf{x}), g(\mathbf{y})\rangle$ is uniformly pseudo-Lipschitz of order 2k.
\end{lemma}
\begin{lemma}
Let t,s and k be any three positive integers. Consider a sequence (in N) of $\mathbf{x}_{1},\mathbf{x}_{2}, ...,\mathbf{x}_{s} \in \mathbb{R}^{N}$ such that $\frac{1}{\sqrt{N}}\norm{\mathbf{x}_{j}} \leqslant c_{j}$ for some constant $c_{j}$ independent of N, for $j=1,...,s$ and a sequence of order-k uniformly pseudo-Lipschitz functions $\varphi_{N}:(\mathbb{R}^{N})^{t+s} \to \mathbb{R}$. The sequence of functions $\phi_{N}(.) = \varphi_{N}(.,\mathbf{x}_{1},\mathbf{x}_{2},...,\mathbf{x}_{s})$ is also uniformly pseudo-Lipschitz of order k.
\end{lemma}
\begin{lemma}
    \label{lemma:extra_exp}
Let t be any positive integer. Consider a sequence of uniformly pseudo-Lipschitz functions $\varphi_{N} :(\mathbb{R}^{N})^{t} \to \mathbb{R}$ of order k. The sequence of functions $\Phi_{N} : (\mathbb{R}^{N})^{t}\to \mathbb{R}$ such that $\Phi_{N}(\mathbf{x}_{1},\mathbf{x}_{2},...,\mathbf{x}_{t)} = \mathbb{E}\left[\varphi_{N}(\mathbf{x}_{1},...,\mathbf{x}_{t-1},\mathbf{x}_{t}+\mathbf{\bZ})\right]$, in which $\mathbf{\bZ} \sim \mathbf{N}(0,a\mathbf{I}_{N})$ and $a\leqslant0$, is also uniformly pseudo-Lipschitz of order k.
\end{lemma}

We now state a result on Gaussian concentration of matrix-valued pseudo-Lipschitz functions. This is an extension to the matrix case (of finite width) of Lemma C.8 from \cite{berthier2020state}.

\begin{lemma}
\label{pseudo-lip-conv}
Let $\bZ \sim \mathbf{N}(0, \boldsymbol{\kappa} \otimes \mathbf{I}_{N})$ where $\boldsymbol{\kappa} \in \cS_{q}^{+}$. Let $\Phi_{N} : \mathbb{R}^{N\times q} \to \mathbb{R}$ be a sequence of random functions, independent of $\bZ$, such that $\mathbb{P}(\mathcal{E}_{N}) \to 1$ as $N \to \infty$, where $\mathcal{E}_{N}$ is the event that $\Phi_{N}$ is pseudo-Lipschitz of (deterministic) order $k$ with (deterministic) pseudo-Lipschitz constant $L$. Then $\Phi_{N}(\bZ) \stackrel{P} \simeq \mathbb{E}[\Phi_{N}(\bZ)]$. 
\end{lemma}
\begin{proof}
First, it is straightforward to see that
\begin{align}
    \Phi_{N}(\bZ) = \Phi_{N}(\tilde{\bZ}\boldsymbol{\kappa}^{1/2})=\tilde{\Phi}_{N}(\tilde{\bZ})
\end{align}
where $\tilde{\bZ} \in \mathbb{R}^{N\times q}$ is an i.i.d.~standard normal matrix, and $\tilde{\Phi}_{N} = \Phi_{N}(.\boldsymbol{\kappa}^{1/2})$. Since $\norm{\boldsymbol{\kappa}}_{op}$ is bounded for all N, $\tilde{\Phi}$ is also pseudo-Lipschitz of order k, with constant $L\max(\norm{\boldsymbol{\kappa}}_{op}^{1/2},\norm{\boldsymbol{\kappa}}_{op}^{k/2})$.
$\tilde{\Phi}_{N}$ can then be considered as a function acting on a vector of size $N \times q$ with i.i.d. standard normal components. The proof is then identical to that of Lemma C.8 from \cite{berthier2020state} with an additional finite factor $q$. We remind this proof for completeness. Under $\mathcal{E}_{N}$, using the definition of pseudo-Lipschitz functions and proposition \ref{prop:gauss_poinc}:
\begin{equation}
\label{eq:var_bound}
    \mathbb{E}_{\bZ}\left[\norm{\nabla \Phi_{N} (\bZ)}_{2}^{2}\right] \leqslant \frac{L^{2}}{Nq}\mathbb{E}_{\bZ}\left[\left(1+2\left(\frac{1}{\sqrt{Nq}}\norm{\bZ}_{2}\right)^{k-1}\right)^{2}\right]\leqslant \frac{L^{2}}{Nq}C(k)
\end{equation}
for a constant $C(k)$ that only depends on k. Then for any $\epsilon >0$, there exists a constant $c>0$, independent of $N$, such that:
\begin{align}
    \mathbb{P}\{\abs{\Phi_{N}(\bZ)-\mathbb{E}_{\bZ}[\Phi_{N}(\bZ)]}>\epsilon\} &\leqslant \mathbb{E}\{    \mathbb{P}\{\abs{\Phi_{N}(\bZ)-\mathbb{E}_{\bZ}[\Phi_{N}(\bZ)]}>\epsilon\}\mathbb{I}_{\mathcal{E}_{N}}\}+\mathbb{P}(\bar{\mathcal{E}}_{N}) \notag \\
    &\leqslant \frac{\mbox{Var}\left[\Phi_{N}(\bZ)\right]}{\epsilon^{2}}+\mathbb{P}(\bar{\mathcal{E}}_{N}) \notag\\
    & \leqslant \frac{L^{2}C(k)}{Nq\epsilon^{2}}+\mathbb{P}(\bar{\mathcal{E}}_{N})
\end{align}
where the second and third line are obtained by applying Chebyshev's inequality and proposition \ref{prop:gauss_poinc} with the variance bound evaluated at Eq.\eqref{eq:var_bound}.
\end{proof}
The next lemmas are matrix generalizations of the ones used in \cite{berthier2020state}.
\begin{lemma}
\label{conv_lemmas_app}
Consider a sequence of matrices $\mathbf{A} \sim GOE(N)$ and two sequences of non-random matrices, $\bU,\bV \in \mathbb{R}^{N \times q}$ such that the columns of $\bU$ and $\bV$ verify $\norm{\bU^{i}}_{2} = \norm{\bV^{i}}_{2} = \sqrt{N}$. Under this hypothesis, define the finite quantity $\mathbf{G} = \lim_{N \to \infty}\frac{1}{N}\bU^{\top}\bU$, the
limiting Gram matrix of the columns of $\bU$. We then have: 
\begin{enumerate}[label=\alph*)]
\item $\frac{1}{N}\bV^{\top}\mathbf{A}\bU \xrightarrow[N \to \infty]{P} 0_{q \times q}$ and $\frac{1}{N}\norm{\bV^{\top}\mathbf{A}\bU}_{F} \xrightarrow[N \to \infty]{P} 0$.
\item Let $\mathbf{P} \in \mathbb{R}^{N \times N}$ be a sequence of non-random projection matrices such that there exists a constant t that satisfies, for all N, k=rank($\mathbf{P}$)$\leqslant t$. Then $\frac{1}{N}\norm{ \mathbf{P}\mathbf{A}\bU}_{F}^{2} \xrightarrow[N \to \infty]{P} 0$.
\item There exists a sequence of random matrices $\bZ \in \mathbb{R}^{N \times q}$, such that $\frac{1}{N} \norm{\mathbf{A}\bU-\bZ}^{2}_{F} \xrightarrow[N \to \infty]{P} 0$ where $\bZ \sim \mathbf{N}(0,\mathbf{G}\otimes \mathbf{I}_{N})$.
\item $\frac{1}{N} (\mathbf{A}\bU)^{\top}\mathbf{A}\bU  \xrightarrow[N \to \infty]{P} \mathbf{G}$.
\end{enumerate}
\end{lemma}
\begin{proof} 
In this proof, the $i$-th line of a given matrix $\bZ$ is denoted $\bZ_{i}$ and its $j$-th column $\bZ^{j}$.
\begin{enumerate}[label=\alph*)]
\item For any $1 \leqslant i,j \leqslant q$, the $i$-th element of the $j$-th column verifies:
\begin{align}
    \frac{1}{N}(\bV^{\top}\mathbf{A}\bU)_{i}^{j} &= \frac{1}{N}(\bV^{i})^{\top}\mathbf{A}\bU^{j} \notag \\
    &= \frac{1}{N}(\bV^{i})^{\top}\mathbf{H}\bU^{j} + \frac{1}{N}(\bV^{i})^{\top}\mathbf{H}^{\top}\bU^{j}
\end{align}\\
where $\mathbf{H}$ is a matrix with i.i.d. $\mathbf{N}(0,\frac{1}{2N})$ elements. The random variable $\frac{1}{N}(\bV^{i})^{\top}\mathbf{H}\bU^{j}$ is centered Gaussian with variance 
\begin{equation}
    \frac{1}{N^{2}} \sum_{k,l=1}^{N}(\bV^{i}_{k})^{2}(\bU^{j}_{l})^{2}\frac{1}{2N} = \frac{\norm{\bV^{i}}_{2}^{2}\norm{\bU^{j}}_{2}^{2}}{2N^{3}} = \frac{1}{2N} \to 0
\end{equation}
which shows that $\frac{1}{N}(\bV^{i})^{\top}\mathbf{H}\bU^{j}$ converges in probability to zero. A similar argument shows that $\frac{1}{N}(\bV^{i})^{\top}\mathbf{H}^{\top}\bU^{j}$ also converges in probability to zero. The union bound then immediately gives that $\frac{1}{N}(\bV^{\top}\mathbf{A}\bU)_{i}^{j}  \xrightarrow[N \to \infty]{P} 0$. Thus each element of the finite size $q \times q$ matrix $\frac{1}{N}\bV^{\top}\mathbf{A}\bU$ goes to zero. Since $q$ is finite, the union bound then gives the desired result on the Frobenius norm.
\item For any $1 \leqslant i \leqslant q$:
\begin{equation}
    \frac{1}{N}(\mathbf{P}\mathbf{A}\bU)^{i} = \frac{1}{N}(\mathbf{P}\mathbf{A}\bU^{i})
\end{equation}
Now let $\mathbf{v}_{1}, ..., \mathbf{v}_{k}$ be an orthogonal basis of the image of $\mathbf{P}$, such that $\norm{\mathbf{v}_{1}} = ... = \norm{\mathbf{v}_{k}} = \sqrt{N}$, and $\bV \in \mathbb{R}^{N \times t}$ the matrix of concatenated $\bv$. Note that k can depend on N, but k is uniformly bounded by t. Then, using point (a) and the fact that $q$ and $k$ are finite for all $N$:
\begin{equation}
    \frac{1}{N}\norm{\mathbf{P}\mathbf{A}\bU}_{F}^{2} = \frac{1}{N}\norm{\bV^{\top}\bA\bU} \xrightarrow[N\to \infty]{P} 0
\end{equation}
This proves point $(b)$.
\item The matrix $\mathbf{A}\bU$ is a $\mathbb{R}^{N \times q}$ correlated Gaussian matrix. For any two columns $\bU^{l},\bU^{m}$, the vector  $(\mathbf{A}\bU^{l},\mathbf{A}\bU^{m})$ is a Gaussian vector with zero mean, whose covariance matrix has elements:
\begin{align}
    \mathbb{E}\left[\left(\mathbf{A}\bU^{l}\right)\left(\mathbf{A}\bU^{m}\right)^{\top}\right]_{i}^{j} &= \mathbb{E}\left[\left(\mathbf{A}\bU^{l}\right)_{i}\left(\mathbf{A}\bU^{m}\right)_{j}\right] \notag\\
    &=\mathbb{E}\left[\sum_{k=1}^{N}\bA_{i}^{k}\bU^{l}_{k}\sum_{k'=1}^{N}\bA_{j}^{k'}\bU^{m}_{k'}\right] \notag\\
    &= \mathbb{E}\bigg[\sum_{k,k'}\bH_{i}^{k}\bH_{j}^{k'}\bU^{l}_{k}\bU^{m}_{k'}+\bH_{i}^{k}\bH_{k'}^{j}\bU^{l}_{k}\bU^{m}_{k'}+\bH_{k}^{i}\bH_{j}^{k'}\bU^{l}_{k}\bU^{m}_{k'}+\bH_{k}^{i}\bH_{k'}^{j}\bU^{l}_{k}\bU^{m}_{k'}\bigg] \notag \\
    &=\frac{1}{N}\left(\delta_{ij}\sum_{k}\bU^{l}_{k}\bU^{m}_{k}+\bU^{l}_{i}\bU^{m}_{j}\right)
\end{align}
    which gives the block
    \begin{align}
    &\mathbb{E}\left[\left(\mathbf{A}\bU^{l}\right)\left(\mathbf{A}\bU^{m}\right)^{\top}\right] = \frac{1}{N}(\bU^{l})^{\top}\bU^{m}\mathbf{I}_{N}+\frac{1}{N}\bU^{l}(\bU^{m})^{\top}
    \end{align}
    and the covariance matrix
    \begin{align}
    \Sigma &= \begin{bmatrix}\mathbf{I}_{N}+\frac{1}{N}\bU^{l}(\bU^{l})^{\top} & \frac{(\bU^{l})^{\top}\bU^{m}}{N}\mathbf{I}_{N}+\frac{1}{N}\bU^{l}(\bU^{m})^{\top} \\ \frac{(\bU^{l})^{\top}\bU^{m}}{N}\mathbf{I}_{N}+\frac{1}{N}\bU^{m}(\bU^{l})^{\top} & \mathbf{I}_{N}+\frac{1}{N}\bU^{m}(\bU^{m})^{\top} \end{bmatrix} 
\end{align}
and in turn the following covariance matrix for the joint law of the $q$ vectors $\mathbf{A}\bU^{1}, ..., \mathbf{A}\bU^{q}$. 
\begin{tiny}
\begin{align*} \begin{bmatrix}\mathbf{I}_{N}+\frac{1}{N}\bU^{1}(\bU^{1})^{\top} & ...& ... & ... & \frac{(\bU^{1})^{\top}\bU^{q}}{n}\mathbf{I}_{N}+\frac{1}{N}\bU^{1}(\bU^{q})^{\top} \\
    ... &... &...&...&... \\
    ... &\frac{(\bU^{i-1})^{\top}\bU^{i}}{N}\mathbf{I}_{N}+\frac{1}{N}\bU^{i}(\bU^{i-1})^{\top}& \mathbf{I}_{N}+\frac{1}{N}\bU^{i}(\bU^{i})^{\top} &\frac{(\bU^{i})^{\top}\bU^{i+1}}{N}\mathbf{I}_{N}+\frac{1}{N}\bU^{i}(\bU^{i+1})^{\top}& ...
    \\
    ... &... &...&...&... \\
    \frac{(\bU^{q})^{\top}\bU^{1}}{N}\mathbf{I}_{N}+\frac{1}{N}\bU^{q}(\bU^{1})^{\top} & ... &... &...  &\mathbf{I}_{N}+\frac{1}{N}\bU^{q}(\bU^{q})^{\top} \end{bmatrix} \\
\end{align*}
\end{tiny}
which can be rewritten
\begin{align}
    \Sigma&=\frac{1}{N}\bU^{\top}\bU\otimes\mathbf{I}_{N}+\frac{1}{N}\begin{bmatrix}\bU^{1}(\bU^{1})^{\top} & ...& ... & ... & \bU^{1}(\bU^{q})^{\top} \\
    ... &... &...&...&... \\
    ... &\bU^{i}(\bU^{i-1})^{\top}& \bU^{i}(\bU^{i})^{\top} &\bU^{i}(\bU^{i+1})^{\top}& ...
    \\
    ... &... &...&...&... \\
    \bU^{q}(\bU^{1})^{\top} & ... &... &...  &\bU^{q}(\bU^{q})^{\top} \end{bmatrix} \notag\\
    &= \frac{1}{N}\bU^{\top}\bU\otimes\mathbf{I}_{N}+\frac{1}{N}\tilde{\bU}\tilde{\bU}^{\top}
\end{align}
where $\tilde{\bU} \in \mathbb{R}^{Nq}$ is the vector of vertically concatenated columns of $\bU$.
Now consider two independent  $\mathbf{N}(0,\mathbf{I}_{Nq})$ vectors $\tilde{\bZ}^{1},\tilde{\bZ}^{2}$ and $\tilde{\bV}\in \mathbb{R}^{Nq}$ the vector of vertically concatenated columns of $\bA\bU$. We can write that the quantity:
\begin{equation}
    \frac{\norm{\tilde{\bV}-\left(\frac{1}{N}\bU^{\top}\bU\otimes\mathbf{I}_{N}\right)^{1/2}\tilde{\bZ}^{1}}_{2}}{\sqrt{N}}
\end{equation}
is distributed as
\begin{align}
\frac{\norm{(\frac{1}{N}\bU^{\top}\bU\otimes\mathbf{I}_{N})^{1/2}\tilde{\bZ}^{1}+(\frac{1}{N}\tilde{\bU}\tilde{\bU}^{\top})^{1/2}\tilde{\bZ}^{2}-(\frac{1}{N}\bU^{\top}\bU\otimes\mathbf{I}_{N})^{1/2}\tilde{\bZ}^{1}}_{2}}{\sqrt{N}}&= \frac{1}{N\sqrt{N}}\norm{\tilde{\bU}\tilde{\bU}^{\top}\tilde{\bZ}^{2}}_{2} \notag\\
    &= \frac{\sqrt{q}}{N}\abs{\tilde{\bU}^{\top}\tilde{\bZ}^{2}} \xrightarrow[N \to \infty]{P} 0
\end{align}
where the last convergence follows from the fact that $\frac{1}{N}\tilde{\bU}^{\top}\tilde{\bZ}^{2}$ is a centered Gaussian random variable with variance $\norm{\tilde{\bU}}_{2}^{2}/N^{2} = q/N$, where $q$ is kept finite. This concludes the proof of point (c).
\item The function $\Phi : \mathbb{R}^{N\times q} \to \mathbb{R},\bX \to \frac{1}{N}\bX^{\top}\bX$ is pseudo-Lipschitz of order 2. A straightforward calculation shows that, for any $\bZ \sim \mathbf{N}(0,\mathbf{G}\otimes\mathbf{I}_{N})$, we have $\mathbb{E}[\phi(\bZ)] = \mathbf{G}$. Then :
\begin{align}
\mathbb{P}\left(\norm{\Phi(\bA\bU)-\mathbb{E}[\Phi(\bZ)]}_{F} \geqslant \epsilon\right) \leqslant \mathbb{P}\left(\norm{\Phi(A\bU)-\Phi(\bZ)}_{F} \geqslant \epsilon\right)+\mathbb{P}\left(\norm{\Phi(\bZ)-\mathbb{E}[\Phi(\bZ)]}_{F} \geqslant \epsilon\right)
\end{align}

the second term on the right-hand side vanishes as $N \to \infty$ using the Gaussian concentration of matrix-valued pseudo-Lipschitz functions Lemma \ref{pseudo-lip-conv}, and the first term vanishes using the definition of pseudo-Lipschitz function and the statement (c) proven above. This concludes the proof of statement (d).
\end{enumerate}
\end{proof}

\end{document}